\documentclass[11pt,psfig,a4]{article}
\makeatletter
\def\@seccntformat#1{\@ifundefined{#1@cntformat}%
   {\csname the#1\endcsname\quad}  % default
   {\csname #1@cntformat\endcsname}% enable individual control
}
\let\oldappendix\appendix %% save current definition of \appendix
\renewcommand\appendix{%
    \oldappendix
    \newcommand{\section@cntformat}{\appendixname~\thesection:\,\,}
}
\makeatother

\usepackage{url}

\usepackage[T1]{fontenc}
\usepackage{amsfonts}
\usepackage{geometry}
\usepackage{graphics}
\usepackage{lscape}
\usepackage{subfigure}
\usepackage{graphicx}
\usepackage{setspace}
\usepackage{amsthm}
\usepackage{amssymb}
\usepackage{amsmath}
\usepackage{color,soul}
\usepackage{multirow}
\usepackage{hhline}
\usepackage[table]{xcolor}
\usepackage{booktabs}
\usepackage{rotating}
\usepackage[skip=30pt]{caption}

\usepackage{epsfig}
\usepackage{dsfont}
\usepackage{imakeidx}
\usepackage{enumitem}
\usepackage{dsfont}
\usepackage{stmaryrd}

\usepackage[ruled,lined,boxed]{algorithm2e}
\usepackage{epstopdf}
\usepackage{graphics}

\usepackage{comment}

\usepackage[round]{natbib}
\PassOptionsToPackage{
        natbib=true,
        style=authoryear-comp,
        hyperref=true,
        backend=biber,
        maxbibnames=99,
        firstinits=true,
        uniquename=init,
        maxcitenames=2,
        parentracker=true,
        url=false,
        doi=false,
        isbn=false,
        eprint=false,
        backref=true,
            }   {biblatex}

\newcommand{\Lyx}{L\kern-.1667em\lower.25em\hbox{y}\kern-.125emX\spacefactor1000}
\definecolor{armygreen}{rgb}{0.19, 0.53, 0.43}
\definecolor{atomictangerine}{rgb}{1.0, 0.6, 0.4}

\usepackage{ulem}

\newtheorem{remark}{Remark}
\newtheorem{theorem}{Theorem}
\newtheorem{definition}[theorem]{Definition}
\newtheorem{lemma}[theorem]{Lemma}
\newtheorem{prop}{Proposition}

\DeclareMathOperator*{\argmin}{arg\,min}

\begin{document}

\title{\textbf{
New approximate stochastic dominance approaches for Enhanced Indexation models
}$^{\star}$}
\author{Francesco Cesarone$^1$, Justo Puerto$^2$ \\
{\small $^1$\normalem\emph{ Roma Tre University - Department of Business Economics}}\\
{\footnotesize francesco.cesarone@uniroma3.it}\\
{\small $^2$\normalem\emph{Universidad de Sevilla - Instituto Universitario de Investigaci\'on  Matem\'{a}tica (IMUS)}}\\
{\footnotesize puerto@us.es}\\
}

%\date{September 18, 2018}
\date{\today}
\maketitle
{$^{\star}$ This paper is dedicated to the memory of Mois\'{e}s Rodr\'{\i}guez-Madrena}

\begin{abstract}

In this paper, we discuss portfolio selection strategies for Enhanced Indexation (EI), which are based on stochastic dominance relations. The goal is to select portfolios that stochastically dominate a given benchmark but that, at the same time, must generate some excess return with respect to a benchmark index. To achieve this goal, we propose a new methodology that selects portfolios using the ordered weighted average (OWA) operator, which generalizes previous approaches based on minimax selection rules and still leads to solving linear programming models.
We also introduce a new type of approximate stochastic dominance rule and show that it implies the almost Second-order Stochastic Dominance (SSD) criterion proposed by \cite{lizyayev2012tractable}.
We prove that our EI model based on OWA selects portfolios that dominate a given benchmark through this new form of stochastic dominance criterion.
We test the performance of the obtained portfolios in an extensive empirical analysis based on real-world datasets.
The computational results show that our proposed approach outperforms several SSD-based strategies widely used in the literature, as well as the global minimum variance portfolio.
%}
\medskip


%\medskip
%\noindent
%\textbf{JEL classification}: .
\end{abstract}

\section{Introduction}\label{sec:Intro}

It is well-known among academics and practitioners that active portfolio management is unlikely to generate extra returns over their benchmarks \citep{bufalo2022straightening}.
For this reason, passive asset management has attracted great interest in recent decades, offering returns close to benchmarks at a lower cost.
Index tracking (IT) is a popular problem for dealing with the construction of passive management portfolios and has received large attention in the literature.
IT basically consists of a constrained optimization problem, in which one minimizes the distance between a benchmark and the tracking portfolio by typically using a fixed number of assets less than those available in the investment universe.
Comprehensive surveys on this problem can be found in
\cite{Bea03}, \cite{canakgoz2009mixed}, \cite{guastaroba2012kernel}, \cite{bruni2012new,bruni2013no,bruni2015linear,bruni2017exact},
\cite{scozzari2013exact},
\cite{cesarone2019risk},
and, more recently, in \cite{shu2020high}, \cite{sant2022risk}.

A step towards a more sophisticated strategy, which has attracted increasing attention over the past decade, is the Enhanced Indexation (EI) problem, where the aim is not only to track a market index but simultaneously to generate an excess return over the benchmark.
An extensive review of the literature on EI is provided in \cite{canakgoz2009mixed}, \cite{guastaroba2016linear}, and more recently in \cite{guastaroba2020enhanced},  \cite{beraldi2020enhanced}.

This paper focuses on portfolio selection strategies for EI based on Stochastic Dominance (SD) relations. Indeed, one relatively recent and promising approach for Enhanced Indexation consists of selecting portfolios that stochastically dominate a given benchmark.
SD provides a (partial) order in the space of random returns, has connections with Expected Utility Theory, and takes into account the full information regarding the distribution of the portfolio random return \citep[see][]{whitmore1978stochastic,levy1992stochastic,shaked1994stochastic,muller2002comparison,levy2015stochastic}.
\cite{roman2013enhanced} develop a Second-order Stochastic Dominance (SSD) approach, which consists of a multi-objective optimization problem to construct portfolios whose returns dominate those of a benchmark \citep[see also][]{roman2006portfolio,fabian2011processing}.
Among the infinite Pareto-optimal solutions, 
% which are SSD-efficient (or at least approximatively SSD-efficient), 
the authors scalarize their multi-objective model as a Minimax problem, which is linear and is solved efficiently by cutting plane techniques.
Several SSD-based approaches are tested in \cite{hodder2015improved}, where the authors implement the exact SSD methods developed by \cite{kuosmanen2004efficient} and \cite{kopa2015general}.

\noindent
Alternative portfolio selection approaches for EI, based on relaxing SD rules, i.e., using some forms of approximate SD, can be found in \cite{bruni2012new},
\cite{bruni2015linear},
\cite{bruni2017exact},
and \cite{sharma2017enhanced}.
In fact, as discussed by \cite{leshno2002preferred},
in some economic situations, the relaxation of SD conditions could provide advantages w.r.t. the exact ones.
 \cite{lizyayev2012tractable} propose different relaxations of SD for both First and Second-order rules, named $\varepsilon$-Almost-FSD ($\varepsilon$-AFSD) and $\varepsilon$-Almost-SSD ($\varepsilon$-ASSD), respectively.
They also describe the optimization models for
potential applications of $\varepsilon$-AFSD and $\varepsilon$-ASSD to portfolio selection, but without providing empirical tests on real datasets.
Furthermore, as discussed in the conclusions of \cite{post2017portfolio}, ``the current literature offers little guidance for the specification of ``epsilon'', or the admissible violation area.''
We point out that our approach does not require to explicitly specify ``epsilon'' since it aims to obtain solutions that satisfy the SD conditions with minimum violations.

\noindent
Further developments and generalizations on Almost Stochastic Dominance (ASD) conditions can be found in \cite{levy2010economically}, \cite{tzeng2013revisiting}, \cite{post2013general}, \cite{guo2014moment}, \cite{denuit2014almost}, and \cite{tsetlin2015generalized}.
However, there appear to be no available applications of ASD to portfolio selection.

According to our literature review, the reader may realize that there is room for improvement in the portfolio selection methods based on SD rules.
Our goal is to elaborate in this direction by proposing a new methodology and to shed some light on this field, showing some of their connections.
Our contributions can be summarized as follows.
First, we address the multi-objective representation of the EI problem proposed by \cite{roman2013enhanced} by an alternative approach based on the
Ordered Weighted Average (OWA) operator
\citep{yager1988ordered}.
Its advantage is that it allows us to consider a set of the worst realized objectives in the problem, whereas the Minimax operator, used in \cite{roman2013enhanced}, pays attention only to the maximum of them.
Furthermore, the convenience of this type of scalarization is that it allows highlighting a broader family of Pareto-optimal portfolios, improving the properties of the Minimax ones.
The application of the OWA operator for scalarizing multi-objective optimization problems has been widely studied in the literature \citep[see, e.g.,][]{fernandez2014ordered,fernandez2017ordered}.
Indeed, OWA is a flexible operator that generalizes several criteria, including the Minimax one \citep[see, e.g.,][]{ogryczak2000multiple,nickel2006location}.
As a consequence, we can provide a more general method for highlighting a novel subset of Pareto-optimal solutions of the multi-objective problem in \cite{roman2013enhanced}.
Second, we propose a new methodology to select a portfolio that dominates a given benchmark index in terms of SD.
This methodology uses the OWA representation developed in \cite{marin2020fresh} and leads to solving linear programming problems.
Furthermore, it is more robust and contains as special cases some known SD-based approaches to asset allocation \citep[see][]{lizyayev2012tractable,roman2013enhanced}.
%
% Thus, we develop an innovative and versatile optimization model for selecting an investment portfolio that should generate excess return w.r.t. a benchmark index.
%
Third, with the purpose of providing a common framework to understand and explain approximate dominance rules,
we introduce a new rule that we call Cumulative Second-order $\mathcal{E}$-Stochastic Dominance (CS$\mathcal{E}$SD), and
we show that it implies the family of $\varepsilon$-ASSD rules proposed by \cite{lizyayev2012tractable}.
Exploiting some results in multi-criteria optimization \citep{ehrgott2005multicriteria},
we prove that the optimal solution of our EI model based on OWA selects portfolios that dominate the benchmark index in terms of CS$\mathcal{E}$SD.
Fourth, to test and validate the performance of the portfolios obtained by our proposed strategies, we provide extensive empirical analysis based on the FTSE100, NASDAQ100, SP500, and Fama \& French 49 Industry portfolios datasets, comparing their out-of-sample behavior with that of the portfolios constructed by
several SD-based approaches proposed in the literature
and with that of the long-only, long-short, 1-norm, and 2-norm minimum variance portfolios.

The remainder of the paper is organized as follows.
In Section \ref{PreConc}, we recall the most commonly used exact stochastic dominance rules and present new approximate stochastic dominance criteria.
In Section \ref{sec:EI}, we introduce the multi-objective optimization model of \cite{roman2013enhanced}, and then we provide a general method for finding a class of Pareto-optimal solutions, which SSD-dominate or almost SSD-dominate a given benchmark index.
In Section \ref{sec:TheoreticalAnalysis}, we discuss some theoretical features of our proposed approach, providing some statistical and economic interpretations and some illustrative examples.
Section \ref{sec:EmpiricalAnalysis} presents the results of an extensive empirical analysis based on real-world data, comparing the out-of-sample performance of our proposed approach with other concurrent strategies.
Conclusions and further research projects are presented in Section \ref{sec:Conclusions}.

\section{Preliminary concepts on some exact and approximate SD criteria \label{PreConc}}

One of the reasons that has made the use of stochastic dominance criteria for portfolio selection particularly attractive is that these criteria are able to provide a partial order in the space of random variables by avoiding the specification of a particular utility function, which represents the preferences of an investor.
Indeed, from Expected Utility Theory \citep[EUT, see, e.g.,][]{morgenstern1953theory}, given a utility function, a decision-maker prefers a random variable to another if it provides a larger value in terms of its expected utility, i.e., it determines a greater satisfaction.
However, the specification of a utility function is a strictly subjective issue, although it is a key element of EUT.

For the sake of readability, we recall here the most commonly used exact stochastic dominance criteria.
Let $A$ and $B$ be two random variables with cumulative distribution functions
$F_{A}(\alpha) = \mathbf{Pr}(A \le \alpha)$ and $F_{B}(\alpha) = \mathbf{Pr}(B \le \alpha)$
for $\alpha \in \mathbb{R}$.
\begin{definition}[Zero-order Stochastic Dominance - ZSD] \label{def:ZSD}
%\textbf{}: \newline
$A$ is preferred to $B$ w.r.t. ZSD iff
\begin{equation}
F_{A-B} (0) = \mathbf{Pr}(A - B \le 0)= 0 
\end{equation}
\end{definition}
\noindent
The ZSD relation represents behavior of a decision maker who prefers a random variable over another only when the first gives better outcomes than the second in (almost) all states of the world.
\begin{definition}[First-order Stochastic Dominance - FSD] \label{FSDdef}

$A$ is preferred to $B$ w.r.t. FSD iff
\begin{equation}
F_A(\alpha) \le F_B(\alpha) \ \ \ \ \forall \alpha \in \mathbb{R}.
\end{equation}
\end{definition}
\begin{definition}[Second-order Stochastic Dominance - SSD] \label{SSDdef}

$A$ is preferred to $B$ w.r.t. SSD iff
\begin{equation}
\int_{-\infty}^{\alpha} F_A(\tau)d\tau \ \le \int_{-\infty}^{\alpha}
F_B(\tau)d\tau \ \ \ \ \forall \alpha \in \mathbb{R}.
\end{equation}
\end{definition}
\noindent
These two last definitions require that each inequality is strict in at least one case.

\noindent
As discussed, e.g., in \cite{levy1992stochastic}, these relations can be linked to EUT in terms of different classes of utility functions.
Indeed, $A$ is preferred to $B$ w.r.t. FSD if and only if $E[u(A)] \geq E[u(B)]$ for all non-decreasing utility functions $u$;
$A$ is preferred to $B$ w.r.t. SSD if and only if the same holds for all non-decreasing and concave utility functions.
The SSD relation is less demanding than that of ZSD and of FSD.
More precisely, given an SD criterion of order $v$, when increasing the order,
the $v+1$-SD condition becomes less restrictive.
Furthermore, $v$-SD implies $(v+1)$-SD,
while the opposite is not necessarily true
\citep[see, e.g.,][]{levy2015stochastic}.
In this context, \cite{post2017portfolio} remark that Third-order Stochastic Dominance (TSD), although very appealing from a theoretical perspective, is highly demanding to test. 
Therefore, these authors propose a method of selecting investment portfolios that dominate a given benchmark in terms of TSD, which is less demanding.
More precisely, their strategy is based on a refinement of the ``SuperConvex'' TSD (SCTSD) approach of \cite{bawa1985determination} that leads to a portfolio selection model formulated as a convex quadratic constrained programming. 
As stated by the authors, SCTSD can be seen as an approximation of TSD although in any case results in a difficult quadratic program. 
For the reasons mentioned above, we have restricted ourselves from considering TSD in our experiments. We want to remain within the framework of LP models, avoiding the more complex quadratic or even superconvex mathematical programs.
We also observe that fractional-order SD rules extending the classical integer-order SD have recently been proposed in the literature \citep[see][]{muller2017between,huang2020fractional}.
However, there currently are no portfolio selection applications in this framework.

\noindent
In this paper, we essentially focus on the SSD condition and its relaxations, characterizing our approach as an LP.
An example of an approximate (relaxed) version of the SSD criterion is given in
\cite{lizyayev2012tractable}.

\begin{definition}[Almost Second-order Stochastic Dominance - ASSD] \label{ASSDdef}
Given a tolerance $\vartheta>0,\,  A$ is preferred to $B$ w.r.t. ASSD if
$$
\int_{-\infty}^{\alpha}\left(F_{A}(\tau)-F_{B}(\tau)\right) d \tau \leq \vartheta \quad \forall \alpha \in\left[\alpha^{\prime}, \alpha^{\prime \prime}\right],
$$
where $\left[\alpha^{\prime}, \alpha^{\prime \prime}\right]$ is the combined range of outcomes of $A$ and $B$.
\end{definition}
\noindent
Another example of relaxation of SD conditions can be found in \cite{bruni2017exact}, where the authors introduce
the Zero-order $\varepsilon$-Stochastic Dominance ($\mathrm{Z} \varepsilon \mathrm{SD}$) relation,
which is a relaxation (approximation) of the standard ZSD (see Definition \ref{def:ZSD}).
Although in their work only the ZSD relaxation is addressed,
this idea could also be extended to higher-order stochastic dominance rules.
\begin{remark}[$\mathrm{Z} \varepsilon \mathrm{SD}$ vs. ASSD]\label{rem:ASSD_1}
  We observe that the extension of the $\mathrm{Z} \varepsilon \mathrm{SD}$ approach to SSD, i.e., the Second-order $\varepsilon$-Stochastic Dominance ($\mathrm{S} \varepsilon \mathrm{SD}$), is equivalent to ASSD, by simply considering $\varepsilon = \vartheta$.
\end{remark}

\noindent
Based on Definitions \ref{SSDdef} and \ref{ASSDdef}, the SSD and ASSD relations can be expressed by the following equivalent criteria.
\begin{itemize}
  \item $E[u(A)] - E[u(B)] \geq -\tau$ holds for all non-decreasing and concave utility functions $u$;
  \item $E[[\eta - A]_{+}] - E[[\eta - B]_{+}] \leq \tau$ holds for all $\eta \in \mathbb{R}$;
  \item $Tail_{\beta}(A) - Tail_{\beta}(B) \geq -\tau$ holds for all $\beta \in (0, 1]$,
  where $Tail_{\beta}(A)$ is the unconditional expectation of the worst $\beta$100\% outcomes of $A$.
\end{itemize}
Clearly, these conditions are equivalent to SSD for $\tau = 0$ and to ASSD for $\tau>0$.

\medskip

\noindent
In the remaining part of the paper, we assume that the asset returns are discrete random variables defined on a discrete state space,
and that there are $T$ states of nature, each with probability $\pi_t$ with $t=1,...,T$.
We use a look-back approach, where the realizations of the discrete random returns correspond to the historical scenarios, and the investment decision is performed on an in-sample window of $T$ historical outcomes, each with probability $\pi_t=1/T$, as is typically done in portfolio optimization \citep[see, e.g.,][and references therein]{carleo2017approximating,cesarone2020computational,bellini2021risk}.
\begin{remark}[Equally likely states and Shannon entropy]\label{rem:Shannon}
We observe that given a discrete random variable with probability distribution $\pi_t$ (where $t=1,...,T$) the associated Shannon entropy value $S(\pi)=-\sum_{t=1}^{T} \pi_t \log_2 (\pi_t)$ is a measure for the degree of uncertainty or ignorance about the probability distribution.
Its value is zero when the probability of an event is 1 (i.e., $\pi_t=1$ for a $t=1,...,T$) or a certainty, while it is maximum when all the states of nature are equally likely ($\pi_t=1/T$), namely the case of maximum uncertainty \citep[see, e.g.,][]{myung1996maximum}.
As for instance in \cite{post2017portfolio}, our experiments consider the case of maximum ignorance about the probability distribution $\pi_t$.
\end{remark}

\noindent
We denote by $R^{I}$ the return of a given benchmark index and by $R$ the return of a portfolio of assets.
Furthermore, we indicate by $(R_{(j)})_{1\le j\le T}$ the ordered realizations such that $R_{(1)} \leq \ldots \leq R_{(T)}$, and, reversing the order, by $(R^{(j)})_{1\le j\le T}$ the ordered realizations such that $R^{(1)} \geq \ldots \geq R^{(T)}$. 
Let us assume w.l.o.g. that the scenarios of the portfolio return $R_t$ and the index $R^{I}_t$ occur with probabilities $\pi_t$ for $t=1,...,T$.

\noindent
In a discrete world, Definitions \ref{FSDdef} and \ref{SSDdef} can be then expressed as follows
\begin{itemize}
\item $R$ FSD-dominates $R^{I}$ iff {$\displaystyle \sum_{R_t\le \alpha} \pi_t \le \sum_{R^I_{t}\le \alpha} \pi_{t}, \; \forall \alpha \in \mathbb{R}.$}
\item $R$ SSD-dominates $R^{I}$ iff 
$  \eta \sum_{R_t\le \eta} \pi_t - \sum_{R_t\le \eta} R_t \pi_t\le \eta \sum_{R_t^I\le \eta} \pi_t - \sum_{R_t^I\le \eta} R_t^I \pi_t,\  \forall \eta\in \mathbb{R}.
$
\end{itemize}

\noindent
Note that for a given threshold $\beta=\sum_{R_t\le \eta}\pi_t $, we have that  $Tail_{\beta}(R)=\displaystyle \sum_{R_t\le \eta} \pi_t R_{t}$. 
Then, since

\begin{equation}\label{eq:CVaRexpre_2}
CVaR_{\beta}(R)=-  \frac{1}{\beta} \sum_{R_t\le \eta} \pi_t R_{t},
\end{equation}
we have that
$Tail_{\beta}(R)=-  \displaystyle  \beta CVaR_{\beta}(R)$.
Thus, requiring $Tail_{\beta}(R) \geq Tail_{\beta}(R_{I})$ is equivalent to
consider $CVaR_{\beta}(R) \leq CVaR_{\beta}(R_{I})$ for all $\beta\in (0,1]$.

\noindent
In the particularly interesting case, that is usual in portfolio optimization (see also Remark \ref{rem:Shannon}), where $\pi_t=1/T$ for all $t$, the above criteria simplify resulting in:

\begin{itemize}
  \item $R$ FSD-dominates $R_{I}$ iff $R_{(j)}  \geq R^{I}_{(j)}$ for all $j=1, \ldots, T$.
  \item $R$ SSD-dominates $R_{I}$ iff $\sum_{t=1}^{j} R_{(t)} \geq \sum_{t=1}^{j} R^{I}_{(t)}$ for all $j=1, \ldots, T$.
\end{itemize}
Also in this case, $Tail_{\frac{j}{T}}(R)=\displaystyle \frac{1}{T} \sum_{t=1}^{j} R_{(t)}$.
Next, since
\begin{equation}\label{eq:CVaRexpre}
CVaR_{\frac{j}{T}}(R)=- \frac{1}{j} \sum_{t=1}^{j}  R_{(t)},	
\end{equation}
we have that
$Tail_{\frac{j}{T}}(R)=- \displaystyle \frac{j}{T} CVaR_{\frac{j}{T}}(R)$.
Thus, requiring $Tail_{\frac{j}{T}}(R) \geq Tail_{\frac{j}{T}}(R_{I})$ is equivalent to
consider $CVaR_{\frac{j}{T}}(R) \leq CVaR_{\frac{j}{T}}(R_{I})$ for all $j=1, \ldots, T$.

When discrete random variables are considered,
we can introduce a \normalem \emph{cumulative} version of $\mathrm{S} \varepsilon \mathrm{SD}$, named here $\mathrm{CS}\varepsilon \mathrm{SD}$,
similar to the $\mathrm{CZ} \varepsilon \mathrm{SD}$ relation developed in \cite{bruni2017exact}.
This last approximate stochastic dominance relation was thought as a criterion for bounding the cumulative losses for Enhanced Indexation purposes
and, among other properties, \cite{bruni2017exact} show that $\mathrm{CZ} \varepsilon \mathrm{SD}$ implies Z$\varepsilon$SD.
Therefore, similarly to how $\mathrm{Z} \varepsilon \mathrm{SD}$ can be extended to $\mathrm{S}\varepsilon \mathrm{SD}$, $\mathrm{CZ}\varepsilon\mathrm{SD}$ can also be generalized to SSD.
This generalization leads to the following definition and the fact that this new criterion implies $\mathrm{S} \varepsilon \mathrm{SD}$, or, equivalently, ASSD (see Section \ref{sec:TheoreticalAnalysis}, Proposition \ref{prop:Thoretical}, point 1).
\begin{definition}[Cumulative Second-order $\varepsilon$-Stochastic Dominance - $C S \varepsilon S D$] \label{def:CSeSD_1}
Given a tolerance $\varepsilon>0$, $A$ is preferred to $B$ w.r.t. $C S \varepsilon S D$ if
\begin{equation}\label{def:CSeSD}
  \sum_{t \in \mathcal{S}}CVaR_{\frac{t}{T}}(A) - \sum_{t \in \mathcal{S}} CVaR_{\frac{t}{T}}(B)\leq \varepsilon \quad \forall \mathcal{S} \subseteq \mathcal{T},
\end{equation}
where $\mathcal{T}=\{ 1, \ldots, T \}$.
\end{definition}
\noindent
Definition \ref{def:CSeSD_1} requires to check an exponential number of conditions to verify $C S \varepsilon S D$.
However, the next result shows that the condition can be verified in polynomial time.

\noindent
For this purpose, we need to introduce some additional notation. 
Let us denote by $\sigma$ the permutation that provides the sorting of $\left( CVaR_{\frac{t}{T}}(A) - CVaR_{\frac{t}{T}}(B)\right)_{1\le t\le T}$ as $CVaR_{\frac{\sigma(1)}{T}}(A) - CVaR_{\frac{\sigma(1)}{T}}(B)\ge \ldots \ge CVaR_{\frac{\sigma(T)}{T}}(A) - CVaR_{\frac{\sigma(T)}{T}}(B)$. 
Next, for the ease of presentation, and whenever the framework is clear, we denote
$$
\sum_{j=1}^t \left(CVaR_{\frac{\cdot}{T}}(A) - CVaR_{\frac{\cdot}{T}}(B)\right)^{(j)}:=\sum_{j=1}^t CVaR_{\frac{\sigma(j)}{T}}(A) - CVaR_{\frac{\sigma(j)}{T}}(B)
$$
\begin{prop}
\label{prop:CSeSD_form2}
Given a tolerance $\varepsilon>0, A$ is preferred to $B$ w.r.t. $C S \varepsilon S D$ iff
$$
\sum_{j=1}^t \left(CVaR_{\frac{\cdot}{T}}(A) - CVaR_{\frac{\cdot}{T}}(B)\right)^{(j)} \leq \varepsilon \quad \forall t \in  \mathcal{T}.
$$
\end{prop}
\begin{proof}
We observe that Conditions \eqref{def:CSeSD} are equivalent to require
$$
\max_{\mathcal{S} \subseteq \mathcal{T} : |\mathcal{S}|=t} \left\lbrace\sum_{t \in \mathcal{S}}CVaR_{\frac{t}{T}}(A) - \sum_{t \in \mathcal{S}} CVaR_{\frac{t}{T}}(B)\right\rbrace \leq \varepsilon  \quad \forall t\in  \mathcal{T}.
$$
Furthermore, the maximum in the left-hand-side is attained  on the $t$ largest values, so that
$$
\max_{\mathcal{S} \subseteq \mathcal{T} : |\mathcal{S}|=t} \left\lbrace\sum_{t \in \mathcal{S}}CVaR_{\frac{t}{T}}(A) - \sum_{t \in \mathcal{S}} CVaR_{\frac{t}{T}}(B)\right\rbrace = \sum_{j=1}^t \left(CVaR_{\frac{\cdot}{T}}(A) - CVaR_{\frac{\cdot}{T}}(B)\right)^{(j)}.
$$
This proves the result.
\end{proof}

\noindent
We point out that in Proposition \ref{prop:CSeSD_form2} we have a single threshold $\varepsilon$, $\forall t$.
However, it can be generalized so that we can impose different tolerances $\varepsilon_t$ for each $t \in  \mathcal{T}$, as described in the following definition.
\begin{definition}[Cumulative Second-order $\mathcal{E}$-Stochastic Dominance - $C S \mathcal{E} S D$]\label{def:CSESD}
Given a tolerance vector $\mathcal{E} =(\varepsilon_1,\ldots,\varepsilon_T)\in \mathbb{R}^T_{>0}$, $A$ is preferred to $B$ w.r.t. $C S \mathcal{E} S D$ if
$$
\sum_{j=1}^t \left(CVaR_{\frac{\cdot}{T}}(A) - CVaR_{\frac{\cdot}{T}}(B)\right)^{(j)} \leq \varepsilon_t \quad \forall t \in  \mathcal{T}.
$$
\end{definition}

In the next section, we introduce the Enhanced Indexation model of \cite{roman2013enhanced}, which properly generalized is closely connected to $C S \mathcal{E} S D$, as shown in Section \ref{OWA-CESD}.

\section{Enhanced Indexation based on stochastic dominance \label{sec:EI}}

Here, we first report the multi-objective approach of \cite{roman2013enhanced} for finding portfolios that SSD-dominate or almost SSD-dominate a given benchmark index.
Then, we provide a general method for finding a novel subset of Pareto-optimal solutions of such a multi-objective model.
More precisely, we generalize the existing minimax selection rule for the multi-objective approach of \cite{roman2013enhanced}, which, as shown in Section \ref{sec:Minimax}, is strictly linked to the ASSD conditions of \cite{lizyayev2012tractable} (see Definition \ref{ASSDdef}, Remarks \ref{rem:ASSD_1} and \ref{rem:MinimaxASSD}).
Thus still keeping the model as an LP, in Section \ref{sec:OWAsca} we propose more flexible choices of Pareto-optimal solutions of such a multi-objective model, which, as discussed in Section \ref{sec:TheoreticalAnalysis}, are related to the $C S \mathcal{E} S D$ rules and as we shall see, imply those of \cite{lizyayev2012tractable}.

We consider a setting where $n$ assets are available in the investment universe with rates of return described by
the random variables $R_1, \ldots, R_n$.
A portfolio is identified with a vector $x = (x_1 ,\ldots , x_n) \in \mathbb{R}_{+}^{n}$, where $x_k$
denotes the fraction of capital invested in the $k^{th}$ asset.
Hence, the (linear) portfolio return is given by
$R(x)=\sum_{k=1}^{n} x_k R_k$.

\noindent
As mentioned above, \cite{roman2013enhanced} propose the following multi-objective model whose Pareto-optimal solutions SSD-dominate or almost SSD-dominate a given benchmark index \citep[see also][]{roman2006portfolio,fabian2011enhanced,fabian2011processing,valle2017novel,cesarone2022comparing}
\begin{equation}
	\begin{array}{rl}
	\displaystyle	\min_{x} & \left[ (CVaR_{\frac{1}{T}}(R(x))-CVaR_{\frac{1}{T}}(R_{I})),\ldots,
 (CVaR_{\frac{T}{T}}(R(x))-CVaR_{\frac{T}{T}}(R_{I}))  \right]
		\\
		\mbox{s.t.} & x \in C
	\end{array}
	\label{eq:SSDCVaR}
\end{equation}
where $C$ is a polyhedron.
Alternatively, they find the SSD efficient portfolios by solving a slightly different
multi-objective optimization problem based on centered Tails
\begin{equation}
	\begin{array}{rl}
	\displaystyle	\max_{x} & \left[ (Tail_{\frac{1}{T}}(R(x))-Tail_{\frac{1}{T}}(R_{I})),\ldots,
 (Tail_{\frac{T}{T}}(R(x))-Tail_{\frac{T}{T}}(R_{I}))  \right]
		\\
		\mbox{s.t.} & x \in C
	\end{array}
	\label{eq:SSDTail}
\end{equation}
In the following, we will focus on Problem \eqref{eq:SSDCVaR} and then briefly discuss how the results obtained from this problem are also valid for Problem \eqref{eq:SSDTail}.

\subsection{Minimax scalarization\label{sec:Minimax}}

Among the infinite Pareto-optimal points belonging to the efficient $T$-dimensional hypersurface generated by Problem \eqref{eq:SSDCVaR}, \cite{roman2013enhanced} consider the one obtained by solving the following Minimax problem
\begin{equation}
	\begin{array}{rl}
		\displaystyle	\min_{x} & \displaystyle\max_{1 \leq j \leq T} (CVaR_{\frac{j}{T}}(R(x))-CVaR_{\frac{j}{T}}(R_{I}))
		\\
		\mbox{s.t.} & x \in C
	\end{array}
	\label{eq:SSDCVaRMinMax}
\end{equation}
This model can be expressed as an LP problem,
using the CVaR reformulation of \cite{RockUrya00,RockUrya02}.
However, due to the high
number of variables and constraints (more than $T^{2}$),
Problem \eqref{eq:SSDCVaRMinMax} is solved by implementing cutting
plane techniques, as proposed by \cite{fabian2011processing}.

\noindent
We note that the above problem is equivalent to
\begin{equation}
	\begin{array}{rl}
		\displaystyle	\min_{x} & \displaystyle\left(CVaR_{\frac{\cdot}{T}}(R(x))-CVaR_{\frac{\cdot}{T}}(R_{I})\right)^{(1)}
		\\
		\mbox{s.t.} & x \in C
	\end{array}
	\label{eq:SSDCVaRMinMax_b}
\end{equation}
\begin{remark}[Minimax scalarization vs. ASSD conditions]\label{rem:MinimaxASSD}
  From Problem \eqref{eq:SSDCVaRMinMax_b}, one can directly observe the link between Model \eqref{eq:SSDCVaRMinMax} in the paper by \cite{roman2013enhanced}, $S S D$ (Definition \ref{SSDdef}) and ASSD  (Definition \ref{ASSDdef}) conditions.
Indeed, denoting by $x^{\star}$ the optimal solution of Problem \eqref{eq:SSDCVaRMinMax_b}, if the corresponding optimal value function
$\left(CVaR_{\frac{\cdot}{T}}(R(x^{\star}))-CVaR_{\frac{\cdot}{T}}(R_{I})\right)^{(1)}$ is nonpositive (positive), then the corresponding optimal portfolio SSD-dominates (ASSD-dominates) the benchmark index.
\end{remark}

\subsection{Ordered weighted average scalarization\label{sec:OWAsca}}

To address the multi-objective optimization problems \eqref{eq:SSDCVaR} and \eqref{eq:SSDTail}, we propose here an approach based on the Ordered Weighted Average (OWA) operator.
As shown below, using the OWA scalarization, we are able to provide a more general methodology to highlight a new subset of Pareto-optimal solutions of such multi-objective problems based on centered CVaRs (see Section \ref{sec:MaxOWACVaRs}) and centered Tails (see Section \ref{sec:MaxOWATails}).

\subsubsection{Maximizing OWA of the centered CVaRs} \label{sec:MaxOWACVaRs}

An alternative scalarization of Problem \eqref{eq:SSDCVaR} consists of the following simple idea:
rather than minimizing the maximum centered CVaR, namely Model \eqref{eq:SSDCVaRMinMax},
one can investigate the performance of the Pareto-optimal portfolios obtained by minimizing the sum of the $k$ maximal CVaRs for a fixed $k=1,\ldots,T$.
It is well-known that this form of scalarization always results in Pareto-optimal solutions since the ordered median function is an isotone function provided than the lambda weights are non-negative \citep[see Lemma 1.1 in][]{nickel2006location}.
Obviously, when $k=1$, one recovers the Minimax approach of \cite{roman2013enhanced}.
More specifically, we deal with Problem \eqref{eq:SSDCVaR} by scalarizing the $T$ goals in the following way.
To simplify the presentation, we introduce some notation.
Let us denote by
$g_t(x):=CVaR_{t/T}(R(x))-CVaR_{t/T} (R_I)$ with $t \in \mathcal{T}$, and following our notation on Section \ref{PreConc}, by $(g^{(j)}(x))_{1\le t\le T}$ the elements of a vector such that $g^{(1)}(x)\ge g^{(2)}(x) \ge \ldots \ge g^{(T)}(x)$.
We then consider the following problem
\begin{equation}
	\begin{array}{rl}
	\displaystyle	\min_{x} & \displaystyle \sum_{j=1}^T \lambda^j g^{(j)}(x)
		\\
		\mbox{s.t.} & x \in C
	\end{array}
	\label{eq:ProbPropCVaR1_}
\end{equation}
where $(\lambda^j)_{1\le j\le T}$ is a vector such that $\lambda^1\ge \lambda^2 \ge \ldots \ge \lambda^T\ge 0$.
We can observe that our approach is less conservative (i.e., risk-averse) than that of \cite{roman2013enhanced}, since Problem \eqref{eq:ProbPropCVaR1_} takes into account all differences between CVaRs at different confidence levels and not just the worst one.
However, our approach still models conservative behavior, since larger weights are assigned to the worst differences.

\noindent
Thus, our goal is to obtain a tractable representation for the problem that minimizes the sum of an ordered weighted average (OWA) of the centered CVaRs, $( CVaR_{\cdot/T}(R(x))-CVaR_{\cdot/T} (R_I))_{1\le j\le T}^{(j)}$, with weights $\lambda_1\ge \lambda_2 \ge \ldots \ge \lambda_T \ge 0$.
Note that if $\lambda=(1,0,\ldots,0)$ we obtain the Minimax portfolio \citep[as in][]{roman2013enhanced}, whereas if $\lambda=(1,1,\stackrel{(k)}{\ldots},1,0,\ldots,0)$ we get the portfolio minimizing the sum of the $k$-worse centered CVaRs.

\noindent
To achieve a more tractable formulation of Problem \eqref{eq:ProbPropCVaR1_}, we denote by $f_t(x) = - g_t(x)=CVaR_{t/T} (R_I)- CVaR_{t/T}(R(x))$  $\forall t \in \mathcal{T}$, and by $(f_{(j)}(x))_{1\le t\le T}$ the elements of a vector such that $f_{(1)}(x)\le f_{(2)}(x) \le \ldots \le f_{(T)}(x)$.
Hence, $f_{(j)}(x) = - g^{(j)}(x)$ $\forall j \in \mathcal{T}$.
We then obtain the following maximization problem which is equivalent to \eqref{eq:ProbPropCVaR1_}
\begin{equation}
	\begin{array}{rl}
	\displaystyle	\max_{x} & \displaystyle \sum_{j=1}^T \lambda^j f_{(j)}(x)
		\\
		\mbox{s.t.} & x \in C
	\end{array}
	\label{eq:ProbPropCVaR2}
\end{equation}
where, again, $(\lambda^j)_{1\le j\le T}$ is a vector such that $\lambda^1\ge \lambda^2 \ge \ldots \ge \lambda^T\ge 0$.
Furthermore, using the approach described in \cite{nickel2006location} (Theorem 1.3, pag. 13) and \cite{marin2020fresh}, we can write
\begin{subequations}\label{eq:BinAssign}
  \begin{align}
\sum_{j=1}^T \lambda^j f_{(j)}(x) = & \displaystyle \min_{z} \; \sum_{t=1}^T \sum_{j=1}^T \lambda^j f_{t}(x) z_{jt} \label{eq:BinAssign1}  \\
& \mbox{s.t.}  \hspace*{0.1cm} \sum_{j=1}^T z_{jt} =1,\; t=1,\ldots,T, \label{eq:BinAssign2} \\
&  \hspace*{0.6cm} \sum_{t=1}^T z_{jt} =1,\; j=1,\ldots,T, \label{eq:BinAssign3} \\
&  \hspace*{0.6cm} z_{jt} \in \{0,1\} , \label{eq:BinAssign4}
\end{align}
\end{subequations}
Since the above problem is an assignment problem, and the constraint matrix \eqref{eq:BinAssign2}-\eqref{eq:BinAssign3} is totally unimodular, we can
relax the integrality constraints \eqref{eq:BinAssign4} into $z_{jt} \ge 0$, thus obtaining an LP.
Then, Problem \eqref{eq:BinAssign1}-\eqref{eq:BinAssign2}-\eqref{eq:BinAssign3} with $z_{jt} \ge 0$ has an exact dual so that
\begin{align}
\sum_{j=1}^T \lambda^j f_{(j)}(x) = & \displaystyle \max_{u, v} \; \sum_{t=1}^T u_t + \sum_{j=1}^T v_j  \\
& \mbox{s.t.}  \hspace*{0.1cm} u_t + v_j \le \lambda^j f_{t}(x), \; \forall j,t =1,\ldots,T, \nonumber \\
&  \hspace*{0.6cm} u_t, v_j \mbox{ free} \qquad \forall j,t =1,\ldots,T \nonumber
\end{align}
Therefore, if we wish to maximize the OWA of the centered CVaRs, then we can proceed by solving the following problem
\begingroup\makeatletter\def\f@size{11}\check@mathfonts
\begin{equation}\label{eq:OWA_a}
  \begin{array}{llll}
\displaystyle\max_{x} \; & \sum_{j=1}^T \lambda^j f_{(j)}(x)  = & \displaystyle\max_{x, u, v} & \sum_{t=1}^T u_t + \sum_{j=1}^T v_j \\
& & \mbox{s.t.}  & u_t + v_j \le \lambda^j(CVaR_{t/T} (R_I) -CVaR_{t/T}(R(x))), \; \forall j,t \\
& & & u_t, v_j \mbox{ free}, \; x\in C
\end{array}
\end{equation}
\endgroup
where $f_{(j)}(x) = (CVaR_{\cdot/T} (R_I) -CVaR_{\cdot/T}(R(x)))_{(j)}$ and $C$ is a polyhedron.
This representation still needs some further refinement to become linear since the CVaR function requires an appropriate expansion.
Knowing that for a fixed portfolio $x\in C$
\begin{align}
CVaR_{t/T}(R(x)):= &   \min_{\zeta \in \mathbb{R}, d} \;  \zeta + \frac{T}{t} \sum_{j=1}^T \pi_j d_j \nonumber \\
& \mbox{s.t.}  \hspace*{0.1cm} d_j\ge -\zeta - \sum_{\ell=1}^{n} r_{\ell j} x_{\ell}, \quad j=1,\ldots,T, \nonumber \\
& \hspace*{.6cm} d_j\ge 0,\; \forall j \, , \nonumber
\end{align}
we can express $f_t(x)= CVaR_{t/T} (R_I) -CVaR_{t/T}(R(x))$ using the following linear program
\begin{equation}\label{def:CVaR}
  \begin{array}{lll}
f_t(x):= &  \displaystyle \max_{\zeta \in \mathbb{R}, d} \; CVaR_{t/T}(R_I) - \zeta - \frac{T}{t}  \sum_{j=1}^T \pi_j d_j & \\
& \mbox{s.t.}  \hspace*{0.1cm} d_j\ge -\zeta - \sum_{\ell=1}^{n} r_{\ell j} x_{\ell}, \quad j=1,\ldots,T,  \\
& \hspace*{.6cm} d_j\ge 0,\; \forall j.
\end{array}
\end{equation}
Now, the largest feasible region for the inequalities $u_t+v_j +   \lambda^j( \zeta_{t} +\frac{1}{t} \sum_{k=1}^T d_{t k}) \le \lambda^j  CVaR_{t/T}(R_I)  , \; \forall j,t$ subject to the constraints of the above problem, is attained when $  \lambda^j(\zeta_{t} +\frac{1}{t} \sum_{k=1}^T d_{t k})$ is minimized. 
Thus, we obtain an alternative valid formulation for \eqref{eq:OWA_a} substituting $CVaR_{t/T}(R(x))$ by $\lambda^j( \zeta_{t} +\frac{1}{t} \sum_{k=1}^T d_{t k})$ plus the constraints in \eqref{def:CVaR}. Hence, Problem \eqref{eq:OWA_a} becomes equivalent to the following one:
\begin{equation}\label{eq:OWA_def}
  \begin{array}{ll}
\displaystyle\max_{x, u, v, d, \zeta} & \displaystyle \sum_{t=1}^T u_t + \sum_{j=1}^T v_j \\
\quad \mbox{s.t.}  & \displaystyle u_t + v_j \le \lambda^j ( CVaR_{t/T}(R_I) - \zeta_{t} - \frac{1}{t} \sum_{k=1}^T d_{t k} ), \; \forall j,t \\
& \displaystyle d_{t k}  \ge -\zeta_{t} - \sum_{\ell=1}^{n} r_{\ell k} x_{\ell} , \;
t, k=1,\ldots,T,\\
& \zeta_{t}, u_t, v_j \mbox{ free, } d_{t k} \ge 0, \\
& x \in C \, ,
\end{array}
\end{equation}
where, again, $C$ is a polyhedron.
This last formulation is a Linear Program (LP) with $n+T+T+T^2+T$ (namely, $x, u, v, d, \zeta$)  variables, and $2 T^2$ constraints plus those defining the polyhedron $C$.
This LP problem returns the portfolio $x$ that maximizes the ordered weighted average (OWA) of the centered CVaRs. 
The reader may think that these figures are large, but still, the model solves as a linear program, which simplifies a lot due to nowadays efficiency of current solvers.

\subsubsection{Maximizing OWA of the centered Tails \label{sec:MaxOWATails}}

In this section, we briefly discuss how to extend the results obtained from the multi-objective optimization problem based on centered CVaRs for Problem \eqref{eq:SSDTail}, based on centered Tails, which is recalled here for the sake of readability:
\begin{equation} \label{eq:SSDTail_b}
	\begin{array}{rl}
		\displaystyle\max_x & \left[ (Tail_{\frac{1}{T}}(R(x))-Tail_{\frac{1}{T}}(R_{I})),\ldots,
 (Tail_{\frac{T}{T}}(R(x))-Tail_{\frac{T}{T}}(R_{I}))  \right]
		\\
		\mbox{s.t.} & x \in C \, .
	\end{array}
\end{equation}
\cite{roman2013enhanced} scalarize Model \eqref{eq:SSDTail_b} considering the following Maximin problem
\begin{equation}
	\begin{array}{rl}
		\displaystyle\max_{x} & \displaystyle\min_{1 \leq j \leq T} (Tail_{\frac{j}{T}}(R(x))-Tail_{\frac{j}{T}}(R_{I}))
		\\
		\mbox{s.t.} & x \in C
	\end{array}
	\label{eq:SSDTailMaxMin}
\end{equation}
Note that Problems \eqref{eq:SSDCVaRMinMax} and \eqref{eq:SSDTailMaxMin} lead to different solutions
\citep[see Remark 1 in][]{roman2013enhanced}.
Following the same procedure described in Section \ref{sec:MaxOWACVaRs}, we can denote by  $h_t(x):=Tail_{t/T} (R(x))-Tail_{t/T}(R_I)$ with $t \in \mathcal{T}$
and by $(h_{(j)}(x))_{1\le t\le T}$ the elements of a vector such that $h_{(1)}(x)\le h_{(2)}(x) \le \ldots \le h_{(T)}(x)$.
Therefore, we consider the following problem
\begin{equation}
	\begin{array}{rl}
	\displaystyle	\max_{x} & \displaystyle \sum_{j=1}^T \lambda^j h_{(j)}(x)
		\\
		\mbox{s.t.} & x \in C
	\end{array}
	\label{eq:0ProbPropCVaR1}
\end{equation}
where $(\lambda^j)_{1\le j\le T}$ is a vector such that $\lambda^1\ge \lambda^2 \ge \ldots \ge \lambda^T\ge 0$.

\noindent
Following the same arguments that lead to Problem \eqref{eq:OWA_a} and knowing that for a fixed portfolio $x\in C$, we can express $h_t(x)$ as the following linear program
\begin{align}
h_t(x):= &   \max_{\zeta \in \mathbb{R}, d} \; -\frac{t}{T} \zeta -  \sum_{j=1}^T \pi_j d_j  - Tail_{t/T}(R_I) \label{tail} \\
& \mbox{s.t.}  \hspace*{0.1cm} d_j\ge -\zeta - \sum_{\ell=1}^{n} r_{\ell j} x_{\ell}, \quad j=1,\ldots,T, \nonumber \\
& \hspace*{.6cm} d_j\ge 0,\; \forall j \, , \nonumber
\end{align}
%
%Thus,
we can reformulate Problem \eqref{eq:0ProbPropCVaR1} as %
\begin{equation}
	\begin{array}{rl}
		\displaystyle	\max_{x, u, v, d, \zeta} & \displaystyle \sum_{t=1}^T u_t + \sum_{j=1}^T v_j \\
		
		\mbox{s.t.} & \displaystyle u_t + v_j \le \lambda_j (- \frac{t}{T} \zeta_{t} - \frac{1}{T} \sum_{k=1}^T d_{t k} -Tail_{t/T}(R_I)), \; \forall j,t \\
& \displaystyle d_{t k}  \ge -\zeta_{t} - \sum_{\ell=1}^{n} r_{\ell k} x_{\ell} , \;
t, k=1,\ldots,T,\\
&  \zeta_{t}, u_t, v_j \mbox{ free, } d_{t k} \ge 0, \\
&  x \in C.
	\end{array}
	\label{eq:OWATailcase}
\end{equation}
Again, the above problem is an LP, and its solution returns the portfolio $x$ that maximizes the ordered weighted average (OWA) of the centered tails.

\section{Theoretical Analysis \label{sec:TheoreticalAnalysis}}

This section discusses some theoretical features of the OWA maximization of centered CVaRs approach and the CS$\mathcal{E}$SD relations.
More precisely, Section \ref{OWA-CESD} shows the theoretical link between Model \eqref{eq:SSDCVaR} scalarized by the OWA operator and the CS$\mathcal{E}$SD rule.
In Section \ref{OWA-CESD2}, we provide some statistical motivations for using the OWA operator. Furthermore, we discuss the economic interpretation of the CS$\mathcal{E}$SD concept, emphasizing its main advantages.
Finally, Section \ref{sec:examples} shows some illustrative examples about the $S \varepsilon S D$ (or, equivalently, ASSD), $C S \varepsilon S D$, and $C S \mathcal{E} S D$ relations in three practical cases.

\subsection{OWA maximization and CS$\mathcal{E}$SD \label{OWA-CESD}}

We provide here the theoretical relationship between Model \eqref{eq:ProbPropCVaR2}, which consists of maximizing the OWA of the centered CVaRs, and the CS$\mathcal{E}$SD criterion, described in Definition \ref{def:CSESD}.

\noindent
Let $g^{(t)}(x) := \left( CVaR_{t/T}(R(x))-CVaR_{t/T} (R_I) \right)^{(t)}$
and $G^{(t)}(x) := \sum_{j=1}^t g^{(j)}(x)$ for each $t \in \mathcal{T}$.
Note that the efficient portfolios, whose returns dominate the return of a given benchmark index $R_I$ w.r.t. CS$\mathcal{E}$SD, are those obtained by solving the following problems
\begin{equation}
	\begin{array}{rl}
	\displaystyle	\min_{x} & \displaystyle  G^{(j)}(x)
		\\
		\mbox{s.t.}  & G^{(t)}(x) \leq \varepsilon_t \quad t=1,\ldots,T \quad t\neq j \\
		& x \in C
	\end{array}
	\label{eq:ProbPropCVaR3}
\end{equation}
for each $j\in \mathcal{T}$.
In the following we show that Problem \eqref{eq:ProbPropCVaR3} is equivalent to Problem \eqref{eq:ProbPropCVaR2}.
For the sake of readability, we recall below some well-known concepts in multi-criteria optimization.
Consider the following multi-objective optimization problem
\begin{equation}\label{eq:MOP}
  \min _{x \in \mathcal{X}} \left( h_{1}(x), h_{2}(x), \ldots h_{p}(x) \right) \, .
\end{equation}
A standard technique to scalarize Problem \eqref{eq:MOP}
is the $\varepsilon$-constraint method, namely
\begin{equation}
	\begin{array}{rl}
	\displaystyle	\min_{x} & \displaystyle  h_{j}(x)
		\\
		\mbox{s.t.}  & h_{k}(x) \leq \varepsilon_k \quad k=1,\ldots,p \quad k\neq j \\
		& x \in \mathcal{X}
	\end{array}
	\label{eq:MOP1}
\end{equation}
An alternative is the following weighted sum
problem
\begin{equation}\label{eq:MOP2}
  \min _{x \in \mathcal{X}} \sum_{k=1}^{p} \gamma_{k} h_{k}(x)
\end{equation}
We first recall a theorem that links these two methods.
\begin{theorem}[Theorem 4.6 of \cite{ehrgott2005multicriteria}] \label{th:ehrgott}
%\hfill \break
\hspace{10cm}
\begin{enumerate}
  \item Suppose $\hat{x}$ is an optimal solution of \eqref{eq:MOP2}. If $\gamma_{j}>0$ there exists $\hat{\varepsilon}$ such that $\hat{x}$ is an optimal solution of \eqref{eq:MOP1}, too.
  \item Suppose $\mathcal{X}$ is a convex set and $h_{k}: \mathbb{R}^{n} \rightarrow \mathbb{R}$ are convex functions. If $\hat{x}$ is an optimal solution of \eqref{eq:MOP1} for some $j$, there exists $\hat{\gamma} \in \mathbb{R}_{\geq}^{p}$ such that $\hat{x}$ is optimal for $\min _{x \in \mathcal{X}} \sum_{k=1}^{p} \hat{\gamma}_{k} h_{k}(x)$.
\end{enumerate}
\end{theorem}

\noindent
The following result is instrumental to get the announced equivalence between models.
\begin{lemma} Assume that $C$ is a convex set. Then, $G^{(t)}$ is a convex function for all $t\in \mathcal{T}$. \label{le:num8}
\end{lemma}
\begin{proof}
Let $t\in \mathcal{T}$. For $x,y\in C$ and $\alpha \in [0,1]$ we have that
$$G^{(t)} (\alpha x + (1-\alpha) y) = \sum_{j=1}^t g^{(j)}(\alpha x + (1-\alpha) y).$$
Let $k_1,\ldots,k_t$ be indices such that
$$\sum_{j=1}^t g^{(j)}(\alpha x + (1-\alpha) y) = \sum_{j=1}^t g_{k_j}(\alpha x + (1-\alpha) y).$$
Thus, since CVaR function is convex,
$$\sum_{j=1}^t g_{k_j}(\alpha x + (1-\alpha) y)\leq \alpha \sum_{j=1}^t g_{k_j}(x) +  (1-\alpha) \sum_{j=1}^t g_{k_j}(y).$$
Finally, note that
$$\alpha \sum_{j=1}^t g_{k_j}(x) +  (1-\alpha) \sum_{j=1}^t g_{k_j}(y)\leq \alpha \sum_{j=1}^t g^{(j)}(x) +  (1-\alpha) \sum_{j=1}^t g^{(j)}(y) = \alpha G^{(t)}(x) + (1-\alpha)G^{(t)}(y).$$
\end{proof}
\noindent
Exploiting Theorem \ref{th:ehrgott} and Lemma \ref{le:num8}, we now provide the theoretical link between Problems \eqref{eq:ProbPropCVaR3} and \eqref{eq:ProbPropCVaR1_}.
\begin{theorem} \label{th:num9}
\hfill
%\break
%\hspace{1cm}
%
\begin{enumerate}
\item Suppose $\hat{x}$ is an optimal solution of Problem \eqref{eq:ProbPropCVaR1_}.
If $\lambda^j-\lambda^{j+1}>0$, then there exists $\underline{\varepsilon} \in \mathbb{R}_{\geq}^{T}$ such that $\hat{x}$ is an optimal solution of Problem \eqref{eq:ProbPropCVaR3}, too.
\item If $\hat{x}$ is an optimal solution of Problem \eqref{eq:ProbPropCVaR3} for some $j$, there exists $\hat{\lambda} \in \mathbb{R}_{\geq}^{T}$ with $\hat{\lambda}_1\geq \ldots \geq \hat{\lambda}_T$ such that $\hat{x}$ is optimal for Problem \eqref{eq:ProbPropCVaR1_}.
\end{enumerate}
\end{theorem}
\begin{proof}
\hspace{10cm}
\begin{enumerate}
\item We note that the objective function of Problem \eqref{eq:ProbPropCVaR1_} can be rewritten as follows
    \begin{eqnarray}
      \sum_{j=1}^T \lambda^j g^{(j)}(x) &=& \sum_{j=1}^T \sum_{t=j}^T \left( \lambda^t-\lambda^{t+1} \right) g^{(j)}(x) \quad \mbox{or, equivalently,} \label{eq:rel1} \\
       &=& \sum_{j=1}^T \sum_{t=1}^j \left( \lambda^j-\lambda^{j+1} \right) g^{(t)}(x) \label{eq:rel2} \\
       &=& \sum_{j=1}^T \left( \lambda^j-\lambda^{j+1} \right) G^{(j)}(x) \nonumber
    \end{eqnarray}
where $G^{(j)} := \sum_{t=1}^j g^{(t)}$ for each $j \in \mathcal{T}$ and we assume that $\lambda^{T+1}=0$.
Thus, we can reformulate Problem \eqref{eq:ProbPropCVaR1_} as
\begin{equation}
	\begin{array}{rl}
	\displaystyle	\min_{x} & \displaystyle \sum_{j=1}^T \left( \lambda^j-\lambda^{j+1} \right) G^{(j)}(x)
		\\
		\mbox{s.t.} & x \in C \, .
	\end{array}
	\label{eq:ProbPropCVaR1}
\end{equation}
Knowing that, by construction, $\lambda^1\ge \lambda^2 \ge \ldots \ge \lambda^T\ge 0$ and by hypothesis of item 1, $\lambda^j-\lambda^{j+1} >0$ for some $j \in \mathcal{T}$,
the result follows by Theorem \ref{th:ehrgott}.
\item Assuming that $\hat{x}$ is an optimal solution of Problem \eqref{eq:ProbPropCVaR3} and using point 2 of Theorem \ref{th:ehrgott}, there must exist some $\hat{\gamma} = \left(\hat{\gamma}^{1}, \hat{\gamma}^{2}, \ldots, \hat{\gamma}^{T}\right) \in \mathbb{R}_{\geq}^{T}$ such that
$$
\sum_{t=1}^{T} \hat{\gamma}^{t} G^{(t)}(x) \geq \sum_{t=1}^{T} \hat{\gamma}^{t} G^{(t)}(\hat{x})
$$
for all $x \in C$. Equivalently,
$$
\sum_{t=1}^{T} \hat{\gamma}^{t} \sum_{j=1}^{t} g^{(j)}(x) \geq \sum_{t=1}^{T} \hat{\gamma}^{t} \sum_{j=1}^{t} g^{(j)}(\hat{x})
$$
Exploiting relations \eqref{eq:rel1} and \eqref{eq:rel2} and denoting by $\hat{\gamma}^{t}=\hat{\lambda}^{t}-\hat{\lambda}^{t+1}$ $\forall t \in \mathcal{T}$, but
with $\hat{\lambda}^{T+1}=0$,
we have
$$
\sum_{t=1}^{T} \sum_{j=t}^{T} \left(\hat{\lambda}^{j}-\hat{\lambda}^{j+1}\right) g^{(t)}(x) \geq \sum_{t=1}^{T} \sum_{j=t}^{T} \left(\hat{\lambda}^{j}-\hat{\lambda}^{j+1}\right) g^{(t)}(\hat{x}) \, ,
$$
namely,
$$
\sum_{t=1}^T \hat{\lambda}^t g^{(t)}(x) \geq \sum_{t=1}^T \hat{\lambda}^t g^{(t)}(\hat{x}) \, ,
$$
for all $x \in C$, which completes the proof.
\end{enumerate}
\end{proof}

\subsection{Statistical and economic interpretations of CS$\mathcal{E}$SD \label{OWA-CESD2}}

In the previous section, we have shown (see Theorem \ref{th:num9}) the equivalence between the maximization of the OWA of the centered CVaRs, and the CS$\mathcal{E}$SD rule.

\noindent
Next, we note that the OWA of the centered CVaRs, for $\lambda=(1,\stackrel{t}{\ldots},1,0,\ldots,0)$ can be interpreted as the $t/T$-CVaR of the centered CVaRs.
Therefore, from a statistical viewpoint, according to \cite{BertsimasS04}, it is more robust to maximize the $t/T$-CVaR of the centered CVaRs than the minimax model ($1/T$-CVaR) proposed by \cite{roman2013enhanced}. 
It has been well-known since the early '70s that minimax models provide robust solutions against uncertain data, although they may be too conservative. 
An alternative is to consider the minimization of $t$-sums. 
In this regard, the reader may note that our proposal, namely the minimization of the $t$ largest losses of a sequence, provides more robust solutions than considering only the largest loss ($t=1$). 
The interested reader is referred to the papers by \cite{BertsimasS04}, \cite{puerto_rodriguez-chia_tamir_2017}, and the references therein for further insights on $k$-sum optimization.

On the other hand, based on the equivalence provided by Theorem \ref{th:num9}, we can also give a insightful interpretation to the CS$\mathcal{E}$SD rule. Indeed, any CS$\mathcal{E}$SD portfolio comes from an optimal solution of an OWA solution of centered CVaRs for particular choices of lambda weights. 
In other words, the CS$\mathcal{E}$SD portfolios optimize weighted sums of centered CVaRs; each choice of the lambda parameters reflects the importance given by the decision-maker to the ordinal sequence of centered CVaRs.

Finally, we discuss on the advantages of the $C S \varepsilon S D$ concept.
We prove that it is a refinement of the concept of ASSD in \cite{lizyayev2012tractable} and, moreover, that the expected utility of a $C S \varepsilon S D$ portfolio does not exceed the expected utility of the dominating one by more than $\varepsilon$ for a large class of reasonable utility functions.
\begin{prop}\label{prop:Thoretical}
\mbox{ \null}
Let $A$ and $B$ be two random variables.
\begin{enumerate}
\item If $A$ $C S \varepsilon S D$-dominates $B$ then $A$ ASSD-dominates $B$  ($C S \varepsilon S D\Rightarrow$ ASSD).
\item If A $C S \varepsilon S D$-dominates B, then $E[ u(A)] + \varepsilon \geq E[u(B)]$ for any nondecreasing concave utility function $u$ with first derivative $u^{\prime} \leq 1$.
\end{enumerate}
\end{prop}
\begin{proof}
  From Proposition \ref{prop:CSeSD_form2} it follows that $C S \varepsilon S D$ implies ASSD, or equivalently $S \varepsilon S D$, which requires  $\left(CVaR_{\frac{\cdot}{T}}(A)-CVaR_{\frac{\cdot}{T}}(B)\right)^{(1)} \leq \varepsilon$. This proves assertion (1.).

\noindent
Next, to prove assertion (2.) we use Theorem 2 in \cite{lizyayev2012tractable} and the result in assertion (1.).
\end{proof}

\subsection{Illustrative examples for $S \varepsilon S D$, $C S \varepsilon S D$ and $C S \mathcal{E} S D$}\label{sec:examples}

To support intuition, in Fig. \ref{fig:CumCVaRExample}, we give some illustrative examples of how the $S \varepsilon S D$ (and hence the Minimax scalarization), $C S \varepsilon S D$ and $C S \mathcal{E} S D$ (and hence the OWA scalarization) conditions work in three practical situations covering all possible real-world cases.
More precisely, we consider two discrete random variables $A$ and $B$ represented by five scenarios $t=1,...,5$.
Denoting by $g^{(t)} = \left( CVaR_{\cdot/T}(A)-CVaR_{\cdot/T} (B) \right)^{(t)}$
and $G^{(t)} := \sum_{j=1}^t g^{(j)}$, on the left side of Fig. \ref{fig:CumCVaRExample} we show a bar graph of $g^{(t)}$ as a function of $t$, while on the right side we report $G^{(t)}$ as a function of $t$.
In Case 1, $g^{(t)}>0$ $\forall t$, therefore all the SSD inequalities are violated;
in Case 2, $g^{(t)}>0$ for $t=1,2,3$ and $g^{(t)}<0$ for $t=4,5$;
in Case 3, $g^{(t)}<0$ $\forall t$, namely $A$ SSD-dominates $B$.

\noindent
In simple terms, the Minimax scalarization of \cite{roman2013enhanced} consists in fixing a single threshold ``epsilon'' (see the blue dashed line on the left side of Fig. \ref{fig:CumCVaRExample}) and minimizing it; therefore, we can only focus on the worst violation $g^{(1)}$.
On the other hand, using the OWA scalarization with an appropriate choice of lambda weights, we can potentially consider one or more thresholds "epsilon" (see the red dashed lines on the right side of Fig. \ref{fig:CumCVaRExample}), limiting not only the worst violation $G^{(1)}=g^{(1)}$ but also the remaining cumulative sums of the violations.
\begin{figure}[tbp]
\centering
\includegraphics[width=1.0\textwidth]{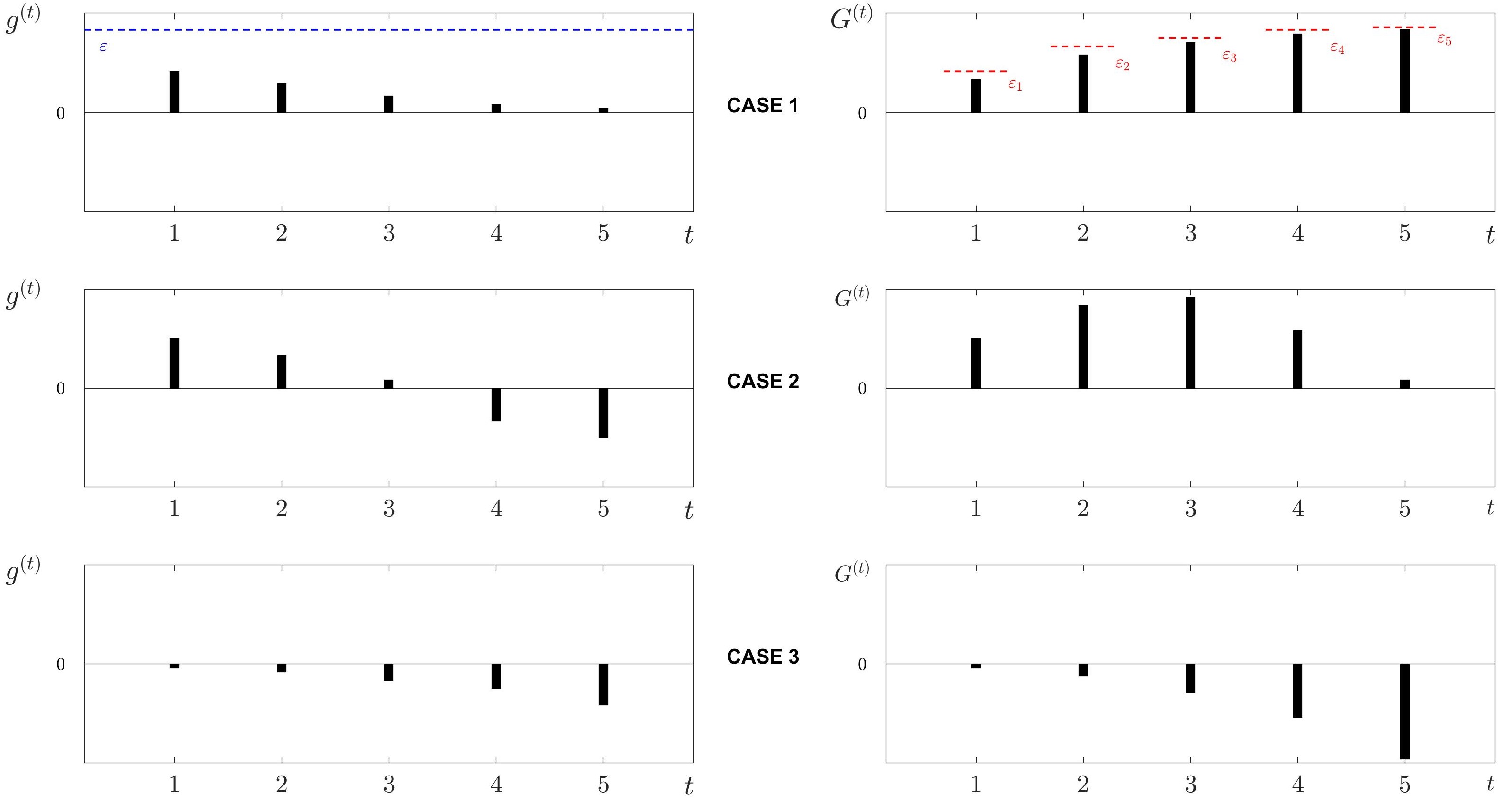} %
\caption{Illustrative examples of $S \varepsilon S D$ (left side), $C S \varepsilon S D$ and $C S \mathcal{E} S D$ (right side).}
\label{fig:CumCVaRExample}
\end{figure}

\section{Empirical Analysis \label{sec:EmpiricalAnalysis}}

We present here a detailed empirical analysis where the performance of efficient portfolios obtained by the OWA maximization of centered CVaRs (or Tails), are compared with those achieved by different SD-based strategies proposed in the literature for Enhanced Indexation.
In addition, we also consider 
the long-only,
long-short, 1-norm, and 2-norm minimum variance portfolios.
Summing up, we thus compare the portfolio selection models listed below.
\begin{itemize}
  \item \textbf{1-OWA-CVaR}$\, \equiv \,$ \textbf{RomanCVaR}: Model \eqref{eq:OWA_def} with $\lambda=(1, 0, 0, \ldots, 0)$, which is equivalent to Model \eqref{eq:SSDCVaRMinMax} \citep{roman2013enhanced}.
  \item \textbf{$k$-OWA-CVaR-$\beta$}: Model \eqref{eq:OWA_def} with $\lambda=(1,1,\stackrel{(k)}{\ldots},1,0,\ldots,0)$, where $k$ rounds $\beta T$ to the nearest integer, and the values of $\beta$ considered are $5\%, 10\%,25\%,50\%, 75\%, 100\%$.
     \item \textbf{$k$-OWA-CumCVaR-$\beta$}: Model \eqref{eq:OWA_def} with $\lambda=(k,k-1,k-2, \ldots, 2, 1,0,\ldots,0)$, where $k$ rounds $\beta T$ to the nearest integer, and the values of $\beta$ considered are $5\%, 10\%,25\%,50\%, 75\%, 100\%$.
      \item \textbf{1-OWA-Tail}$\, \equiv \,$ \textbf{RomanTail}: Model \eqref{eq:OWATailcase} with $\lambda=(1, 0, 0, \ldots, 0)$, which is equivalent to Model \eqref{eq:SSDTailMaxMin} \citep{roman2013enhanced}.
  \item \textbf{$k$-OWA-Tail-$\beta$}: Model \eqref{eq:OWATailcase} with $\lambda=(1,1,\stackrel{(k)}{\ldots},1,0,\ldots,0)$, where $k$ rounds $\beta T$ to the nearest integer, and the values of $\beta$ considered are $5\%, 10\%,25\%,50\%, 75\%, 100\%$.
  \item \textbf{$k$-OWA-CumTail-$\beta$}: Model \eqref{eq:OWATailcase} with $\lambda=(k,k-1,k-2, \ldots, 2, 1,0,\ldots,0)$, where $k$ rounds $\beta T$ to the nearest integer, and the values of $\beta$ considered are $5\%, 10\%,25\%,50\%, 75\%, 100\%$.
  \item \textbf{KP2011Min}: Model (10) of \cite{kopa2015general} with
      $$
      w_s =\left\{
    \begin{array}{ll}
      0.001, & \mbox{for $s=2, \ldots, T$} \\
      1- \sum_{i=2}^{T} w_i , & \mbox{for $s=1$}
    \end{array}
  \right.
      $$
\citep[see also][]{hodder2015improved}.
      \item \textbf{MinV}: long-only minimum variance portfolio \citep{Mark:52,Mark:59}, $x_{{MinV}}^{*}=\displaystyle\argmin_{x \geq 0} x^{T} \Sigma x$, where $\Sigma$ is the sample covariance matrix obtained from the returns of the assets belonging to an investment universe.
  \item \textbf{MinVsh}: long-short minimum variance portfolio, $x_{{MinV\!sh}}^{*}=\displaystyle\argmin_{x \in \mathbb{R}} x^{T} \Sigma x $.
  \item \textbf{MinVN1}: 1-norm constrained minimum variance portfolio, $x_{{MinV\!N1}}^{*}=\displaystyle\argmin_{x \in \mathbb{R}} x^{T} \Sigma x + \delta_1 \sum_{k=1}^{n} | x_k |$ \citep[see, e.g.][]{dhingra2023norm}.
  \item \textbf{MinVN2}: 2-norm constrained minimum variance portfolio, $x_{{MinV\!N2}}^{*}=\displaystyle\argmin_{x \in \mathbb{R}} x^{T} \Sigma x + \delta_2 \sum_{k=1}^{n} x^{2}_{k} $ \citep[see, e.g.][]{coqueret2015diversified,dhingra2023norm}.
\end{itemize}

\noindent
The models have been implemented in MATLAB R2019b, and they make calls to XPRESS solver version 8.9 for solving the LP problems. All experiments were run on a computer DellT5500 with an Intel(R) Xeon(R) processor with a CPU X5690 at 3.75 GHz and 48 GB of RAM.

\subsection{Datasets and Methodology}

In this section, we provide some details about the four real-world datasets on which we test all the portfolio selection strategies described in the previous section:
\begin{enumerate}
  \item FTSE100 (Financial Times Stock Exchange, UK), containing 80 assets and 3715 observations (Oct 2006-Dec 2020);
  \item NASDAQ 100 (National Association of Securities Dealers Automated Quotation, USA), containing 54 assets and 3715 observations (Oct 2006-Dec 2020);
  \item SP500 (Standard \& Poor's, USA), containing 336 assets and 3715 observations (Oct 2006-Dec 2020);
  \item FF49 (Fama \& French 49 Industry portfolios, USA), containing 49 portfolios considered as assets and using the subsample where all the returns of the 49 industries are available, namely from July 1969 to June 2023.
\end{enumerate}
The first three datasets are publicly available on the website %
\url{https://www.francescocesarone.com/data-sets} 
and consist of daily prices, adjusted for dividends and stock splits and obtained from Refinitiv \citep[see, for more details,][]{cesarone2022ESG}.
For the last dataset, we use daily returns downloaded from
\url{https://mba.tuck.dartmouth.edu/pages/faculty/ken.french/data_library.html#Research}. In this case, we consider as a benchmark the Equally-Weighted portfolio.

Our aim is to examine the performance of all models described above, using a rolling time windows
scheme of evaluation.
More precisely, we consider an in-sample time window of 6 months (i.e., 125 financial days).
We then evaluate the portfolio out-of-sample performance over the next month (i.e., 20 financial days). Next, we shift the in-sample window by 1 month to cover the out-of-sample period, recalculate the optimal portfolio relative to the new in-sample window, and repeat the operation.

\noindent
For each portfolio strategy, we calculate some standard performance measures recalled below \citep[see, e.g.,][]{cesarone2015linear,cesarone2016optimally,cesarone2020optimization,bruni2017exact,cesarone2017minimum}.
\begin{itemize}
  \item The average $\hat{\mu}^{out}$ (\textbf{ExpRet}) and the standard deviation $\hat{\sigma}^{out}$ (\textbf{Vol}) of the out-of-sample portfolio returns.
  \item The \textbf{Sharpe} ratio defined as $\hat{\mu}^{out}/\hat{\sigma}^{out}$ \citep{10.2307/2351741,sharpe1994sharpe}.
  \item The Maximum Drawdown \citep[\textbf{MDD}, see, e.g.,][and references therein]{chekhlov2005drawdown}.
 Denoting by $R_{\tau}^{out}$ the out-of-sample portfolio returns
for each portfolio strategy, we first consider the cumulative out-of-sample portfolio returns, which correspond to the values of wealth after $\tau $
periods $W_{\tau }=W_{\tau -1}(1+R_{\tau }^{out})$, where $\tau =M+1,\ldots ,\bar{T}$, $M=125$, $M$ is the length of the in-sample window, $\bar{T}$ is the number of observations available for a single dataset, and the initial wealth $W_{M}=1$.
We define the drawdowns as
\begin{equation*}
dd_{\tau }=\frac{W_{\tau }-\max_{M+1 \leq s \leq \tau }(W_{s})}{\max_{M+1 \leq s
\leq \tau }(W_{s})}. % \label{dd}
\end{equation*}
Note that the drawdowns $dd_{\tau }$ are obviously negative.
The Maximum Drawdown $Mdd$ corresponds to the worst drawdown over the entire
out-of-sample period, namely $Mdd=\min_{M+1 \leq \tau \leq T}(dd_{\tau })$.
\item The \textbf{Ulcer} Index, which is defined as the square root
of the mean of the squared percentage drawdowns $dd_{\tau }$ with $\tau
=M+1,\ldots ,T$
\begin{equation}
UI=\sqrt{\frac{\sum_{\tau=M+1}^{T}dd_{\tau}^{2}}{T-M}}.  \label{eq:UI}
\end{equation}
It evaluates the depth and the duration of drawdowns
in prices over the out-of-sample period \citep[see, e.g.,][]{martin1989investor}.
\item The \textbf{Sortino} ratio \citep{sortino2001managing}, defined as the ratio between the
average of $R^{out}$ and its downside deviation
$$
\frac{E[R^{out}-r_f]}{\sigma(\min\{R^{out}-r_f,\, 0\})} \, .
$$
\item The \textbf{Rachev} ratio with confidence levels of $5\%$ \citep{biglova2004different}, defined as the ratio between the average of the best $5\%$
values of $R^{out}$ (with the opposite sign) and that of the worst $5\%$ values of $R^{out}$
$$
\frac{C\!V\!a\!R_{5\%}( - R^{out})}{C\!V\!a\!R_{5\%}(R^{out})} \, .
$$
\item The \textbf{Turnover} \citep[see, e.g.,][]{Demiguel2009} evaluates the amount of trading required to perform in practice the portfolio strategy, and is defined as
\begin{equation*}
Turnover = \frac{1}{Q}\sum_{q=1}^{Q}\sum_{k=1}^{n}\mid x_{q,k}-x_{q-1,k}\mid ,
\end{equation*}
where $Q$ represents the number of rebalances, $x_{q,k}$ is the portfolio weight of asset $k$ after rebalancing, and $x_{q-1,k}$ is the portfolio weight before rebalancing at time $q$. Lower turnover values indicate better portfolio performance. We point out that this definition of portfolio turnover is a proxy of the effective one, since it evaluates only the amount
of trading generated by the models at each rebalance,
without considering the trades due to changes in asset
prices between one rebalance and the next. Thus, by definition,
the turnover of the EW portfolio is zero.
\item The Jensen's Alpha \citep[\textbf{AlphaJ},][]{jensen1968performance}, defined as the intercept of the line given by the linear regression of $R^{out}$ on $R_{I}^{out}$, namely
       $
        \alpha = E[R^{out}] - \beta E[R_{I}^{out}]
        $,
where $\beta=Cov(R^{out}, R_{I}^{out})/ \sigma^2(R_{I}^{out})$.
\item The Value-at-Risk of the out-of-sample portfolio returns $R_{\tau}^{out}$ with a confidence level equal to 1\% (\textbf{VaR1}).
\item The \textbf{Appraisal} ratio \citep{treynor1973use}, defined as the Jensen's Alpha divided by the standard deviation of the residual return $\epsilon_r$, namely:
$$
\frac{\alpha}{\sigma(\epsilon_r)} \, .
$$
where $\epsilon_r=R^{out} -\alpha - \beta E[R_{I}^{out}]$.
The larger is its value, the better is the portfolio performance.
This measure is particularly used by practitioners for evaluating the asset selection ability of a model.
It grasps the idea that a portfolio manager to generate $\alpha$ has to deviate from the benchmark, i.e., take residual risk $\sigma(\epsilon_r)$.
\item The Information ratio \citep[\textbf{Info},][]{goodwin1998information}, defined as the expected value of the difference between the out-of-sample portfolio return and that of the benchmark index divided by the standard deviation of such difference, namely
$$
\frac{E[R^{out}-R_{I}^{out}]}{\sigma(R^{out}-R_{I}^{out})} \, .
$$
The larger is its value, the better is the portfolio performance.
This measure is equivalent to the Appraisal ratio with $\beta=1$.
\item The \textbf{Omega} ratio, defined as
$
\displaystyle\frac{E[\max(0, R^{out} - \eta)]}{E[\min(0, R^{out} - \eta)]} \, ,
$
where $\eta$ is fixed to 0.
\item The average number of selected assets (\textbf{ave \#}) and running \textbf{time} in seconds.

\item The Return On Investment (\textbf{ROI}), namely the time-by-time return generated by each portfolio strategy over a specified time horizon $\Delta H$.
More precisely, ROI is defined as follows
\begin{equation}\label{eq:ROI}
ROI_{\tau }=\frac{W_{\tau }-W_{\tau - \Delta H}}{W_{\tau - \Delta H}} \qquad \tau = M + \Delta H + 1, \ldots, \bar{T}
\end{equation}
where $W_{\tau }=W_{\tau - \Delta H} \prod_{t=\tau - \Delta H +1}^{\tau} (1+R_{t}^{out})$ is the portfolio wealth generated by investing at the beginning of the time horizon
the amount of capital $W_{\tau - \Delta H}$.
\end{itemize}
We also test the statistical significance of our results. To this end, the significance of the difference between the Sharpe ratios of the selected portfolios and that of the benchmark is tested both using the bootstrapping and HAC methodologies as presented in \cite{Ledoit2008}.
Moreover, to test the statistical significance of the difference between the variances of the selected portfolios and that of the benchmark,
we use the bootstrapping method as proposed in \cite{ledoit2011robust}.
For both the bootstrap tests, we choose $M = 3000$ bootstrap repetitions and
block sizes $b =1, 10$.
Furthermore, we adopt a $t$-statistic as in \cite{goodwin1998information} to test the statistical significance of the positivity of the Information ratio.

\noindent
Tables \ref{tab:Perf_out_FTSE100_DL125_DH20_C}, \ref{tab:Perf_out_NASDAQ100_DL125_DH20_C}, \ref{tab:Perf_out_SP500_DL125_DH20_C}, and
\ref{tab:Perf_out_49IndPort_DL125_DH20_C} report the $p$-values in percentage
\begin{itemize}
  \item for the Sharpe ratio differences in the columns \textbf{Sharpe$\shortrightarrow$p-v/$b=1$}, \textbf{Sharpe$\shortrightarrow$p-v/$b=10$} (bootstrap test), and \textbf{Sharpe$\shortrightarrow$p-v/HAC} (HAC test);
  \item for the variance differences in the columns \textbf{Vol$\shortrightarrow$p-v/$b=1$} and \textbf{Vol$\shortrightarrow$p-v/$b=10$} (bootstrap test);
  \item for the positivity of the Information ratio in the column \textbf{InfoRatio$\shortrightarrow$p-v} ($t$-test).
\end{itemize}
Note that for FTSE100, NASDAQ100, and  SP500 the benchmark is represented by the Market index, while for FF49 by the Equally-Weighted portfolio.

\subsection{Computational results}

In this section, we report the out-of-sample results obtained by the OWA maximization of centered CVaRs approach, that are compared with those obtained by the Roman strategy with CVaRs and by the long-only, long-short, 1-norm, and 2-norm minimum variance portfolios.

\noindent
We point out that for the 1-norm constrained minimum variance portfolio we set $\delta_1=10^{-5}$, while for the 2-norm constrained minimum variance portfolio $\delta_2=10^{-3}$.
Furthermore, in the case of SP500, since the solution of the standard minimum variance model might not be unique, we solve the 2-norm minimum variance model regularizing the covariance matrix with $\delta_2=10^{-6}$.

\noindent
Appendix \ref{sec:OWATails} provides some additional computational results that compare the OWA maximization of centered Tails approach with that of Kopa and Post, which can be interpreted by means of Tails of portfolio returns, and, again, the MinV portfolio.

\noindent
To simplify the reading of the results reported in the following tables, for each dataset, we show with different colors the rank of the performance results of the proposed models. More specifically, for each column corresponding to a specific performance measure the colors span
from deep-green to deep-red, where deep-green represents
the best value while deep-red the worst
one.

In Tables \ref{tab:Perf_out_FTSE100_DL125_DH20_C}, \ref{tab:Perf_out_NASDAQ100_DL125_DH20_C},  \ref{tab:Perf_out_SP500_DL125_DH20_C}, and
\ref{tab:Perf_out_49IndPort_DL125_DH20_C}
we provide the out-of-sample performance results obtained by
the SD-based strategies using CVaRs, that are
compared with the long-only, long-short, 1-norm, 2-norm minimum variance portfolios and the Market Index on FTSE100, NASDAQ100, SP500, and FF49, respectively.
We can observe that the $k$-OWA-CVaR-$\beta$ and $k$-OWA-CCVaR-$\beta$ portfolios with $\beta=5\%, 10\%$ typically show the best performance in terms of \textbf{ExpRet}, \textbf{Sharpe}, \textbf{Sortino}, \textbf{Rachev}, \textbf{AlphaJ}, \textbf{Info}, \textbf{Appraisal}, \textbf{Omega}.
We find that on FTSE100 and FF49, the
Sharpe and Information ratios obtained with the $k$-OWA-CVaR-$\beta$ and $k$-OWA-CCVaR-$\beta$ portfolios ($\beta=5\%, 10\%$) are statistically
different from the benchmark,
while statistically significant differences are not realized in
the other investment universes.
These strategies also seem to provide slightly better performance than the Roman-CVaR portfolio.
On the other hand, the best values in terms of risk, namely \textbf{Vol}, \textbf{MDD}, \textbf{Ulcer}, and \textbf{VaR1}, are shown by the $k$-OWA-CVaR-$\beta$ and $k$-OWA-CCVaR-$\beta$ portfolios with $\beta=75\%, 100\%$, and by the MinV, MinVN1, and MinVN2 portfolios.
We observe that on all the datasets analyzed, the
variances obtained with the $k$-OWA-CVaR-$\beta$ and $k$-OWA-CCVaR-$\beta$ portfolios ($\beta=75\%, 100\%$), and with the MinV, MinVN1, and MinVN2 portfolios are statistically different from the benchmark.
Note that by increasing the values of $\beta$ for the $k$-OWA-CVaR-$\beta$ and $k$-OWA-CCVaR-$\beta$ portfolios, an investor can mitigate risk similarly to the minimum variance portfolios.
The best performance in terms of Turnover are typically achieved by the MinV and MinVN2 portfolios and by the $k$-OWA-CVaR-$\beta$ and $k$-OWA-CCVaR-$\beta$ portfolios with $\beta=75\%, 100\%$.
% CCIAO
% Table generated by Excel2LaTeX from sheet 'test (3)'
\begin{sidewaystable}[htbp]
  \centering
  \caption{Out-of-sample performance results for FTSE100}
  \scalebox{0.63}{
  {\renewcommand{\arraystretch}{1.1}
        \begin{tabular}{|l|c|c|c|c|c|r|r|r|c|c|c|c|c|c|c|c|c|c|c|}
    \hline
    \textbf{Approach} & \textbf{ExpRet} & \multicolumn{3}{c|}{\textbf{Vol}} & \multicolumn{4}{c|}{\textbf{Sharpe}} & \textbf{MDD} & \textbf{Ulcer} & \textbf{Sortino} & \textbf{Rachev5} & \textbf{Turn} & \textbf{AlphaJ} & \multicolumn{2}{c|}{\textbf{InfoRatio}} & \textbf{ApprRatio} & \textbf{VaR1} & \textbf{Omega} \\
    \hline
          &       & \textbf{value} & \multicolumn{1}{l|}{\textbf{p-v/b=1}} & \multicolumn{1}{l|}{\textbf{p-v/b=10}} & \textbf{value} & \multicolumn{1}{l|}{\textbf{p-v/b=1}} & \multicolumn{1}{l|}{\textbf{p-v/b=10}} & \multicolumn{1}{l|}{\textbf{p-v/HAC}} &       &       &       &       &       &       & \textbf{value} & \textbf{p-v} &       &       &  \\
    \hline
    \textbf{Roman-CVaR} & \cellcolor[rgb]{ .824,  .875,  .51}0.0490\% & \cellcolor[rgb]{ .973,  .412,  .42}1.3023\% & \multicolumn{1}{r|}{22.26\%} & \multicolumn{1}{r|}{17.39\%} & \cellcolor[rgb]{ .996,  .863,  .506}0.0377 & 1.63\% & 1.63\% & 1.37\% & \cellcolor[rgb]{ .98,  .616,  .459}-0.4736 & \cellcolor[rgb]{ .996,  .816,  .498}0.1382 & \cellcolor[rgb]{ .996,  .875,  .506}0.0549 & \cellcolor[rgb]{ .427,  .757,  .486}0.9897 & \cellcolor[rgb]{ 1,  .914,  .518}1.0020 & \cellcolor[rgb]{ .871,  .886,  .514}0.0428\% & \cellcolor[rgb]{ .996,  .851,  .502}0.0418 & \multicolumn{1}{r|}{0.62\%} & \cellcolor[rgb]{ .992,  .816,  .494}0.0461 & \cellcolor[rgb]{ .973,  .412,  .42}0.0357 & \cellcolor[rgb]{ .996,  .875,  .506}1.1214 \\
    \hline
    \textbf{k-OWA-CVaR-5} & \cellcolor[rgb]{ .451,  .765,  .486}0.0537\% & \cellcolor[rgb]{ .98,  .49,  .435}1.2828\% & \multicolumn{1}{r|}{32.69\%} & \multicolumn{1}{r|}{27.32\%} & \cellcolor[rgb]{ .953,  .91,  .518}0.0419 & 0.40\% & 0.70\% & 0.49\% & \cellcolor[rgb]{ .996,  .851,  .502}-0.4521 & \cellcolor[rgb]{ 1,  .91,  .518}0.1262 & \cellcolor[rgb]{ .859,  .882,  .51}0.0612 & \cellcolor[rgb]{ .541,  .792,  .494}0.9829 & \cellcolor[rgb]{ 1,  .914,  .518}0.9786 & \cellcolor[rgb]{ .463,  .769,  .49}0.0475\% & \cellcolor[rgb]{ 1,  .922,  .518}0.0470 & \multicolumn{1}{r|}{0.25\%} & \cellcolor[rgb]{ .996,  .898,  .51}0.0520 & \cellcolor[rgb]{ .98,  .502,  .439}0.0351 & \cellcolor[rgb]{ .843,  .878,  .51}1.1357 \\
    \hline
    \textbf{k-OWA-CVaR-10} & \cellcolor[rgb]{ .424,  .757,  .486}0.0541\% & \cellcolor[rgb]{ .984,  .569,  .451}1.2622\% & \multicolumn{1}{r|}{47.62\%} & \multicolumn{1}{r|}{46.08\%} & \cellcolor[rgb]{ .827,  .875,  .51}0.0428 & 0.50\% & 0.50\% & 0.31\% & \cellcolor[rgb]{ .992,  .922,  .518}-0.4448 & \cellcolor[rgb]{ .996,  .918,  .514}0.1242 & \cellcolor[rgb]{ .694,  .835,  .502}0.0631 & \cellcolor[rgb]{ .451,  .765,  .486}0.9884 & \cellcolor[rgb]{ 1,  .918,  .518}0.9728 & \cellcolor[rgb]{ .424,  .757,  .486}0.0480\% & \cellcolor[rgb]{ .906,  .898,  .514}0.0480 & \multicolumn{1}{r|}{0.21\%} & \cellcolor[rgb]{ .996,  .918,  .514}0.0535 & \cellcolor[rgb]{ .984,  .627,  .463}0.0342 & \cellcolor[rgb]{ .71,  .839,  .502}1.1392 \\
    \hline
    \textbf{k-OWA-CVaR-25} & \cellcolor[rgb]{ .498,  .776,  .49}0.0531\% & \cellcolor[rgb]{ .996,  .827,  .502}1.1948\% & \multicolumn{1}{r|}{72.51\%} & \multicolumn{1}{r|}{73.04\%} & \cellcolor[rgb]{ .616,  .812,  .498}0.0445 & 0.10\% & 0.43\% & 0.14\% & \cellcolor[rgb]{ .996,  .878,  .506}-0.4497 & \cellcolor[rgb]{ .925,  .898,  .51}0.1203 & \cellcolor[rgb]{ .549,  .792,  .494}0.0647 & \cellcolor[rgb]{ .557,  .796,  .494}0.9821 & \cellcolor[rgb]{ 1,  .922,  .518}0.9084 & \cellcolor[rgb]{ .498,  .776,  .49}0.0471\% & \cellcolor[rgb]{ .627,  .816,  .498}0.0509 & \multicolumn{1}{r|}{0.12\%} & \cellcolor[rgb]{ .69,  .835,  .502}0.0579 & \cellcolor[rgb]{ .996,  .82,  .498}0.0329 & \cellcolor[rgb]{ .541,  .788,  .494}1.1436 \\
    \hline
    \textbf{k-OWA-CVaR-50} & \cellcolor[rgb]{ .992,  .922,  .518}0.0470\% & \cellcolor[rgb]{ .792,  .859,  .502}1.0698\% & \multicolumn{1}{r|}{0.03\%} & \multicolumn{1}{r|}{0.03\%} & \cellcolor[rgb]{ .69,  .835,  .502}0.0439 & 0.13\% & 0.23\% & 0.11\% & \cellcolor[rgb]{ .878,  .886,  .514}-0.4326 & \cellcolor[rgb]{ .804,  .863,  .506}0.1136 & \cellcolor[rgb]{ .776,  .859,  .506}0.0621 & \cellcolor[rgb]{ .996,  .906,  .514}0.9541 & \cellcolor[rgb]{ .91,  .894,  .51}0.8198 & \cellcolor[rgb]{ 1,  .922,  .518}0.0413\% & \cellcolor[rgb]{ .741,  .847,  .506}0.0497 & \multicolumn{1}{r|}{0.15\%} & \cellcolor[rgb]{ .471,  .769,  .49}0.0610 & \cellcolor[rgb]{ .784,  .859,  .502}0.0298 & \cellcolor[rgb]{ .765,  .855,  .506}1.1378 \\
    \hline
    \textbf{k-OWA-CVaR-75} & \cellcolor[rgb]{ .992,  .812,  .494}0.0385\% & \cellcolor[rgb]{ .655,  .82,  .494}1.0037\% & \multicolumn{1}{r|}{0.03\%} & \multicolumn{1}{r|}{0.03\%} & \cellcolor[rgb]{ .996,  .875,  .506}0.0384 & 0.47\% & 0.43\% & 0.42\% & \cellcolor[rgb]{ .851,  .878,  .51}-0.4295 & \cellcolor[rgb]{ .914,  .894,  .51}0.1196 & \cellcolor[rgb]{ .996,  .863,  .506}0.0539 & \cellcolor[rgb]{ .996,  .851,  .502}0.9458 & \cellcolor[rgb]{ .851,  .878,  .506}0.7601 & \cellcolor[rgb]{ .988,  .769,  .486}0.0331\% & \cellcolor[rgb]{ .992,  .835,  .498}0.0404 & \multicolumn{1}{r|}{0.78\%} & \cellcolor[rgb]{ 1,  .922,  .518}0.0536 & \cellcolor[rgb]{ .612,  .808,  .494}0.0279 & \cellcolor[rgb]{ .996,  .859,  .506}1.1187 \\
    \hline
    \textbf{k-OWA-CVaR-100} & \cellcolor[rgb]{ .992,  .808,  .494}0.0384\% & \cellcolor[rgb]{ .561,  .792,  .49}0.9591\% & \multicolumn{1}{r|}{0.03\%} & \multicolumn{1}{r|}{0.03\%} & \cellcolor[rgb]{ .996,  .898,  .51}0.0400 & 0.40\% & 0.37\% & 0.37\% & \cellcolor[rgb]{ .612,  .812,  .498}-0.4040 & \cellcolor[rgb]{ .663,  .824,  .498}0.1060 & \cellcolor[rgb]{ .996,  .886,  .51}0.0561 & \cellcolor[rgb]{ .992,  .816,  .494}0.9401 & \cellcolor[rgb]{ .8,  .863,  .506}0.7093 & \cellcolor[rgb]{ .988,  .773,  .486}0.0333\% & \cellcolor[rgb]{ .992,  .827,  .498}0.0397 & \multicolumn{1}{r|}{0.87\%} & \cellcolor[rgb]{ .863,  .882,  .51}0.0555 & \cellcolor[rgb]{ .478,  .769,  .486}0.0263 & \cellcolor[rgb]{ .996,  .882,  .51}1.1235 \\
    \hline
    \textbf{k-OWA-CCVaR-5} & \cellcolor[rgb]{ .443,  .765,  .486}0.0538\% & \cellcolor[rgb]{ .976,  .478,  .435}1.2853\% & \multicolumn{1}{r|}{30.09\%} & \multicolumn{1}{r|}{30.49\%} & \cellcolor[rgb]{ .953,  .91,  .518}0.0419 & 0.50\% & 0.50\% & 0.50\% & \cellcolor[rgb]{ .992,  .792,  .49}-0.4573 & \cellcolor[rgb]{ 1,  .89,  .514}0.1285 & \cellcolor[rgb]{ .847,  .878,  .51}0.0613 & \cellcolor[rgb]{ .388,  .745,  .482}0.9919 & \cellcolor[rgb]{ 1,  .914,  .518}0.9828 & \cellcolor[rgb]{ .455,  .765,  .486}0.0476\% & \cellcolor[rgb]{ .98,  .918,  .518}0.0472 & \multicolumn{1}{r|}{0.24\%} & \cellcolor[rgb]{ .996,  .898,  .51}0.0521 & \cellcolor[rgb]{ .98,  .529,  .443}0.0349 & \cellcolor[rgb]{ .847,  .878,  .51}1.1356 \\
    \hline
    \textbf{k-OWA-CCVaR-10} & \cellcolor[rgb]{ .388,  .745,  .482}0.0545\% & \cellcolor[rgb]{ .98,  .518,  .443}1.2753\% & \multicolumn{1}{r|}{36.75\%} & \multicolumn{1}{r|}{33.52\%} & \cellcolor[rgb]{ .839,  .878,  .51}0.0427 & 0.47\% & 0.40\% & 0.36\% & \cellcolor[rgb]{ .996,  .91,  .514}-0.4469 & \cellcolor[rgb]{ 1,  .922,  .518}0.1244 & \cellcolor[rgb]{ .737,  .847,  .506}0.0626 & \cellcolor[rgb]{ .549,  .792,  .494}0.9825 & \cellcolor[rgb]{ 1,  .914,  .518}0.9768 & \cellcolor[rgb]{ .388,  .745,  .482}0.0484\% & \cellcolor[rgb]{ .914,  .898,  .514}0.0479 & \multicolumn{1}{r|}{0.21\%} & \cellcolor[rgb]{ .996,  .914,  .514}0.0532 & \cellcolor[rgb]{ .98,  .506,  .439}0.0351 & \cellcolor[rgb]{ .733,  .847,  .506}1.1385 \\
    \hline
    \textbf{k-OWA-CCVaR-25} & \cellcolor[rgb]{ .424,  .757,  .486}0.0541\% & \cellcolor[rgb]{ .988,  .671,  .471}1.2351\% & \multicolumn{1}{r|}{69.51\%} & \multicolumn{1}{r|}{68.44\%} & \cellcolor[rgb]{ .706,  .839,  .502}0.0438 & 0.23\% & 0.40\% & 0.23\% & \cellcolor[rgb]{ .996,  .855,  .502}-0.4518 & \cellcolor[rgb]{ 1,  .91,  .518}0.1261 & \cellcolor[rgb]{ .6,  .808,  .498}0.0642 & \cellcolor[rgb]{ .478,  .773,  .49}0.9867 & \cellcolor[rgb]{ 1,  .922,  .518}0.9401 & \cellcolor[rgb]{ .42,  .757,  .486}0.0480\% & \cellcolor[rgb]{ .773,  .859,  .506}0.0494 & \multicolumn{1}{r|}{0.16\%} & \cellcolor[rgb]{ .867,  .882,  .51}0.0555 & \cellcolor[rgb]{ .988,  .671,  .471}0.0340 & \cellcolor[rgb]{ .604,  .808,  .498}1.1420 \\
    \hline
    \textbf{k-OWA-CCVaR-50} & \cellcolor[rgb]{ .518,  .784,  .49}0.0529\% & \cellcolor[rgb]{ .945,  .906,  .514}1.1448\% & \multicolumn{1}{r|}{7.53\%} & \multicolumn{1}{r|}{6.43\%} & \cellcolor[rgb]{ .388,  .745,  .482}0.0462 & 0.17\% & 0.10\% & 0.08\% & \cellcolor[rgb]{ .902,  .894,  .514}-0.4349 & \cellcolor[rgb]{ .769,  .855,  .502}0.1117 & \cellcolor[rgb]{ .388,  .745,  .482}0.0666 & \cellcolor[rgb]{ .667,  .827,  .502}0.9757 & \cellcolor[rgb]{ .961,  .91,  .514}0.8708 & \cellcolor[rgb]{ .51,  .78,  .49}0.0470\% & \cellcolor[rgb]{ .388,  .745,  .482}0.0535 & \multicolumn{1}{r|}{0.07\%} & \cellcolor[rgb]{ .388,  .745,  .482}0.0622 & \cellcolor[rgb]{ .937,  .902,  .514}0.0315 & \cellcolor[rgb]{ .388,  .745,  .482}1.1476 \\
    \hline
    \textbf{k-OWA-CCVaR-75} & \cellcolor[rgb]{ .996,  .918,  .514}0.0467\% & \cellcolor[rgb]{ .792,  .859,  .502}1.0705\% & \multicolumn{1}{r|}{0.03\%} & \multicolumn{1}{r|}{0.03\%} & \cellcolor[rgb]{ .725,  .843,  .502}0.0436 & 0.10\% & 0.20\% & 0.13\% & \cellcolor[rgb]{ .898,  .894,  .514}-0.4348 & \cellcolor[rgb]{ .835,  .871,  .506}0.1154 & \cellcolor[rgb]{ .808,  .867,  .51}0.0617 & \cellcolor[rgb]{ .996,  .91,  .514}0.9542 & \cellcolor[rgb]{ .902,  .89,  .51}0.8111 & \cellcolor[rgb]{ .996,  .914,  .514}0.0410\% & \cellcolor[rgb]{ .792,  .863,  .506}0.0492 & \multicolumn{1}{r|}{0.16\%} & \cellcolor[rgb]{ .514,  .784,  .49}0.0604 & \cellcolor[rgb]{ .831,  .871,  .506}0.0303 & \cellcolor[rgb]{ .8,  .867,  .51}1.1368 \\
    \hline
    \textbf{k-OWA-CCVaR-100} & \cellcolor[rgb]{ .996,  .851,  .502}0.0416\% & \cellcolor[rgb]{ .671,  .824,  .498}1.0116\% & \multicolumn{1}{r|}{0.03\%} & \multicolumn{1}{r|}{0.03\%} & \cellcolor[rgb]{ .996,  .914,  .514}0.0411 & 0.27\% & 0.40\% & 0.23\% & \cellcolor[rgb]{ .843,  .878,  .51}-0.4286 & \cellcolor[rgb]{ .875,  .882,  .51}0.1175 & \cellcolor[rgb]{ .996,  .902,  .514}0.0578 & \cellcolor[rgb]{ .996,  .867,  .506}0.9477 & \cellcolor[rgb]{ .851,  .878,  .506}0.7606 & \cellcolor[rgb]{ .992,  .827,  .498}0.0362\% & \cellcolor[rgb]{ .996,  .878,  .51}0.0439 & \multicolumn{1}{r|}{0.44\%} & \cellcolor[rgb]{ .737,  .847,  .506}0.0573 & \cellcolor[rgb]{ .635,  .816,  .494}0.0281 & \cellcolor[rgb]{ .996,  .902,  .514}1.1275 \\
    \hline
    \textbf{MinV} & \cellcolor[rgb]{ .992,  .78,  .49}0.0362\% & \cellcolor[rgb]{ .494,  .773,  .486}0.9254\% & \multicolumn{1}{r|}{0.03\%} & \multicolumn{1}{r|}{0.03\%} & \cellcolor[rgb]{ .996,  .886,  .51}0.0391 & 0.13\% & 0.33\% & 0.13\% & \cellcolor[rgb]{ .388,  .745,  .482}-0.3802 & \cellcolor[rgb]{ .388,  .745,  .482}0.0909 & \cellcolor[rgb]{ .996,  .867,  .506}0.0545 & \cellcolor[rgb]{ .973,  .914,  .518}0.9577 & \cellcolor[rgb]{ .678,  .827,  .498}0.5836 & \cellcolor[rgb]{ .988,  .729,  .478}0.0310\% & \cellcolor[rgb]{ .992,  .839,  .498}0.0407 & \multicolumn{1}{r|}{0.75\%} & \cellcolor[rgb]{ .584,  .804,  .494}0.0595 & \cellcolor[rgb]{ .541,  .788,  .49}0.0271 & \cellcolor[rgb]{ .996,  .875,  .506}1.1215 \\
    \hline
    \textbf{MinVsh} & \cellcolor[rgb]{ .984,  .651,  .463}0.0266\% & \cellcolor[rgb]{ .98,  .529,  .443}1.2719\% & \multicolumn{1}{r|}{13.53\%} & \multicolumn{1}{r|}{14.20\%} & \cellcolor[rgb]{ .98,  .616,  .459}0.0209 & 48.32\% & 49.05\% & 47.75\% & \cellcolor[rgb]{ .973,  .412,  .42}-0.4923 & \cellcolor[rgb]{ .976,  .451,  .427}0.1848 & \cellcolor[rgb]{ .98,  .608,  .455}0.0289 & \cellcolor[rgb]{ .992,  .827,  .498}0.9422 & \cellcolor[rgb]{ .973,  .412,  .42}5.1897 & \cellcolor[rgb]{ .98,  .604,  .455}0.0241\% & \cellcolor[rgb]{ .973,  .471,  .427}0.0123 & \multicolumn{1}{r|}{23.16\%} & \cellcolor[rgb]{ .973,  .447,  .424}0.0197 & \cellcolor[rgb]{ .98,  .486,  .435}0.0352 & \cellcolor[rgb]{ .98,  .592,  .451}1.0603 \\
    \hline
    \textbf{MinVN1} & \cellcolor[rgb]{ .976,  .522,  .439}0.0166\% & \cellcolor[rgb]{ .463,  .765,  .486}0.9100\% & \multicolumn{1}{r|}{0.03\%} & \multicolumn{1}{r|}{0.03\%} & \cellcolor[rgb]{ .98,  .576,  .451}0.0183 & 49.08\% & 48.52\% & 47.85\% & \cellcolor[rgb]{ .976,  .529,  .439}-0.4815 & \cellcolor[rgb]{ .973,  .412,  .42}0.1895 & \cellcolor[rgb]{ .98,  .565,  .447}0.0249 & \cellcolor[rgb]{ .98,  .584,  .451}0.9055 & \cellcolor[rgb]{ 1,  .882,  .514}1.2402 & \cellcolor[rgb]{ .973,  .412,  .42}0.0134\% & \cellcolor[rgb]{ .973,  .412,  .42}0.0077 & \multicolumn{1}{r|}{32.25\%} & \cellcolor[rgb]{ .973,  .412,  .42}0.0172 & \cellcolor[rgb]{ .475,  .769,  .486}0.0263 & \cellcolor[rgb]{ .98,  .565,  .447}1.0547 \\
    \hline
    \textbf{MinVN2} & \cellcolor[rgb]{ .98,  .557,  .447}0.0195\% & \cellcolor[rgb]{ .388,  .745,  .482}0.8738\% & \multicolumn{1}{r|}{0.03\%} & \multicolumn{1}{r|}{0.03\%} & \cellcolor[rgb]{ .984,  .635,  .463}0.0223 & 14.46\% & 14.43\% & 14.08\% & \cellcolor[rgb]{ .867,  .882,  .51}-0.4312 & \cellcolor[rgb]{ 1,  .851,  .506}0.1335 & \cellcolor[rgb]{ .98,  .62,  .459}0.0301 & \cellcolor[rgb]{ .973,  .412,  .42}0.8793 & \cellcolor[rgb]{ .388,  .745,  .482}0.2878 & \cellcolor[rgb]{ .973,  .431,  .424}0.0147\% & \cellcolor[rgb]{ .976,  .51,  .435}0.0153 & \multicolumn{1}{r|}{17.94\%} & \cellcolor[rgb]{ .98,  .565,  .447}0.0282 & \cellcolor[rgb]{ .388,  .745,  .482}0.0253 & \cellcolor[rgb]{ .98,  .62,  .459}1.0666 \\
    \hline
    \textbf{Index} & \cellcolor[rgb]{ .973,  .412,  .42}0.0083\% & \cellcolor[rgb]{ .992,  .761,  .49}1.2117\% &       &       & \cellcolor[rgb]{ .973,  .412,  .42}0.0069 &       &       &       & \cellcolor[rgb]{ .98,  .565,  .447}-0.4783 & \cellcolor[rgb]{ .988,  .682,  .475}0.1549 & \cellcolor[rgb]{ .973,  .412,  .42}0.0095 & \cellcolor[rgb]{ .988,  .718,  .478}0.9257 &       &       &       &       &       & \cellcolor[rgb]{ .984,  .588,  .455}0.0345 & \cellcolor[rgb]{ .973,  .412,  .42}1.0210 \\
    \hline
    \end{tabular}%
    }
    }
  \label{tab:Perf_out_FTSE100_DL125_DH20_C}%
\end{sidewaystable}%
%
%%%%%%%%%%%%%%%%%%%%%%%%%%%%%%%%%%%%%%%%%%%%%%%%%%%%
%
% Table generated by Excel2LaTeX from sheet 'test (3)'
\begin{sidewaystable}[htbp]
  \centering
  \caption{Out-of-sample performance results for NASDAQ100}
  \scalebox{0.63}{
  {\renewcommand{\arraystretch}{1.1}
            \begin{tabular}{|l|c|c|c|c|c|r|r|r|c|c|c|c|c|c|c|c|c|c|c|}
    \hline
    \textbf{Approach} & \textbf{ExpRet} & \multicolumn{3}{c|}{\textbf{Vol}} & \multicolumn{4}{c|}{\textbf{Sharpe}} & \textbf{MDD} & \textbf{Ulcer} & \textbf{Sortino} & \textbf{Rachev5} & \textbf{Turn} & \textbf{AlphaJ} & \multicolumn{2}{c|}{\textbf{InfoRatio}} & \textbf{ApprRatio} & \textbf{VaR1} & \textbf{Omega} \\
    \hline
          &       & \textbf{Value} & \multicolumn{1}{l|}{\textbf{p-v/b=1}} & \multicolumn{1}{l|}{\textbf{p-v/b=10}} & \textbf{value} & \multicolumn{1}{l|}{\textbf{p-v/b=1}} & \multicolumn{1}{l|}{\textbf{p-v/b=10}} & \multicolumn{1}{l|}{\textbf{p-v/HAC}} &       &       &       &       &       &       & \textbf{value} & \textbf{p-v} &       &       &  \\
    \hline
    \textbf{Roman-CVaR} & \cellcolor[rgb]{ .392,  .749,  .486}0.0770\% & \cellcolor[rgb]{ .973,  .412,  .42}1.4215\% & \multicolumn{1}{r|}{61.01\%} & \multicolumn{1}{r|}{63.48\%} & \cellcolor[rgb]{ .455,  .765,  .486}0.0541 & 44.05\% & 44.05\% & 43.66\% & \cellcolor[rgb]{ .996,  .855,  .502}-0.4041 & \cellcolor[rgb]{ 1,  .91,  .518}0.1135 & \cellcolor[rgb]{ .463,  .769,  .49}0.0780 & \cellcolor[rgb]{ .584,  .804,  .494}0.9674 & \cellcolor[rgb]{ 1,  .91,  .518}0.9220 & \cellcolor[rgb]{ .408,  .753,  .486}0.0238\% & \cellcolor[rgb]{ .4,  .749,  .486}0.0136 & \multicolumn{1}{r|}{20.76\%} & \cellcolor[rgb]{ .459,  .765,  .486}0.0286 & \cellcolor[rgb]{ .973,  .412,  .42}0.0407 & \cellcolor[rgb]{ .463,  .769,  .49}1.1733 \\
    \hline
    \textbf{k-OWA-CVaR-5} & \cellcolor[rgb]{ .388,  .745,  .482}0.0771\% & \cellcolor[rgb]{ .976,  .471,  .431}1.4016\% & \multicolumn{1}{r|}{73.18\%} & \multicolumn{1}{r|}{73.44\%} & \cellcolor[rgb]{ .388,  .745,  .482}0.0550 & 38.19\% & 36.49\% & 37.46\% & \cellcolor[rgb]{ .992,  .827,  .498}-0.4123 & \cellcolor[rgb]{ .992,  .918,  .514}0.1114 & \cellcolor[rgb]{ .388,  .745,  .482}0.0795 & \cellcolor[rgb]{ .388,  .745,  .482}0.9750 & \cellcolor[rgb]{ 1,  .914,  .518}0.9166 & \cellcolor[rgb]{ .388,  .745,  .482}0.0242\% & \cellcolor[rgb]{ .388,  .745,  .482}0.0141 & \multicolumn{1}{r|}{19.91\%} & \cellcolor[rgb]{ .388,  .745,  .482}0.0299 & \cellcolor[rgb]{ .984,  .588,  .455}0.0393 & \cellcolor[rgb]{ .388,  .745,  .482}1.1759 \\
    \hline
    \textbf{k-OWA-CVaR-10} & \cellcolor[rgb]{ .533,  .788,  .494}0.0721\% & \cellcolor[rgb]{ .98,  .549,  .447}1.3743\% & \multicolumn{1}{r|}{9.46\%} & \multicolumn{1}{r|}{10.13\%} & \cellcolor[rgb]{ .588,  .804,  .494}0.0524 & 51.88\% & 52.85\% & 51.81\% & \cellcolor[rgb]{ .996,  .871,  .506}-0.3992 & \cellcolor[rgb]{ 1,  .894,  .514}0.1155 & \cellcolor[rgb]{ .584,  .804,  .494}0.0756 & \cellcolor[rgb]{ .573,  .8,  .494}0.9679 & \cellcolor[rgb]{ 1,  .914,  .518}0.9196 & \cellcolor[rgb]{ .584,  .804,  .494}0.0201\% & \cellcolor[rgb]{ .518,  .784,  .49}0.0084 & \multicolumn{1}{r|}{30.81\%} & \cellcolor[rgb]{ .608,  .812,  .498}0.0255 & \cellcolor[rgb]{ .988,  .682,  .475}0.0386 & \cellcolor[rgb]{ .631,  .816,  .498}1.1672 \\
    \hline
    \textbf{k-OWA-CVaR-25} & \cellcolor[rgb]{ .784,  .863,  .506}0.0634\% & \cellcolor[rgb]{ .988,  .694,  .475}1.3221\% & \multicolumn{1}{r|}{0.03\%} & \multicolumn{1}{r|}{0.03\%} & \cellcolor[rgb]{ .945,  .906,  .518}0.0479 & 86.34\% & 85.87\% & 85.85\% & \cellcolor[rgb]{ .914,  .898,  .514}-0.3791 & \cellcolor[rgb]{ 1,  .922,  .518}0.1119 & \cellcolor[rgb]{ .941,  .906,  .518}0.0685 & \cellcolor[rgb]{ .922,  .902,  .514}0.9541 & \cellcolor[rgb]{ 1,  .922,  .518}0.8848 & \cellcolor[rgb]{ .914,  .898,  .514}0.0131\% & \cellcolor[rgb]{ .745,  .851,  .506}-0.0021 & \multicolumn{1}{r|}{54.96\%} & \cellcolor[rgb]{ 1,  .922,  .518}0.0175 & \cellcolor[rgb]{ .992,  .737,  .482}0.0381 & \cellcolor[rgb]{ 1,  .922,  .518}1.1538 \\
    \hline
    \textbf{k-OWA-CVaR-50} & \cellcolor[rgb]{ .996,  .871,  .506}0.0532\% & \cellcolor[rgb]{ .898,  .89,  .51}1.1964\% & \multicolumn{1}{r|}{0.03\%} & \multicolumn{1}{r|}{0.03\%} & \cellcolor[rgb]{ .996,  .855,  .502}0.0445 & 86.50\% & 87.17\% & 86.43\% & \cellcolor[rgb]{ .996,  .914,  .514}-0.3865 & \cellcolor[rgb]{ 1,  .906,  .518}0.1141 & \cellcolor[rgb]{ .996,  .851,  .502}0.0631 & \cellcolor[rgb]{ .984,  .682,  .471}0.9322 & \cellcolor[rgb]{ .906,  .894,  .51}0.7893 & \cellcolor[rgb]{ .976,  .486,  .431}0.0076\% & \cellcolor[rgb]{ .996,  .875,  .506}-0.0149 & \multicolumn{1}{r|}{81.33\%} & \cellcolor[rgb]{ .98,  .624,  .459}0.0113 & \cellcolor[rgb]{ .878,  .886,  .51}0.0348 & \cellcolor[rgb]{ .996,  .851,  .502}1.1436 \\
    \hline
    \textbf{k-OWA-CVaR-75} & \cellcolor[rgb]{ .992,  .824,  .498}0.0509\% & \cellcolor[rgb]{ .698,  .831,  .498}1.1073\% & \multicolumn{1}{r|}{0.03\%} & \multicolumn{1}{r|}{0.03\%} & \cellcolor[rgb]{ .996,  .89,  .51}0.0459 & 98.33\% & 98.20\% & 98.19\% & \cellcolor[rgb]{ .514,  .784,  .49}-0.3558 & \cellcolor[rgb]{ .624,  .812,  .494}0.0945 & \cellcolor[rgb]{ .996,  .886,  .51}0.0652 & \cellcolor[rgb]{ .984,  .667,  .467}0.9310 & \cellcolor[rgb]{ .816,  .867,  .506}0.6941 & \cellcolor[rgb]{ .984,  .659,  .467}0.0091\% & \cellcolor[rgb]{ .988,  .769,  .486}-0.0175 & \multicolumn{1}{r|}{85.20\%} & \cellcolor[rgb]{ .988,  .765,  .486}0.0142 & \cellcolor[rgb]{ .667,  .824,  .498}0.0315 & \cellcolor[rgb]{ .996,  .89,  .51}1.1493 \\
    \hline
    \textbf{k-OWA-CVaR-100} & \cellcolor[rgb]{ .992,  .816,  .494}0.0504\% & \cellcolor[rgb]{ .553,  .792,  .49}1.0427\% & \multicolumn{1}{r|}{0.03\%} & \multicolumn{1}{r|}{0.03\%} & \cellcolor[rgb]{ .91,  .898,  .514}0.0484 & 84.34\% & 84.11\% & 83.32\% & \cellcolor[rgb]{ .408,  .753,  .486}-0.3495 & \cellcolor[rgb]{ .388,  .745,  .482}0.0836 & \cellcolor[rgb]{ .918,  .898,  .514}0.0690 & \cellcolor[rgb]{ .929,  .902,  .514}0.9538 & \cellcolor[rgb]{ .729,  .843,  .502}0.6033 & \cellcolor[rgb]{ .969,  .914,  .518}0.0120\% & \cellcolor[rgb]{ .992,  .78,  .49}-0.0172 & \multicolumn{1}{r|}{84.80\%} & \cellcolor[rgb]{ .922,  .902,  .514}0.0191 & \cellcolor[rgb]{ .58,  .8,  .49}0.0301 & \cellcolor[rgb]{ .902,  .894,  .514}1.1574 \\
    \hline
    \textbf{k-OWA-CCVaR-5} & \cellcolor[rgb]{ .404,  .749,  .486}0.0766\% & \cellcolor[rgb]{ .976,  .451,  .427}1.4083\% & \multicolumn{1}{r|}{96.57\%} & \multicolumn{1}{r|}{96.70\%} & \cellcolor[rgb]{ .435,  .761,  .486}0.0544 & 40.85\% & 38.69\% & 41.01\% & \cellcolor[rgb]{ .992,  .835,  .498}-0.4092 & \cellcolor[rgb]{ .965,  .91,  .514}0.1102 & \cellcolor[rgb]{ .431,  .761,  .486}0.0786 & \cellcolor[rgb]{ .408,  .753,  .486}0.9743 & \cellcolor[rgb]{ 1,  .914,  .518}0.9178 & \cellcolor[rgb]{ .416,  .753,  .486}0.0236\% & \cellcolor[rgb]{ .404,  .753,  .486}0.0135 & \multicolumn{1}{r|}{20.99\%} & \cellcolor[rgb]{ .439,  .761,  .486}0.0290 & \cellcolor[rgb]{ .98,  .561,  .451}0.0396 & \cellcolor[rgb]{ .443,  .761,  .486}1.1739 \\
    \hline
    \textbf{k-OWA-CCVaR-10} & \cellcolor[rgb]{ .549,  .792,  .494}0.0716\% & \cellcolor[rgb]{ .98,  .498,  .439}1.3912\% & \multicolumn{1}{r|}{40.45\%} & \multicolumn{1}{r|}{41.05\%} & \cellcolor[rgb]{ .671,  .827,  .502}0.0514 & 57.28\% & 60.61\% & 59.67\% & \cellcolor[rgb]{ .992,  .831,  .498}-0.4106 & \cellcolor[rgb]{ 1,  .867,  .51}0.1192 & \cellcolor[rgb]{ .663,  .827,  .502}0.0740 & \cellcolor[rgb]{ .647,  .82,  .498}0.9649 & \cellcolor[rgb]{ 1,  .91,  .518}0.9259 & \cellcolor[rgb]{ .631,  .816,  .498}0.0191\% & \cellcolor[rgb]{ .533,  .788,  .494}0.0077 & \multicolumn{1}{r|}{32.35\%} & \cellcolor[rgb]{ .69,  .835,  .502}0.0238 & \cellcolor[rgb]{ .988,  .631,  .463}0.0390 & \cellcolor[rgb]{ .722,  .843,  .502}1.1639 \\
    \hline
    \textbf{k-OWA-CCVaR-25} & \cellcolor[rgb]{ .584,  .804,  .494}0.0703\% & \cellcolor[rgb]{ .984,  .612,  .459}1.3520\% & \multicolumn{1}{r|}{0.63\%} & \multicolumn{1}{r|}{0.47\%} & \cellcolor[rgb]{ .627,  .816,  .498}0.0520 & 56.51\% & 56.58\% & 55.55\% & \cellcolor[rgb]{ .855,  .882,  .51}-0.3757 & \cellcolor[rgb]{ .988,  .918,  .514}0.1113 & \cellcolor[rgb]{ .627,  .816,  .498}0.0748 & \cellcolor[rgb]{ .808,  .867,  .51}0.9586 & \cellcolor[rgb]{ 1,  .922,  .518}0.8920 & \cellcolor[rgb]{ .627,  .816,  .498}0.0192\% & \cellcolor[rgb]{ .561,  .796,  .494}0.0063 & \multicolumn{1}{r|}{35.37\%} & \cellcolor[rgb]{ .647,  .82,  .498}0.0247 & \cellcolor[rgb]{ .988,  .651,  .467}0.0388 & \cellcolor[rgb]{ .655,  .824,  .498}1.1664 \\
    \hline
    \textbf{k-OWA-CCVaR-50} & \cellcolor[rgb]{ .949,  .91,  .518}0.0577\% & \cellcolor[rgb]{ .996,  .8,  .494}1.2852\% & \multicolumn{1}{r|}{0.03\%} & \multicolumn{1}{r|}{0.03\%} & \cellcolor[rgb]{ .996,  .863,  .506}0.0449 & 89.30\% & 90.10\% & 89.56\% & \cellcolor[rgb]{ .996,  .886,  .51}-0.3936 & \cellcolor[rgb]{ 1,  .851,  .506}0.1208 & \cellcolor[rgb]{ .996,  .859,  .506}0.0636 & \cellcolor[rgb]{ .992,  .776,  .49}0.9397 & \cellcolor[rgb]{ .976,  .914,  .514}0.8634 & \cellcolor[rgb]{ .98,  .616,  .459}0.0087\% & \cellcolor[rgb]{ .902,  .894,  .514}-0.0092 & \multicolumn{1}{r|}{70.87\%} & \cellcolor[rgb]{ .984,  .655,  .467}0.0120 & \cellcolor[rgb]{ .992,  .706,  .478}0.0384 & \cellcolor[rgb]{ .996,  .851,  .502}1.1440 \\
    \hline
    \textbf{k-OWA-CCVaR-75} & \cellcolor[rgb]{ .996,  .882,  .51}0.0539\% & \cellcolor[rgb]{ .898,  .89,  .51}1.1974\% & \multicolumn{1}{r|}{0.03\%} & \multicolumn{1}{r|}{0.03\%} & \cellcolor[rgb]{ .996,  .867,  .506}0.0450 & 91.24\% & 91.60\% & 91.09\% & \cellcolor[rgb]{ .961,  .914,  .518}-0.3818 & \cellcolor[rgb]{ .992,  .918,  .514}0.1114 & \cellcolor[rgb]{ .996,  .867,  .506}0.0640 & \cellcolor[rgb]{ .988,  .718,  .478}0.9351 & \cellcolor[rgb]{ .902,  .89,  .51}0.7834 & \cellcolor[rgb]{ .98,  .596,  .455}0.0085\% & \cellcolor[rgb]{ 1,  .922,  .518}-0.0138 & \multicolumn{1}{r|}{79.59\%} & \cellcolor[rgb]{ .984,  .678,  .471}0.0125 & \cellcolor[rgb]{ .902,  .89,  .51}0.0351 & \cellcolor[rgb]{ .996,  .867,  .506}1.1458 \\
    \hline
    \textbf{k-OWA-CCVaR-100} & \cellcolor[rgb]{ .996,  .847,  .502}0.0521\% & \cellcolor[rgb]{ .729,  .843,  .502}1.1206\% & \multicolumn{1}{r|}{0.03\%} & \multicolumn{1}{r|}{0.03\%} & \cellcolor[rgb]{ .996,  .902,  .514}0.0465 & 97.63\% & 98.00\% & 97.72\% & \cellcolor[rgb]{ .584,  .804,  .494}-0.3598 & \cellcolor[rgb]{ .643,  .816,  .494}0.0954 & \cellcolor[rgb]{ .996,  .898,  .51}0.0660 & \cellcolor[rgb]{ .988,  .706,  .475}0.9341 & \cellcolor[rgb]{ .824,  .871,  .506}0.7000 & \cellcolor[rgb]{ .988,  .761,  .486}0.0099\% & \cellcolor[rgb]{ .992,  .831,  .498}-0.0159 & \multicolumn{1}{r|}{82.92\%} & \cellcolor[rgb]{ .992,  .812,  .494}0.0153 & \cellcolor[rgb]{ .706,  .835,  .498}0.0321 & \cellcolor[rgb]{ .996,  .902,  .514}1.1510 \\
    \hline
    \textbf{MinV} & \cellcolor[rgb]{ .988,  .741,  .482}0.0465\% & \cellcolor[rgb]{ .459,  .765,  .486}0.9999\% & \multicolumn{1}{r|}{0.03\%} & \multicolumn{1}{r|}{0.03\%} & \cellcolor[rgb]{ .996,  .902,  .514}0.0465 & 97.87\% & 97.90\% & 97.57\% & \cellcolor[rgb]{ .388,  .745,  .482}-0.3485 & \cellcolor[rgb]{ .451,  .761,  .482}0.0865 & \cellcolor[rgb]{ .996,  .898,  .514}0.0661 & \cellcolor[rgb]{ .988,  .725,  .478}0.9357 & \cellcolor[rgb]{ .573,  .796,  .49}0.4399 & \cellcolor[rgb]{ .988,  .757,  .482}0.0099\% & \cellcolor[rgb]{ .98,  .6,  .455}-0.0214 & \multicolumn{1}{r|}{90.00\%} & \cellcolor[rgb]{ .996,  .859,  .506}0.0162 & \cellcolor[rgb]{ .478,  .769,  .486}0.0286 & \cellcolor[rgb]{ .996,  .918,  .514}1.1537 \\
    \hline
    \textbf{MinVsh} & \cellcolor[rgb]{ .973,  .412,  .42}0.0293\% & \cellcolor[rgb]{ .737,  .843,  .502}1.1254\% & \multicolumn{1}{r|}{0.03\%} & \multicolumn{1}{r|}{0.03\%} & \cellcolor[rgb]{ .973,  .412,  .42}0.0260 & 24.83\% & 24.36\% & 24.12\% & \cellcolor[rgb]{ .992,  .835,  .498}-0.4100 & \cellcolor[rgb]{ .973,  .412,  .42}0.1767 & \cellcolor[rgb]{ .973,  .412,  .42}0.0365 & \cellcolor[rgb]{ .894,  .89,  .514}0.9553 & \cellcolor[rgb]{ .973,  .412,  .42}2.4534 & \cellcolor[rgb]{ .973,  .412,  .42}0.0070\% & \cellcolor[rgb]{ .973,  .412,  .42}-0.0260 & \multicolumn{1}{r|}{94.00\%} & \cellcolor[rgb]{ .973,  .412,  .42}0.0069 & \cellcolor[rgb]{ .612,  .808,  .494}0.0307 & \cellcolor[rgb]{ .973,  .412,  .42}1.0801 \\
    \hline
    \textbf{MinVN1} & \cellcolor[rgb]{ .98,  .561,  .447}0.0372\% & \cellcolor[rgb]{ .388,  .745,  .482}0.9681\% & \multicolumn{1}{r|}{0.03\%} & \multicolumn{1}{r|}{0.03\%} & \cellcolor[rgb]{ .988,  .71,  .475}0.0385 & 63.78\% & 61.55\% & 60.89\% & \cellcolor[rgb]{ .62,  .812,  .498}-0.3618 & \cellcolor[rgb]{ 1,  .906,  .518}0.1140 & \cellcolor[rgb]{ .988,  .71,  .475}0.0546 & \cellcolor[rgb]{ .992,  .792,  .49}0.9410 & \cellcolor[rgb]{ 1,  .878,  .51}1.0237 & \cellcolor[rgb]{ 1,  .922,  .518}0.0113\% & \cellcolor[rgb]{ .976,  .49,  .431}-0.0241 & \multicolumn{1}{r|}{92.50\%} & \cellcolor[rgb]{ .988,  .765,  .486}0.0143 & \cellcolor[rgb]{ .388,  .745,  .482}0.0271 & \cellcolor[rgb]{ .988,  .71,  .475}1.1232 \\
    \hline
    \textbf{MinVN2} & \cellcolor[rgb]{ .992,  .82,  .498}0.0506\% & \cellcolor[rgb]{ .518,  .78,  .486}1.0264\% & \multicolumn{1}{r|}{0.03\%} & \multicolumn{1}{r|}{0.03\%} & \cellcolor[rgb]{ .835,  .875,  .51}0.0493 & 71.11\% & 72.88\% & 71.64\% & \cellcolor[rgb]{ .922,  .902,  .514}-0.3795 & \cellcolor[rgb]{ .941,  .902,  .514}0.1090 & \cellcolor[rgb]{ .882,  .89,  .514}0.0697 & \cellcolor[rgb]{ .973,  .412,  .42}0.9108 & \cellcolor[rgb]{ .388,  .745,  .482}0.2424 & \cellcolor[rgb]{ .988,  .722,  .478}0.0096\% & \cellcolor[rgb]{ .984,  .671,  .467}-0.0197 & \multicolumn{1}{r|}{88.10\%} & \cellcolor[rgb]{ .941,  .906,  .518}0.0187 & \cellcolor[rgb]{ .553,  .792,  .49}0.0297 & \cellcolor[rgb]{ .745,  .851,  .506}1.1630 \\
    \hline
    \textbf{Index} & \cellcolor[rgb]{ .737,  .847,  .506}0.0651\% & \cellcolor[rgb]{ .976,  .447,  .427}1.4093\% &       &       & \cellcolor[rgb]{ .996,  .894,  .51}0.0462 &       &       &       & \cellcolor[rgb]{ .973,  .412,  .42}-0.5371 & \cellcolor[rgb]{ .992,  .749,  .486}0.1339 & \cellcolor[rgb]{ .996,  .886,  .51}0.0654 & \cellcolor[rgb]{ .996,  .882,  .51}0.9481 &       &       &       &       &       & \cellcolor[rgb]{ .976,  .424,  .424}0.0407 & \cellcolor[rgb]{ .996,  .91,  .514}1.1523 \\
    \hline
    \end{tabular}%
    }
    }
  \label{tab:Perf_out_NASDAQ100_DL125_DH20_C}%
\end{sidewaystable}%
%
%
%%%%%%%%%%%%%%%%%%%%%%%%%%%%%%%%%%%%%%%%%%%%%%%%%%%%%%%%%%%%
%
% Table generated by Excel2LaTeX from sheet 'test (4)'
\begin{sidewaystable}[htbp]
  \centering
  \caption{Out-of-sample performance results for SP500}
  \scalebox{0.63}{
  {\renewcommand{\arraystretch}{1.1}
    \begin{tabular}{|l|c|c|c|c|c|r|r|r|c|c|c|c|c|c|c|c|c|c|c|}
    \hline
    \textbf{Approach} & \textbf{ExpRet} & \multicolumn{3}{c|}{\textbf{Vol}} & \multicolumn{4}{c|}{\textbf{Sharpe}} & \textbf{MDD} & \textbf{Ulcer} & \textbf{Sortino} & \textbf{Rachev5} & \textbf{Turn} & \textbf{AlphaJ} & \multicolumn{2}{c|}{\textbf{InfoRatio}} & \textbf{ApprRatio} & \textbf{VaR1} & \textbf{Omega} \\
    \hline
          &       & \textbf{value} & \multicolumn{1}{l|}{\textbf{p-v/b=1}} & \multicolumn{1}{l|}{\textbf{p-v/b=10}} & \textbf{value} & \multicolumn{1}{l|}{\textbf{p-v/b=1}} & \multicolumn{1}{l|}{\textbf{p-v/b=10}} & \multicolumn{1}{l|}{\textbf{p-v/HAC}} &       &       &       &       &       &       & \textbf{value} & \textbf{p-v} &       &       &  \\
    \hline
    \textbf{Roman-CVaR} & \cellcolor[rgb]{ .502,  .78,  .49}0.0420\% & \cellcolor[rgb]{ .976,  .431,  .424}1.2895\% & \multicolumn{1}{r|}{79.91\%} & \multicolumn{1}{r|}{81.24\%} & \cellcolor[rgb]{ .933,  .902,  .514}0.0326 & 61.15\% & 60.05\% & 60.91\% & \cellcolor[rgb]{ .839,  .878,  .51}-0.4693 & \cellcolor[rgb]{ .427,  .757,  .482}0.1429 & \cellcolor[rgb]{ .898,  .894,  .514}0.0457 & \cellcolor[rgb]{ .655,  .824,  .498}0.9223 & \cellcolor[rgb]{ 1,  .91,  .518}1.1162 & \cellcolor[rgb]{ .729,  .847,  .506}0.0142\% & \cellcolor[rgb]{ .498,  .78,  .49}0.0081 & \multicolumn{1}{r|}{31.37\%} & \cellcolor[rgb]{ .937,  .906,  .518}0.0181 & \cellcolor[rgb]{ .976,  .475,  .431}0.0384 & \cellcolor[rgb]{ .992,  .922,  .518}1.1011 \\
    \hline
    \textbf{k-OWA-CVaR-5} & \cellcolor[rgb]{ .388,  .745,  .482}0.0436\% & \cellcolor[rgb]{ .976,  .478,  .435}1.2707\% & \multicolumn{1}{r|}{32.59\%} & \multicolumn{1}{r|}{32.76\%} & \cellcolor[rgb]{ .408,  .753,  .486}0.0343 & 50.35\% & 50.95\% & 49.53\% & \cellcolor[rgb]{ .961,  .914,  .518}-0.4837 & \cellcolor[rgb]{ .565,  .796,  .49}0.1483 & \cellcolor[rgb]{ .416,  .753,  .486}0.0482 & \cellcolor[rgb]{ .678,  .831,  .502}0.9214 & \cellcolor[rgb]{ 1,  .918,  .518}1.0794 & \cellcolor[rgb]{ .396,  .749,  .486}0.0162\% & \cellcolor[rgb]{ .388,  .745,  .482}0.0102 & \multicolumn{1}{r|}{27.16\%} & \cellcolor[rgb]{ .408,  .753,  .486}0.0210 & \cellcolor[rgb]{ .976,  .471,  .431}0.0385 & \cellcolor[rgb]{ .604,  .808,  .498}1.1069 \\
    \hline
    \textbf{k-OWA-CVaR-10} & \cellcolor[rgb]{ .533,  .788,  .494}0.0415\% & \cellcolor[rgb]{ .98,  .529,  .443}1.2505\% & \multicolumn{1}{r|}{6.70\%} & \multicolumn{1}{r|}{7.70\%} & \cellcolor[rgb]{ .737,  .847,  .506}0.0332 & 56.61\% & 56.41\% & 56.58\% & \cellcolor[rgb]{ .992,  .816,  .494}-0.5044 & \cellcolor[rgb]{ .945,  .906,  .514}0.1636 & \cellcolor[rgb]{ .702,  .835,  .502}0.0467 & \cellcolor[rgb]{ .576,  .8,  .494}0.9252 & \cellcolor[rgb]{ 1,  .918,  .518}1.0823 & \cellcolor[rgb]{ .647,  .824,  .498}0.0147\% & \cellcolor[rgb]{ .525,  .784,  .49}0.0076 & \multicolumn{1}{r|}{32.43\%} & \cellcolor[rgb]{ .737,  .847,  .506}0.0192 & \cellcolor[rgb]{ .98,  .51,  .439}0.0380 & \cellcolor[rgb]{ .863,  .882,  .51}1.1030 \\
    \hline
    \textbf{k-OWA-CVaR-25} & \cellcolor[rgb]{ .792,  .863,  .506}0.0378\% & \cellcolor[rgb]{ .992,  .745,  .486}1.1674\% & \multicolumn{1}{r|}{0.03\%} & \multicolumn{1}{r|}{0.03\%} & \cellcolor[rgb]{ .996,  .922,  .518}0.0323 & 62.51\% & 63.05\% & 61.13\% & \cellcolor[rgb]{ .996,  .867,  .506}-0.4963 & \cellcolor[rgb]{ .996,  .847,  .506}0.1744 & \cellcolor[rgb]{ .961,  .914,  .518}0.0454 & \cellcolor[rgb]{ .82,  .871,  .51}0.9162 & \cellcolor[rgb]{ 1,  .922,  .518}1.0596 & \cellcolor[rgb]{ 1,  .922,  .518}0.0126\% & \cellcolor[rgb]{ .765,  .855,  .506}0.0032 & \multicolumn{1}{r|}{42.52\%} & \cellcolor[rgb]{ 1,  .922,  .518}0.0177 & \cellcolor[rgb]{ .996,  .812,  .498}0.0337 & \cellcolor[rgb]{ .996,  .894,  .51}1.0995 \\
    \hline
    \textbf{k-OWA-CVaR-50} & \cellcolor[rgb]{ .996,  .882,  .51}0.0335\% & \cellcolor[rgb]{ .925,  .898,  .51}1.0643\% & \multicolumn{1}{r|}{0.03\%} & \multicolumn{1}{r|}{0.03\%} & \cellcolor[rgb]{ .996,  .878,  .506}0.0315 & 68.18\% & 68.74\% & 67.09\% & \cellcolor[rgb]{ .996,  .902,  .514}-0.4913 & \cellcolor[rgb]{ .992,  .725,  .482}0.1886 & \cellcolor[rgb]{ .996,  .871,  .506}0.0438 & \cellcolor[rgb]{ .973,  .482,  .431}0.8816 & \cellcolor[rgb]{ .91,  .894,  .51}1.0046 & \cellcolor[rgb]{ .988,  .749,  .482}0.0106\% & \cellcolor[rgb]{ .996,  .886,  .51}-0.0022 & \multicolumn{1}{r|}{55.16\%} & \cellcolor[rgb]{ .992,  .835,  .498}0.0163 & \cellcolor[rgb]{ .945,  .906,  .514}0.0315 & \cellcolor[rgb]{ .996,  .875,  .506}1.0984 \\
    \hline
    \textbf{k-OWA-CVaR-75} & \cellcolor[rgb]{ .992,  .78,  .49}0.0304\% & \cellcolor[rgb]{ .729,  .843,  .502}0.9740\% & \multicolumn{1}{r|}{0.03\%} & \multicolumn{1}{r|}{0.03\%} & \cellcolor[rgb]{ .996,  .863,  .506}0.0312 & 70.21\% & 70.61\% & 68.98\% & \cellcolor[rgb]{ .918,  .898,  .514}-0.4783 & \cellcolor[rgb]{ 1,  .859,  .506}0.1732 & \cellcolor[rgb]{ .996,  .855,  .502}0.0433 & \cellcolor[rgb]{ .973,  .467,  .427}0.8806 & \cellcolor[rgb]{ .804,  .863,  .506}0.9382 & \cellcolor[rgb]{ .984,  .635,  .463}0.0092\% & \cellcolor[rgb]{ .988,  .737,  .478}-0.0062 & \multicolumn{1}{r|}{64.47\%} & \cellcolor[rgb]{ .992,  .8,  .494}0.0158 & \cellcolor[rgb]{ .725,  .839,  .498}0.0287 & \cellcolor[rgb]{ .996,  .878,  .51}1.0987 \\
    \hline
    \textbf{k-OWA-CVaR-100} & \cellcolor[rgb]{ .992,  .808,  .494}0.0312\% & \cellcolor[rgb]{ .608,  .808,  .494}0.9165\% & \multicolumn{1}{r|}{0.03\%} & \multicolumn{1}{r|}{0.03\%} & \cellcolor[rgb]{ .478,  .773,  .49}0.0341 & 51.95\% & 51.15\% & 50.57\% & \cellcolor[rgb]{ .996,  .859,  .506}-0.4976 & \cellcolor[rgb]{ .855,  .878,  .506}0.1599 & \cellcolor[rgb]{ .549,  .792,  .494}0.0475 & \cellcolor[rgb]{ .98,  .62,  .459}0.8905 & \cellcolor[rgb]{ .694,  .831,  .498}0.8700 & \cellcolor[rgb]{ .992,  .835,  .498}0.0116\% & \cellcolor[rgb]{ .992,  .784,  .49}-0.0050 & \multicolumn{1}{r|}{61.65\%} & \cellcolor[rgb]{ .478,  .773,  .49}0.0206 & \cellcolor[rgb]{ .631,  .812,  .494}0.0275 & \cellcolor[rgb]{ .388,  .745,  .482}1.1101 \\
    \hline
    \textbf{k-OWA-CCVaR-5} & \cellcolor[rgb]{ .443,  .765,  .486}0.0428\% & \cellcolor[rgb]{ .976,  .467,  .431}1.2754\% & \multicolumn{1}{r|}{42.75\%} & \multicolumn{1}{r|}{43.85\%} & \cellcolor[rgb]{ .631,  .816,  .498}0.0336 & 54.65\% & 55.88\% & 54.67\% & \cellcolor[rgb]{ .894,  .894,  .514}-0.4756 & \cellcolor[rgb]{ .463,  .765,  .486}0.1444 & \cellcolor[rgb]{ .631,  .816,  .498}0.0471 & \cellcolor[rgb]{ .835,  .875,  .51}0.9156 & \cellcolor[rgb]{ 1,  .918,  .518}1.0895 & \cellcolor[rgb]{ .533,  .788,  .494}0.0154\% & \cellcolor[rgb]{ .447,  .765,  .486}0.0091 & \multicolumn{1}{r|}{29.24\%} & \cellcolor[rgb]{ .631,  .816,  .498}0.0197 & \cellcolor[rgb]{ .98,  .49,  .435}0.0382 & \cellcolor[rgb]{ .765,  .855,  .506}1.1045 \\
    \hline
    \textbf{k-OWA-CCVaR-10} & \cellcolor[rgb]{ .412,  .753,  .486}0.0433\% & \cellcolor[rgb]{ .98,  .51,  .439}1.2593\% & \multicolumn{1}{r|}{15.19\%} & \multicolumn{1}{r|}{16.59\%} & \cellcolor[rgb]{ .388,  .745,  .482}0.0344 & 49.58\% & 49.35\% & 49.57\% & \cellcolor[rgb]{ .98,  .918,  .518}-0.4857 & \cellcolor[rgb]{ .588,  .8,  .49}0.1494 & \cellcolor[rgb]{ .388,  .745,  .482}0.0483 & \cellcolor[rgb]{ .639,  .82,  .498}0.9229 & \cellcolor[rgb]{ 1,  .918,  .518}1.0782 & \cellcolor[rgb]{ .388,  .745,  .482}0.0162\% & \cellcolor[rgb]{ .412,  .753,  .486}0.0097 & \multicolumn{1}{r|}{28.01\%} & \cellcolor[rgb]{ .388,  .745,  .482}0.0210 & \cellcolor[rgb]{ .984,  .561,  .451}0.0372 & \cellcolor[rgb]{ .608,  .808,  .498}1.1069 \\
    \hline
    \textbf{k-OWA-CCVaR-25} & \cellcolor[rgb]{ .647,  .82,  .498}0.0399\% & \cellcolor[rgb]{ .984,  .627,  .463}1.2137\% & \multicolumn{1}{r|}{0.20\%} & \multicolumn{1}{r|}{0.53\%} & \cellcolor[rgb]{ .847,  .878,  .51}0.0328 & 58.41\% & 58.01\% & 58.32\% & \cellcolor[rgb]{ .992,  .843,  .502}-0.5004 & \cellcolor[rgb]{ .816,  .867,  .506}0.1584 & \cellcolor[rgb]{ .808,  .867,  .51}0.0462 & \cellcolor[rgb]{ .408,  .753,  .486}0.9314 & \cellcolor[rgb]{ 1,  .922,  .518}1.0711 & \cellcolor[rgb]{ .808,  .867,  .51}0.0138\% & \cellcolor[rgb]{ .631,  .816,  .498}0.0057 & \multicolumn{1}{r|}{36.69\%} & \cellcolor[rgb]{ .851,  .878,  .51}0.0185 & \cellcolor[rgb]{ .992,  .757,  .486}0.0345 & \cellcolor[rgb]{ .984,  .918,  .518}1.1012 \\
    \hline
    \textbf{k-OWA-CCVaR-50} & \cellcolor[rgb]{ .788,  .863,  .506}0.0378\% & \cellcolor[rgb]{ .996,  .835,  .502}1.1338\% & \multicolumn{1}{r|}{0.03\%} & \multicolumn{1}{r|}{0.03\%} & \cellcolor[rgb]{ .694,  .835,  .502}0.0334 & 55.98\% & 55.61\% & 54.61\% & \cellcolor[rgb]{ .992,  .843,  .502}-0.5000 & \cellcolor[rgb]{ 1,  .867,  .51}0.1722 & \cellcolor[rgb]{ .71,  .839,  .502}0.0467 & \cellcolor[rgb]{ .992,  .847,  .502}0.9047 & \cellcolor[rgb]{ .961,  .91,  .514}1.0367 & \cellcolor[rgb]{ .863,  .882,  .51}0.0135\% & \cellcolor[rgb]{ .761,  .855,  .506}0.0032 & \multicolumn{1}{r|}{42.31\%} & \cellcolor[rgb]{ .694,  .835,  .502}0.0194 & \cellcolor[rgb]{ 1,  .878,  .51}0.0328 & \cellcolor[rgb]{ .863,  .882,  .51}1.1031 \\
    \hline
    \textbf{k-OWA-CCVaR-75} & \cellcolor[rgb]{ .996,  .902,  .514}0.0342\% & \cellcolor[rgb]{ .91,  .894,  .51}1.0587\% & \multicolumn{1}{r|}{0.03\%} & \multicolumn{1}{r|}{0.03\%} & \cellcolor[rgb]{ .996,  .918,  .514}0.0323 & 62.18\% & 60.11\% & 61.17\% & \cellcolor[rgb]{ .996,  .902,  .514}-0.4908 & \cellcolor[rgb]{ .992,  .757,  .486}0.1848 & \cellcolor[rgb]{ .996,  .914,  .514}0.0450 & \cellcolor[rgb]{ .976,  .49,  .431}0.8822 & \cellcolor[rgb]{ .89,  .89,  .51}0.9932 & \cellcolor[rgb]{ .992,  .816,  .494}0.0114\% & \cellcolor[rgb]{ 1,  .922,  .518}-0.0013 & \multicolumn{1}{r|}{53.07\%} & \cellcolor[rgb]{ .996,  .918,  .514}0.0177 & \cellcolor[rgb]{ .945,  .902,  .514}0.0314 & \cellcolor[rgb]{ .996,  .918,  .514}1.1008 \\
    \hline
    \textbf{k-OWA-CCVaR-100} & \cellcolor[rgb]{ .992,  .792,  .49}0.0307\% & \cellcolor[rgb]{ .749,  .847,  .502}0.9829\% & \multicolumn{1}{r|}{0.03\%} & \multicolumn{1}{r|}{0.03\%} & \cellcolor[rgb]{ .996,  .863,  .506}0.0312 & 69.61\% & 70.44\% & 69.24\% & \cellcolor[rgb]{ .996,  .894,  .51}-0.4924 & \cellcolor[rgb]{ .996,  .788,  .494}0.1813 & \cellcolor[rgb]{ .996,  .855,  .502}0.0433 & \cellcolor[rgb]{ .973,  .412,  .42}0.8771 & \cellcolor[rgb]{ .808,  .867,  .506}0.9414 & \cellcolor[rgb]{ .984,  .663,  .467}0.0095\% & \cellcolor[rgb]{ .988,  .753,  .482}-0.0057 & \multicolumn{1}{r|}{63.42\%} & \cellcolor[rgb]{ .992,  .808,  .494}0.0159 & \cellcolor[rgb]{ .792,  .859,  .502}0.0296 & \cellcolor[rgb]{ .996,  .871,  .506}1.0983 \\
    \hline
    \textbf{MinV} & \cellcolor[rgb]{ .988,  .702,  .475}0.0279\% & \cellcolor[rgb]{ .533,  .784,  .49}0.8816\% & \multicolumn{1}{r|}{0.03\%} & \multicolumn{1}{r|}{0.03\%} & \cellcolor[rgb]{ .996,  .886,  .51}0.0317 & 65.04\% & 65.81\% & 64.39\% & \cellcolor[rgb]{ .741,  .847,  .506}-0.4578 & \cellcolor[rgb]{ .388,  .745,  .482}0.1412 & \cellcolor[rgb]{ .996,  .878,  .51}0.0440 & \cellcolor[rgb]{ .98,  .608,  .455}0.8897 & \cellcolor[rgb]{ .42,  .753,  .482}0.6970 & \cellcolor[rgb]{ .98,  .561,  .447}0.0083\% & \cellcolor[rgb]{ .98,  .612,  .455}-0.0095 & \multicolumn{1}{r|}{71.52\%} & \cellcolor[rgb]{ .992,  .843,  .502}0.0164 & \cellcolor[rgb]{ .451,  .761,  .482}0.0253 & \cellcolor[rgb]{ .667,  .827,  .502}1.1060 \\
    \hline
    \textbf{MinVsh} & \cellcolor[rgb]{ .98,  .569,  .447}0.0237\% & \cellcolor[rgb]{ .506,  .776,  .486}0.8700\% & \multicolumn{1}{r|}{0.03\%} & \multicolumn{1}{r|}{0.03\%} & \cellcolor[rgb]{ .984,  .659,  .467}0.0272 & 99.80\% & 99.67\% & 99.76\% & \cellcolor[rgb]{ .973,  .914,  .518}-0.4850 & \cellcolor[rgb]{ .973,  .412,  .42}0.2241 & \cellcolor[rgb]{ .984,  .671,  .467}0.0381 & \cellcolor[rgb]{ .388,  .745,  .482}0.9321 & \cellcolor[rgb]{ .973,  .412,  .42}3.4491 & \cellcolor[rgb]{ .961,  .914,  .518}0.0129\% & \cellcolor[rgb]{ .98,  .604,  .455}-0.0097 & \multicolumn{1}{r|}{71.99\%} & \cellcolor[rgb]{ .996,  .855,  .502}0.0166 & \cellcolor[rgb]{ .522,  .78,  .486}0.0262 & \cellcolor[rgb]{ .984,  .698,  .475}1.0887 \\
    \hline
    \textbf{MinVN1} & \cellcolor[rgb]{ .973,  .412,  .42}0.0188\% & \cellcolor[rgb]{ .443,  .761,  .482}0.8396\% & \multicolumn{1}{r|}{0.03\%} & \multicolumn{1}{r|}{0.03\%} & \cellcolor[rgb]{ .973,  .412,  .42}0.0223 & 75.14\% & 74.08\% & 75.73\% & \cellcolor[rgb]{ .91,  .898,  .514}-0.4776 & \cellcolor[rgb]{ .984,  .608,  .459}0.2018 & \cellcolor[rgb]{ .973,  .412,  .42}0.0307 & \cellcolor[rgb]{ .996,  .914,  .514}0.9090 & \cellcolor[rgb]{ .996,  .792,  .494}1.6692 & \cellcolor[rgb]{ .973,  .412,  .42}0.0065\% & \cellcolor[rgb]{ .973,  .412,  .42}-0.0149 & \multicolumn{1}{r|}{81.42\%} & \cellcolor[rgb]{ .973,  .412,  .42}0.0092 & \cellcolor[rgb]{ .388,  .745,  .482}0.0245 & \cellcolor[rgb]{ .973,  .412,  .42}1.0726 \\
    \hline
    \textbf{MinVN2} & \cellcolor[rgb]{ .98,  .58,  .451}0.0242\% & \cellcolor[rgb]{ .388,  .745,  .482}0.8134\% & \multicolumn{1}{r|}{0.03\%} & \multicolumn{1}{r|}{0.03\%} & \cellcolor[rgb]{ .992,  .784,  .49}0.0297 & 85.57\% & 86.07\% & 85.12\% & \cellcolor[rgb]{ .388,  .745,  .482}-0.4167 & \cellcolor[rgb]{ 1,  .922,  .518}0.1660 & \cellcolor[rgb]{ .988,  .769,  .486}0.0408 & \cellcolor[rgb]{ .973,  .467,  .427}0.8807 & \cellcolor[rgb]{ .388,  .745,  .482}0.6760 & \cellcolor[rgb]{ .98,  .596,  .455}0.0087\% & \cellcolor[rgb]{ .976,  .525,  .439}-0.0119 & \multicolumn{1}{r|}{76.11\%} & \cellcolor[rgb]{ .988,  .757,  .482}0.0150 & \cellcolor[rgb]{ .396,  .745,  .482}0.0246 & \cellcolor[rgb]{ .996,  .894,  .51}1.0995 \\
    \hline
    \textbf{Index} & \cellcolor[rgb]{ .969,  .914,  .518}0.0352\% & \cellcolor[rgb]{ .973,  .412,  .42}1.2960\% &       &       & \cellcolor[rgb]{ .984,  .659,  .467}0.0272 &       &       &       & \cellcolor[rgb]{ .973,  .412,  .42}-0.5678 & \cellcolor[rgb]{ .988,  .918,  .514}0.1653 & \cellcolor[rgb]{ .984,  .659,  .467}0.0378 & \cellcolor[rgb]{ .992,  .922,  .518}0.9099 &       &       &       &       &       & \cellcolor[rgb]{ .973,  .412,  .42}0.0393 & \cellcolor[rgb]{ .988,  .769,  .486}1.0926 \\
    \hline
    \end{tabular}%
    }
    }
  \label{tab:Perf_out_SP500_DL125_DH20_C}%
\end{sidewaystable}%

\begin{sidewaystable}[htbp]
  \centering
  \caption{Out-of-sample performance results for FF49}
  \scalebox{0.67}{
  {\renewcommand{\arraystretch}{1.1}
    \begin{tabular}{|l|c|c|c|c|c|r|r|r|c|c|c|c|c|c|c|c|c|}
    \hline
    \textbf{Approach} & \textbf{ExpRet} & \multicolumn{3}{c|}{\textbf{Vol}} & \multicolumn{4}{c|}{\textbf{Sharpe}} & \textbf{Sortino} & \textbf{Rachev5} & \textbf{Turn} & \textbf{AlphaJ} & \multicolumn{2}{c|}{\textbf{InfoRatio}} & \textbf{ApprRatio} & \textbf{VaR1} & \textbf{Omega} \\
    \hline
          &       & \textbf{value} & \multicolumn{1}{l|}{\textbf{p-v/b=1}} & \multicolumn{1}{l|}{\textbf{p-v/b=10}} & \textbf{value} & \multicolumn{1}{l|}{\textbf{p-v/b=1}} & \multicolumn{1}{l|}{\textbf{p-v/b=10}} & \multicolumn{1}{l|}{\textbf{p-v/HAC}} &       &       &       &       & \textbf{value} & \textbf{p-v} &       &       &  \\
    \hline
    \textbf{Roman-CVaR} & \cellcolor[rgb]{ .388,  .745,  .482}10.7464\% & \cellcolor[rgb]{ .984,  .576,  .455}95.4540\% & \multicolumn{1}{r|}{0.03\%} & \multicolumn{1}{r|}{0.03\%} & \cellcolor[rgb]{ .996,  .918,  .514}0.1126 & 0.03\% & 0.03\% & 0.00\% & \cellcolor[rgb]{ .867,  .882,  .51}0.1626 & \cellcolor[rgb]{ .914,  .898,  .514}0.9610 & \cellcolor[rgb]{ 1,  .902,  .514}0.8274 & \cellcolor[rgb]{ .388,  .745,  .482}4.8309\% & \cellcolor[rgb]{ .388,  .745,  .482}0.0540 & \multicolumn{1}{r|}{0.00\%} & \cellcolor[rgb]{ .996,  .906,  .514}0.0890 & \cellcolor[rgb]{ .984,  .612,  .459}2.6588 & \cellcolor[rgb]{ .996,  .91,  .514}1.3926 \\
    \hline
    \textbf{k-OWA-CVaR-5} & \cellcolor[rgb]{ .435,  .761,  .486}10.6469\% & \cellcolor[rgb]{ .984,  .616,  .459}94.3017\% & \multicolumn{1}{r|}{0.03\%} & \multicolumn{1}{r|}{0.03\%} & \cellcolor[rgb]{ .933,  .902,  .514}0.1129 & 0.03\% & 0.03\% & 0.00\% & \cellcolor[rgb]{ .831,  .875,  .51}0.1627 & \cellcolor[rgb]{ .949,  .91,  .518}0.9568 & \cellcolor[rgb]{ 1,  .91,  .518}0.8016 & \cellcolor[rgb]{ .424,  .757,  .486}4.7910\% & \cellcolor[rgb]{ .427,  .757,  .486}0.0527 & \multicolumn{1}{r|}{0.00\%} & \cellcolor[rgb]{ .996,  .918,  .514}0.0897 & \cellcolor[rgb]{ .988,  .659,  .471}2.6147 & \cellcolor[rgb]{ .996,  .918,  .514}1.3944 \\
    \hline
    \textbf{k-OWA-CVaR-10} & \cellcolor[rgb]{ .506,  .78,  .49}10.4907\% & \cellcolor[rgb]{ .988,  .655,  .467}92.9753\% & \multicolumn{1}{r|}{0.03\%} & \multicolumn{1}{r|}{0.03\%} & \cellcolor[rgb]{ .953,  .91,  .518}0.1128 & 0.03\% & 0.03\% & 0.00\% & \cellcolor[rgb]{ .976,  .918,  .518}0.1622 & \cellcolor[rgb]{ 1,  .922,  .518}0.9505 & \cellcolor[rgb]{ 1,  .918,  .518}0.7800 & \cellcolor[rgb]{ .506,  .78,  .49}4.6959\% & \cellcolor[rgb]{ .49,  .776,  .49}0.0507 & \multicolumn{1}{r|}{0.00\%} & \cellcolor[rgb]{ 1,  .922,  .518}0.0899 & \cellcolor[rgb]{ .988,  .678,  .475}2.5955 & \cellcolor[rgb]{ .996,  .918,  .514}1.3950 \\
    \hline
    \textbf{k-OWA-CVaR-25} & \cellcolor[rgb]{ .631,  .816,  .498}10.2073\% & \cellcolor[rgb]{ .992,  .776,  .49}89.1384\% & \multicolumn{1}{r|}{0.03\%} & \multicolumn{1}{r|}{0.03\%} & \cellcolor[rgb]{ .408,  .753,  .486}0.1145 & 0.03\% & 0.03\% & 0.00\% & \cellcolor[rgb]{ .427,  .757,  .486}0.1640 & \cellcolor[rgb]{ .988,  .733,  .478}0.9397 & \cellcolor[rgb]{ .961,  .91,  .514}0.7301 & \cellcolor[rgb]{ .576,  .8,  .494}4.6138\% & \cellcolor[rgb]{ .608,  .808,  .498}0.0469 & \multicolumn{1}{r|}{0.00\%} & \cellcolor[rgb]{ .549,  .792,  .494}0.0935 & \cellcolor[rgb]{ .992,  .733,  .482}2.5450 & \cellcolor[rgb]{ .553,  .796,  .494}1.4059 \\
    \hline
    \textbf{k-OWA-CVaR-50} & \cellcolor[rgb]{ .996,  .922,  .518}9.3912\% & \cellcolor[rgb]{ .937,  .902,  .514}82.3421\% & \multicolumn{1}{r|}{0.03\%} & \multicolumn{1}{r|}{0.03\%} & \cellcolor[rgb]{ .557,  .796,  .494}0.1141 & 0.03\% & 0.03\% & 0.00\% & \cellcolor[rgb]{ .996,  .918,  .514}0.1621 & \cellcolor[rgb]{ .973,  .412,  .42}0.9216 & \cellcolor[rgb]{ .827,  .871,  .506}0.6318 & \cellcolor[rgb]{ 1,  .922,  .518}4.1242\% & \cellcolor[rgb]{ 1,  .922,  .518}0.0340 & \multicolumn{1}{r|}{0.00\%} & \cellcolor[rgb]{ .388,  .745,  .482}0.0948 & \cellcolor[rgb]{ .918,  .898,  .51}2.2656 & \cellcolor[rgb]{ .396,  .749,  .486}1.4098 \\
    \hline
    \textbf{k-OWA-CVaR-75} & \cellcolor[rgb]{ .992,  .808,  .494}8.6147\% & \cellcolor[rgb]{ .776,  .855,  .502}76.5745\% & \multicolumn{1}{r|}{0.03\%} & \multicolumn{1}{r|}{0.03\%} & \cellcolor[rgb]{ .996,  .918,  .514}0.1125 & 0.03\% & 0.03\% & 0.00\% & \cellcolor[rgb]{ .996,  .894,  .51}0.1593 & \cellcolor[rgb]{ .973,  .424,  .42}0.9223 & \cellcolor[rgb]{ .714,  .839,  .498}0.5465 & \cellcolor[rgb]{ .992,  .82,  .498}3.6686\% & \cellcolor[rgb]{ .992,  .796,  .49}0.0196 & \multicolumn{1}{r|}{1.15\%} & \cellcolor[rgb]{ .604,  .808,  .498}0.0931 & \cellcolor[rgb]{ .812,  .867,  .506}2.1367 & \cellcolor[rgb]{ .455,  .765,  .486}1.4083 \\
    \hline
    \textbf{k-OWA-CVaR-100} & \cellcolor[rgb]{ .984,  .694,  .475}7.8655\% & \cellcolor[rgb]{ .655,  .82,  .494}72.1837\% & \multicolumn{1}{r|}{0.03\%} & \multicolumn{1}{r|}{0.03\%} & \cellcolor[rgb]{ .996,  .871,  .506}0.1090 & 0.03\% & 0.03\% & 0.00\% & \cellcolor[rgb]{ .996,  .855,  .502}0.1550 & \cellcolor[rgb]{ .996,  .914,  .514}0.9499 & \cellcolor[rgb]{ .596,  .804,  .494}0.4568 & \cellcolor[rgb]{ .988,  .725,  .478}3.2263\% & \cellcolor[rgb]{ .984,  .671,  .467}0.0052 & \multicolumn{1}{r|}{27.24\%} & \cellcolor[rgb]{ .996,  .863,  .506}0.0857 & \cellcolor[rgb]{ .714,  .839,  .498}2.0144 & \cellcolor[rgb]{ .984,  .918,  .518}1.3954 \\
    \hline
    \textbf{k-OWA-CCVaR-5} & \cellcolor[rgb]{ .435,  .761,  .486}10.6470\% & \cellcolor[rgb]{ .984,  .604,  .459}94.6393\% & \multicolumn{1}{r|}{0.03\%} & \multicolumn{1}{r|}{0.03\%} & \cellcolor[rgb]{ .996,  .918,  .514}0.1125 & 0.03\% & 0.03\% & 0.00\% & \cellcolor[rgb]{ .996,  .922,  .518}0.1621 & \cellcolor[rgb]{ .945,  .906,  .518}0.9572 & \cellcolor[rgb]{ 1,  .906,  .518}0.8106 & \cellcolor[rgb]{ .443,  .765,  .486}4.7676\% & \cellcolor[rgb]{ .427,  .757,  .486}0.0527 & \multicolumn{1}{r|}{0.00\%} & \cellcolor[rgb]{ .996,  .91,  .514}0.0890 & \cellcolor[rgb]{ .988,  .655,  .467}2.6202 & \cellcolor[rgb]{ .996,  .91,  .514}1.3926 \\
    \hline
    \textbf{k-OWA-CCVaR-10} & \cellcolor[rgb]{ .459,  .769,  .49}10.5938\% & \cellcolor[rgb]{ .984,  .631,  .463}93.7354\% & \multicolumn{1}{r|}{0.03\%} & \multicolumn{1}{r|}{0.03\%} & \cellcolor[rgb]{ .894,  .894,  .514}0.1130 & 0.03\% & 0.03\% & 0.00\% & \cellcolor[rgb]{ .784,  .859,  .506}0.1628 & \cellcolor[rgb]{ .949,  .91,  .518}0.9565 & \cellcolor[rgb]{ 1,  .914,  .518}0.7938 & \cellcolor[rgb]{ .447,  .765,  .486}4.7636\% & \cellcolor[rgb]{ .447,  .765,  .486}0.0521 & \multicolumn{1}{r|}{0.00\%} & \cellcolor[rgb]{ .98,  .918,  .518}0.0901 & \cellcolor[rgb]{ .988,  .678,  .475}2.5965 & \cellcolor[rgb]{ 1,  .922,  .518}1.3951 \\
    \hline
    \textbf{k-OWA-CCVaR-25} & \cellcolor[rgb]{ .518,  .784,  .49}10.4592\% & \cellcolor[rgb]{ .988,  .698,  .475}91.6365\% & \multicolumn{1}{r|}{0.03\%} & \multicolumn{1}{r|}{0.03\%} & \cellcolor[rgb]{ .529,  .788,  .494}0.1141 & 0.03\% & 0.03\% & 0.00\% & \cellcolor[rgb]{ .388,  .745,  .482}0.1641 & \cellcolor[rgb]{ .996,  .918,  .514}0.9501 & \cellcolor[rgb]{ 1,  .922,  .518}0.7591 & \cellcolor[rgb]{ .463,  .769,  .49}4.7480\% & \cellcolor[rgb]{ .502,  .78,  .49}0.0503 & \multicolumn{1}{r|}{0.00\%} & \cellcolor[rgb]{ .714,  .839,  .502}0.0922 & \cellcolor[rgb]{ .988,  .686,  .475}2.5895 & \cellcolor[rgb]{ .694,  .835,  .502}1.4025 \\
    \hline
    \textbf{k-OWA-CCVaR-50} & \cellcolor[rgb]{ .757,  .851,  .506}9.9291\% & \cellcolor[rgb]{ 1,  .855,  .506}86.6683\% & \multicolumn{1}{r|}{0.03\%} & \multicolumn{1}{r|}{0.03\%} & \cellcolor[rgb]{ .388,  .745,  .482}0.1146 & 0.03\% & 0.03\% & 0.00\% & \cellcolor[rgb]{ .518,  .784,  .49}0.1637 & \cellcolor[rgb]{ .98,  .592,  .451}0.9317 & \cellcolor[rgb]{ .898,  .89,  .51}0.6843 & \cellcolor[rgb]{ .714,  .839,  .502}4.4557\% & \cellcolor[rgb]{ .733,  .847,  .506}0.0427 & \multicolumn{1}{r|}{0.00\%} & \cellcolor[rgb]{ .455,  .765,  .486}0.0943 & \cellcolor[rgb]{ .996,  .824,  .502}2.4592 & \cellcolor[rgb]{ .463,  .769,  .49}1.4081 \\
    \hline
    \textbf{k-OWA-CCVaR-75} & \cellcolor[rgb]{ .996,  .918,  .514}9.3683\% & \cellcolor[rgb]{ .933,  .902,  .514}82.1244\% & \multicolumn{1}{r|}{0.03\%} & \multicolumn{1}{r|}{0.03\%} & \cellcolor[rgb]{ .549,  .792,  .494}0.1141 & 0.03\% & 0.03\% & 0.00\% & \cellcolor[rgb]{ .976,  .918,  .518}0.1622 & \cellcolor[rgb]{ .973,  .412,  .42}0.9215 & \cellcolor[rgb]{ .824,  .871,  .506}0.6270 & \cellcolor[rgb]{ .996,  .918,  .514}4.1232\% & \cellcolor[rgb]{ .996,  .914,  .514}0.0334 & \multicolumn{1}{r|}{0.01\%} & \cellcolor[rgb]{ .408,  .753,  .486}0.0947 & \cellcolor[rgb]{ .922,  .898,  .51}2.2680 & \cellcolor[rgb]{ .388,  .745,  .482}1.4099 \\
    \hline
    \textbf{k-OWA-CCVaR-100} & \cellcolor[rgb]{ .992,  .827,  .498}8.7489\% & \cellcolor[rgb]{ .804,  .863,  .506}77.5644\% & \multicolumn{1}{r|}{0.03\%} & \multicolumn{1}{r|}{0.03\%} & \cellcolor[rgb]{ .969,  .914,  .518}0.1128 & 0.03\% & 0.03\% & 0.00\% & \cellcolor[rgb]{ .996,  .898,  .514}0.1599 & \cellcolor[rgb]{ .973,  .424,  .42}0.9222 & \cellcolor[rgb]{ .737,  .843,  .502}0.5628 & \cellcolor[rgb]{ .992,  .839,  .502}3.7591\% & \cellcolor[rgb]{ .992,  .816,  .494}0.0220 & \multicolumn{1}{r|}{0.53\%} & \cellcolor[rgb]{ .6,  .808,  .498}0.0931 & \cellcolor[rgb]{ .82,  .867,  .506}2.1458 & \cellcolor[rgb]{ .459,  .765,  .486}1.4083 \\
    \hline
    \textbf{MinV} & \cellcolor[rgb]{ .973,  .455,  .427}6.2402\% & \cellcolor[rgb]{ .545,  .788,  .49}68.1950\% & \multicolumn{1}{r|}{0.03\%} & \multicolumn{1}{r|}{0.03\%} & \cellcolor[rgb]{ .984,  .631,  .459}0.0915 & 0.27\% & 0.57\% & 1.17\% & \cellcolor[rgb]{ .98,  .62,  .459}0.1294 & \cellcolor[rgb]{ .957,  .91,  .518}0.9556 & \cellcolor[rgb]{ .388,  .745,  .482}0.2986 & \cellcolor[rgb]{ .973,  .412,  .42}1.7604\% & \cellcolor[rgb]{ .973,  .412,  .42}-0.0253 & \multicolumn{1}{r|}{99.84\%} & \cellcolor[rgb]{ .973,  .412,  .42}0.0527 & \cellcolor[rgb]{ .576,  .8,  .49}1.8493 & \cellcolor[rgb]{ .984,  .678,  .471}1.3299 \\
    \hline
    \textbf{MinVsh} & \cellcolor[rgb]{ .973,  .412,  .42}5.9300\% & \cellcolor[rgb]{ .392,  .745,  .482}62.6256\% & \multicolumn{1}{r|}{0.03\%} & \multicolumn{1}{r|}{0.03\%} & \cellcolor[rgb]{ .984,  .675,  .467}0.0947 & 8.90\% & 8.86\% & 11.69\% & \cellcolor[rgb]{ .984,  .694,  .471}0.1372 & \cellcolor[rgb]{ .388,  .745,  .482}1.0241 & \cellcolor[rgb]{ .973,  .412,  .42}2.3053 & \cellcolor[rgb]{ .996,  .902,  .514}4.0419\% & \cellcolor[rgb]{ .973,  .478,  .431}-0.0174 & \multicolumn{1}{r|}{97.84\%} & \cellcolor[rgb]{ .984,  .651,  .463}0.0704 & \cellcolor[rgb]{ .392,  .745,  .482}1.6200 & \cellcolor[rgb]{ .984,  .686,  .471}1.3315 \\
    \hline
    \textbf{MinVN1} & \cellcolor[rgb]{ .973,  .412,  .42}5.9299\% & \cellcolor[rgb]{ .392,  .745,  .482}62.6202\% & \multicolumn{1}{r|}{0.03\%} & \multicolumn{1}{r|}{0.03\%} & \cellcolor[rgb]{ .984,  .675,  .467}0.0947 & 8.90\% & 9.96\% & 11.68\% & \cellcolor[rgb]{ .984,  .694,  .471}0.1372 & \cellcolor[rgb]{ .392,  .749,  .486}1.0241 & \cellcolor[rgb]{ .976,  .416,  .424}2.3039 & \cellcolor[rgb]{ .996,  .902,  .514}4.0416\% & \cellcolor[rgb]{ .973,  .478,  .431}-0.0174 & \multicolumn{1}{r|}{97.84\%} & \cellcolor[rgb]{ .984,  .651,  .463}0.0704 & \cellcolor[rgb]{ .392,  .745,  .482}1.6198 & \cellcolor[rgb]{ .984,  .686,  .471}1.3315 \\
    \hline
    \textbf{MinVN2} & \cellcolor[rgb]{ .973,  .412,  .42}5.9399\% & \cellcolor[rgb]{ .388,  .745,  .482}62.4297\% & \multicolumn{1}{r|}{0.03\%} & \multicolumn{1}{r|}{0.03\%} & \cellcolor[rgb]{ .984,  .678,  .471}0.0951 & 8.26\% & 8.26\% & 10.77\% & \cellcolor[rgb]{ .984,  .698,  .475}0.1378 & \cellcolor[rgb]{ .396,  .749,  .486}1.0232 & \cellcolor[rgb]{ .976,  .431,  .424}2.2554 & \cellcolor[rgb]{ .996,  .902,  .514}4.0402\% & \cellcolor[rgb]{ .973,  .478,  .431}-0.0174 & \multicolumn{1}{r|}{97.81\%} & \cellcolor[rgb]{ .984,  .659,  .467}0.0707 & \cellcolor[rgb]{ .388,  .745,  .482}1.6143 & \cellcolor[rgb]{ .984,  .694,  .471}1.3339 \\
    \hline
    \textbf{Index - EW} & \cellcolor[rgb]{ .984,  .655,  .463}7.5825\% & \cellcolor[rgb]{ .973,  .412,  .42}100.6446\% &       &       & \cellcolor[rgb]{ .973,  .412,  .42}0.0753 &       &       &       & \cellcolor[rgb]{ .973,  .412,  .42}0.1062 & \cellcolor[rgb]{ .996,  .914,  .514}0.9500 &       &       &       &       &       & \cellcolor[rgb]{ .973,  .412,  .42}2.8491 & \cellcolor[rgb]{ .973,  .412,  .42}1.2568 \\
    \hline
    \end{tabular}%
  }
  }
  \label{tab:Perf_out_49IndPort_DL125_DH20_C}%
\end{sidewaystable}%

\noindent
In Tables \ref{tab:ROI750_out_FTSE100_DL125_DH20_new_C}, \ref{tab:ROI750_out_NASDAQ100_DL125_DH20_new_C}, \ref{tab:ROI750_out_SP500_DL125_DH20_new_C} we report some statistics of ROI based on a 3-years time horizon, computed using Expression \eqref{eq:ROI} with $\Delta H =750$ days, for all the portfolio strategies analyzed.
Again, for the FTSE100, NASDAQ100, and SP500 daily datasets, we can observe that the most profitable portfolio strategies appear to be the $k$-OWA-CVaR-$\beta$ and $k$-OWA-CCVaR-$\beta$ portfolios with $\beta=5\%, 10\%$, along with the Roman-CVaR approach, in terms of mean, 25\%, 50\%, 75\%, and 95\% percentile of ROI.
%
%
%
% Table generated by Excel2LaTeX from sheet 'ROI750_new (3)'
\begin{table}[htbp]
  \centering
  \caption{Summary statistics of ROI based on a 3-years time horizon for FTSE100}
  \scalebox{0.75}{
  {\renewcommand{\arraystretch}{1.1}
    \begin{tabular}{|l|c|c|c|c|c|c|c|}
    \hline
    \textbf{Approach} & \textbf{ExpRet} & \textbf{Vol} & \textbf{5\%-perc} & \textbf{25\%-perc} & \textbf{50\%-perc} & \textbf{75\%-perc} & \textbf{95\%-perc} \\
    \hline
    \textbf{Roman-CVaR} & \cellcolor[rgb]{ .537,  .788,  .494}65\% & \cellcolor[rgb]{ .973,  .412,  .42}45\% & \cellcolor[rgb]{ .98,  .576,  .451}-21\% & \cellcolor[rgb]{ .996,  .906,  .514}20\% & \cellcolor[rgb]{ .451,  .765,  .486}80\% & \cellcolor[rgb]{ .533,  .788,  .494}99\% & \cellcolor[rgb]{ .522,  .784,  .49}123\% \\
    \hline
    \textbf{k-OWA-CVaR-5} & \cellcolor[rgb]{ .412,  .753,  .486}69\% & \cellcolor[rgb]{ .976,  .447,  .427}45\% & \cellcolor[rgb]{ .992,  .812,  .494}-13\% & \cellcolor[rgb]{ .584,  .804,  .494}25\% & \cellcolor[rgb]{ .427,  .757,  .486}81\% & \cellcolor[rgb]{ .427,  .757,  .486}104\% & \cellcolor[rgb]{ .51,  .78,  .49}124\% \\
    \hline
    \textbf{k-OWA-CVaR-10} & \cellcolor[rgb]{ .431,  .757,  .486}68\% & \cellcolor[rgb]{ .976,  .416,  .424}45\% & \cellcolor[rgb]{ .996,  .89,  .51}-10\% & \cellcolor[rgb]{ .937,  .906,  .518}21\% & \cellcolor[rgb]{ .443,  .761,  .486}81\% & \cellcolor[rgb]{ .51,  .78,  .49}100\% & \cellcolor[rgb]{ .388,  .745,  .482}130\% \\
    \hline
    \textbf{k-OWA-CVaR-25} & \cellcolor[rgb]{ .592,  .804,  .494}64\% & \cellcolor[rgb]{ .984,  .608,  .459}41\% & \cellcolor[rgb]{ .875,  .886,  .514}-7\% & \cellcolor[rgb]{ .98,  .918,  .518}21\% & \cellcolor[rgb]{ .557,  .796,  .494}77\% & \cellcolor[rgb]{ .612,  .812,  .498}96\% & \cellcolor[rgb]{ .655,  .824,  .498}116\% \\
    \hline
    \textbf{k-OWA-CVaR-50} & \cellcolor[rgb]{ .996,  .922,  .518}53\% & \cellcolor[rgb]{ .996,  .918,  .514}33\% & \cellcolor[rgb]{ .788,  .863,  .506}-5\% & \cellcolor[rgb]{ .996,  .918,  .514}20\% & \cellcolor[rgb]{ .996,  .922,  .518}64\% & \cellcolor[rgb]{ .988,  .918,  .518}78\% & \cellcolor[rgb]{ .996,  .914,  .514}95\% \\
    \hline
    \textbf{k-OWA-CVaR-75} & \cellcolor[rgb]{ .992,  .8,  .494}43\% & \cellcolor[rgb]{ .816,  .867,  .506}28\% & \cellcolor[rgb]{ .894,  .89,  .514}-7\% & \cellcolor[rgb]{ .988,  .741,  .482}14\% & \cellcolor[rgb]{ .992,  .816,  .494}53\% & \cellcolor[rgb]{ .992,  .8,  .494}63\% & \cellcolor[rgb]{ .988,  .765,  .486}76\% \\
    \hline
    \textbf{k-OWA-CVaR-100} & \cellcolor[rgb]{ .992,  .8,  .494}43\% & \cellcolor[rgb]{ .804,  .863,  .506}27\% & \cellcolor[rgb]{ .71,  .839,  .502}-3\% & \cellcolor[rgb]{ .992,  .831,  .498}17\% & \cellcolor[rgb]{ .992,  .792,  .49}50\% & \cellcolor[rgb]{ .992,  .792,  .49}63\% & \cellcolor[rgb]{ .992,  .792,  .49}80\% \\
    \hline
    \textbf{k-OWA-CCVaR-5} & \cellcolor[rgb]{ .392,  .749,  .486}69\% & \cellcolor[rgb]{ .976,  .435,  .424}45\% & \cellcolor[rgb]{ .992,  .784,  .49}-14\% & \cellcolor[rgb]{ .388,  .745,  .482}28\% & \cellcolor[rgb]{ .404,  .749,  .486}82\% & \cellcolor[rgb]{ .4,  .749,  .486}105\% & \cellcolor[rgb]{ .49,  .776,  .49}125\% \\
    \hline
    \textbf{k-OWA-CCVaR-10} & \cellcolor[rgb]{ .388,  .745,  .482}69\% & \cellcolor[rgb]{ .976,  .416,  .424}45\% & \cellcolor[rgb]{ .996,  .847,  .502}-12\% & \cellcolor[rgb]{ .82,  .871,  .51}23\% & \cellcolor[rgb]{ .388,  .745,  .482}82\% & \cellcolor[rgb]{ .388,  .745,  .482}106\% & \cellcolor[rgb]{ .467,  .769,  .49}126\% \\
    \hline
    \textbf{k-OWA-CCVaR-25} & \cellcolor[rgb]{ .459,  .769,  .49}67\% & \cellcolor[rgb]{ .976,  .478,  .435}44\% & \cellcolor[rgb]{ .949,  .91,  .518}-8\% & \cellcolor[rgb]{ .996,  .918,  .514}20\% & \cellcolor[rgb]{ .431,  .757,  .486}81\% & \cellcolor[rgb]{ .498,  .776,  .49}101\% & \cellcolor[rgb]{ .51,  .78,  .49}124\% \\
    \hline
    \textbf{k-OWA-CCVaR-50} & \cellcolor[rgb]{ .686,  .831,  .502}61\% & \cellcolor[rgb]{ .992,  .737,  .482}38\% & \cellcolor[rgb]{ .694,  .835,  .502}-3\% & \cellcolor[rgb]{ .886,  .89,  .514}22\% & \cellcolor[rgb]{ .655,  .824,  .498}74\% & \cellcolor[rgb]{ .725,  .843,  .502}90\% & \cellcolor[rgb]{ .78,  .859,  .506}108\% \\
    \hline
    \textbf{k-OWA-CCVaR-75} & \cellcolor[rgb]{ .996,  .918,  .514}53\% & \cellcolor[rgb]{ 1,  .922,  .518}33\% & \cellcolor[rgb]{ .804,  .867,  .51}-5\% & \cellcolor[rgb]{ 1,  .922,  .518}20\% & \cellcolor[rgb]{ .996,  .918,  .514}63\% & \cellcolor[rgb]{ .996,  .914,  .514}77\% & \cellcolor[rgb]{ .988,  .922,  .518}97\% \\
    \hline
    \textbf{k-OWA-CCVaR-100} & \cellcolor[rgb]{ .992,  .847,  .502}47\% & \cellcolor[rgb]{ .882,  .886,  .51}30\% & \cellcolor[rgb]{ .812,  .867,  .51}-5\% & \cellcolor[rgb]{ .992,  .792,  .49}16\% & \cellcolor[rgb]{ .996,  .855,  .502}57\% & \cellcolor[rgb]{ .992,  .843,  .502}69\% & \cellcolor[rgb]{ .992,  .831,  .498}85\% \\
    \hline
    \textbf{MinV} & \cellcolor[rgb]{ .988,  .737,  .478}37\% & \cellcolor[rgb]{ .51,  .78,  .486}19\% & \cellcolor[rgb]{ .388,  .745,  .482}3\% & \cellcolor[rgb]{ .592,  .804,  .494}25\% & \cellcolor[rgb]{ .984,  .702,  .475}41\% & \cellcolor[rgb]{ .984,  .686,  .471}50\% & \cellcolor[rgb]{ .984,  .667,  .467}63\% \\
    \hline
    \textbf{MinVsh} & \cellcolor[rgb]{ .984,  .69,  .471}33\% & \cellcolor[rgb]{ .706,  .835,  .498}24\% & \cellcolor[rgb]{ .988,  .749,  .482}-15\% & \cellcolor[rgb]{ .675,  .827,  .502}24\% & \cellcolor[rgb]{ .984,  .682,  .471}39\% & \cellcolor[rgb]{ .984,  .682,  .471}49\% & \cellcolor[rgb]{ .984,  .663,  .467}63\% \\
    \hline
    \textbf{MinVN1} & \cellcolor[rgb]{ .984,  .647,  .463}29\% & \cellcolor[rgb]{ .824,  .871,  .506}28\% & \cellcolor[rgb]{ .973,  .412,  .42}-27\% & \cellcolor[rgb]{ .992,  .831,  .498}17\% & \cellcolor[rgb]{ .984,  .678,  .471}38\% & \cellcolor[rgb]{ .984,  .675,  .467}48\% & \cellcolor[rgb]{ .984,  .655,  .463}62\% \\
    \hline
    \textbf{MinVN2} & \cellcolor[rgb]{ .98,  .604,  .455}25\% & \cellcolor[rgb]{ .51,  .776,  .486}18\% & \cellcolor[rgb]{ .992,  .827,  .498}-12\% & \cellcolor[rgb]{ .992,  .812,  .494}16\% & \cellcolor[rgb]{ .98,  .604,  .455}30\% & \cellcolor[rgb]{ .98,  .588,  .451}38\% & \cellcolor[rgb]{ .976,  .557,  .447}49\% \\
    \hline
    \textbf{Index} & \cellcolor[rgb]{ .973,  .412,  .42}8\% & \cellcolor[rgb]{ .388,  .745,  .482}15\% & \cellcolor[rgb]{ .984,  .635,  .463}-19\% & \cellcolor[rgb]{ .973,  .412,  .42}1\% & \cellcolor[rgb]{ .973,  .412,  .42}10\% & \cellcolor[rgb]{ .973,  .412,  .42}17\% & \cellcolor[rgb]{ .973,  .412,  .42}31\% \\
    \hline
    \end{tabular}%
    }
    }
  \label{tab:ROI750_out_FTSE100_DL125_DH20_new_C}%
\end{table}%
%
%
% Table generated by Excel2LaTeX from sheet 'ROI750_new (2)'
\begin{table}[htbp]
  \centering
  \caption{Summary statistics of ROI based on a 3-years time horizon for NASDAQ100}
  \scalebox{0.75}{
  {\renewcommand{\arraystretch}{1.1}
        \begin{tabular}{|l|c|c|c|c|c|c|c|}
    \hline
    \textbf{Approach} & \textbf{ExpRet} & \textbf{Vol} & \textbf{5\%-perc} & \textbf{25\%-perc} & \textbf{50\%-perc} & \textbf{75\%-perc} & \textbf{95\%-perc} \\
    \hline
    \textbf{Roman-CVaR} & \cellcolor[rgb]{ .388,  .745,  .482}79\% & \cellcolor[rgb]{ .996,  .847,  .506}30\% & \cellcolor[rgb]{ .533,  .788,  .494}12\% & \cellcolor[rgb]{ .494,  .776,  .49}61\% & \cellcolor[rgb]{ .388,  .745,  .482}86\% & \cellcolor[rgb]{ .388,  .745,  .482}102\% & \cellcolor[rgb]{ .455,  .765,  .486}117\% \\
    \hline
    \textbf{k-OWA-CVaR-5} & \cellcolor[rgb]{ .404,  .749,  .486}79\% & \cellcolor[rgb]{ .996,  .918,  .514}27\% & \cellcolor[rgb]{ .451,  .765,  .486}14\% & \cellcolor[rgb]{ .388,  .745,  .482}66\% & \cellcolor[rgb]{ .427,  .757,  .486}84\% & \cellcolor[rgb]{ .486,  .776,  .49}98\% & \cellcolor[rgb]{ .529,  .788,  .494}114\% \\
    \hline
    \textbf{k-OWA-CVaR-10} & \cellcolor[rgb]{ .518,  .784,  .49}74\% & \cellcolor[rgb]{ 1,  .886,  .514}29\% & \cellcolor[rgb]{ .455,  .765,  .486}14\% & \cellcolor[rgb]{ .596,  .808,  .498}58\% & \cellcolor[rgb]{ .557,  .796,  .494}78\% & \cellcolor[rgb]{ .529,  .788,  .494}96\% & \cellcolor[rgb]{ .604,  .808,  .498}112\% \\
    \hline
    \textbf{k-OWA-CVaR-25} & \cellcolor[rgb]{ .812,  .867,  .51}62\% & \cellcolor[rgb]{ 1,  .922,  .518}28\% & \cellcolor[rgb]{ .871,  .886,  .514}6\% & \cellcolor[rgb]{ .918,  .898,  .514}45\% & \cellcolor[rgb]{ .851,  .878,  .51}64\% & \cellcolor[rgb]{ .843,  .878,  .51}83\% & \cellcolor[rgb]{ .871,  .886,  .514}101\% \\
    \hline
    \textbf{k-OWA-CVaR-50} & \cellcolor[rgb]{ .996,  .863,  .506}52\% & \cellcolor[rgb]{ .957,  .906,  .514}27\% & \cellcolor[rgb]{ .992,  .788,  .49}-4\% & \cellcolor[rgb]{ .996,  .863,  .506}37\% & \cellcolor[rgb]{ .996,  .875,  .506}55\% & \cellcolor[rgb]{ .992,  .835,  .498}72\% & \cellcolor[rgb]{ .992,  .827,  .498}91\% \\
    \hline
    \textbf{k-OWA-CVaR-75} & \cellcolor[rgb]{ .988,  .765,  .486}50\% & \cellcolor[rgb]{ .62,  .812,  .494}21\% & \cellcolor[rgb]{ .996,  .918,  .514}3\% & \cellcolor[rgb]{ .996,  .851,  .502}36\% & \cellcolor[rgb]{ .992,  .82,  .498}54\% & \cellcolor[rgb]{ .984,  .694,  .471}67\% & \cellcolor[rgb]{ .976,  .541,  .443}76\% \\
    \hline
    \textbf{k-OWA-CVaR-100} & \cellcolor[rgb]{ .98,  .62,  .459}46\% & \cellcolor[rgb]{ .388,  .745,  .482}17\% & \cellcolor[rgb]{ .808,  .867,  .51}7\% & \cellcolor[rgb]{ .996,  .855,  .502}36\% & \cellcolor[rgb]{ .984,  .675,  .467}49\% & \cellcolor[rgb]{ .976,  .518,  .439}59\% & \cellcolor[rgb]{ .973,  .412,  .42}69\% \\
    \hline
    \textbf{k-OWA-CCVaR-5} & \cellcolor[rgb]{ .416,  .753,  .486}78\% & \cellcolor[rgb]{ 1,  .91,  .518}28\% & \cellcolor[rgb]{ .388,  .745,  .482}15\% & \cellcolor[rgb]{ .408,  .753,  .486}65\% & \cellcolor[rgb]{ .447,  .765,  .486}83\% & \cellcolor[rgb]{ .486,  .773,  .49}98\% & \cellcolor[rgb]{ .49,  .776,  .49}116\% \\
    \hline
    \textbf{k-OWA-CCVaR-10} & \cellcolor[rgb]{ .529,  .788,  .494}73\% & \cellcolor[rgb]{ 1,  .878,  .51}29\% & \cellcolor[rgb]{ .655,  .824,  .498}10\% & \cellcolor[rgb]{ .592,  .804,  .494}58\% & \cellcolor[rgb]{ .549,  .792,  .494}78\% & \cellcolor[rgb]{ .549,  .792,  .494}95\% & \cellcolor[rgb]{ .608,  .808,  .498}111\% \\
    \hline
    \textbf{k-OWA-CCVaR-25} & \cellcolor[rgb]{ .62,  .812,  .498}70\% & \cellcolor[rgb]{ 1,  .902,  .514}28\% & \cellcolor[rgb]{ .627,  .816,  .498}10\% & \cellcolor[rgb]{ .706,  .839,  .502}53\% & \cellcolor[rgb]{ .678,  .831,  .502}72\% & \cellcolor[rgb]{ .627,  .816,  .498}92\% & \cellcolor[rgb]{ .671,  .827,  .502}109\% \\
    \hline
    \textbf{k-OWA-CCVaR-50} & \cellcolor[rgb]{ .941,  .906,  .518}56\% & \cellcolor[rgb]{ 1,  .882,  .51}29\% & \cellcolor[rgb]{ .992,  .824,  .498}-2\% & \cellcolor[rgb]{ .996,  .875,  .506}38\% & \cellcolor[rgb]{ .976,  .918,  .518}58\% & \cellcolor[rgb]{ .945,  .906,  .518}79\% & \cellcolor[rgb]{ .914,  .898,  .514}99\% \\
    \hline
    \textbf{k-OWA-CCVaR-75} & \cellcolor[rgb]{ .996,  .902,  .514}53\% & \cellcolor[rgb]{ .961,  .91,  .514}27\% & \cellcolor[rgb]{ .992,  .804,  .494}-3\% & \cellcolor[rgb]{ .996,  .871,  .506}38\% & \cellcolor[rgb]{ .996,  .898,  .51}56\% & \cellcolor[rgb]{ .996,  .863,  .506}74\% & \cellcolor[rgb]{ .996,  .851,  .502}92\% \\
    \hline
    \textbf{k-OWA-CCVaR-100} & \cellcolor[rgb]{ .992,  .804,  .494}51\% & \cellcolor[rgb]{ .643,  .816,  .494}22\% & \cellcolor[rgb]{ 1,  .922,  .518}3\% & \cellcolor[rgb]{ .996,  .859,  .506}37\% & \cellcolor[rgb]{ .996,  .855,  .502}55\% & \cellcolor[rgb]{ .988,  .729,  .478}68\% & \cellcolor[rgb]{ .98,  .573,  .451}78\% \\
    \hline
    \textbf{MinV} & \cellcolor[rgb]{ .973,  .412,  .42}41\% & \cellcolor[rgb]{ .49,  .773,  .486}19\% & \cellcolor[rgb]{ .996,  .867,  .506}0\% & \cellcolor[rgb]{ .988,  .765,  .486}30\% & \cellcolor[rgb]{ .976,  .502,  .435}44\% & \cellcolor[rgb]{ .973,  .412,  .42}55\% & \cellcolor[rgb]{ .973,  .42,  .42}70\% \\
    \hline
    \textbf{MinVsh} & \cellcolor[rgb]{ .973,  .455,  .427}42\% & \cellcolor[rgb]{ .973,  .412,  .42}46\% & \cellcolor[rgb]{ .973,  .412,  .42}-25\% & \cellcolor[rgb]{ .973,  .412,  .42}3\% & \cellcolor[rgb]{ .973,  .424,  .42}41\% & \cellcolor[rgb]{ .918,  .898,  .514}80\% & \cellcolor[rgb]{ .388,  .745,  .482}120\% \\
    \hline
    \textbf{MinVN1} & \cellcolor[rgb]{ .973,  .443,  .424}42\% & \cellcolor[rgb]{ 1,  .855,  .506}30\% & \cellcolor[rgb]{ .98,  .612,  .455}-14\% & \cellcolor[rgb]{ .98,  .6,  .455}18\% & \cellcolor[rgb]{ .973,  .412,  .42}41\% & \cellcolor[rgb]{ .988,  .761,  .486}69\% & \cellcolor[rgb]{ .988,  .71,  .475}85\% \\
    \hline
    \textbf{MinVN2} & \cellcolor[rgb]{ .996,  .871,  .506}52\% & \cellcolor[rgb]{ .878,  .886,  .51}26\% & \cellcolor[rgb]{ .984,  .627,  .459}-13\% & \cellcolor[rgb]{ .804,  .867,  .51}49\% & \cellcolor[rgb]{ .988,  .918,  .518}58\% & \cellcolor[rgb]{ .988,  .714,  .475}67\% & \cellcolor[rgb]{ .984,  .655,  .467}82\% \\
    \hline
    \textbf{Index} & \cellcolor[rgb]{ .992,  .922,  .518}54\% & \cellcolor[rgb]{ .753,  .851,  .502}23\% & \cellcolor[rgb]{ .996,  .894,  .51}2\% & \cellcolor[rgb]{ .914,  .898,  .514}45\% & \cellcolor[rgb]{ .996,  .875,  .506}55\% & \cellcolor[rgb]{ .988,  .725,  .478}68\% & \cellcolor[rgb]{ .992,  .784,  .49}88\% \\
    \hline
    \end{tabular}%
    }
    }
  \label{tab:ROI750_out_NASDAQ100_DL125_DH20_new_C}%
\end{table}%
%
%
% Table generated by Excel2LaTeX from sheet 'ROI750_new (2)'
\begin{table}[htbp]
  \centering
  \caption{Summary statistics of ROI based on a 3-years time horizon for SP500}
  \scalebox{0.75}{
  {\renewcommand{\arraystretch}{1.1}
        \begin{tabular}{|l|c|c|c|c|c|c|c|}
    \hline
    \textbf{Approach} & \textbf{ExpRet} & \textbf{Vol} & \textbf{5\%-perc} & \textbf{25\%-perc} & \textbf{50\%-perc} & \textbf{75\%-perc} & \textbf{95\%-perc} \\
    \hline
    \textbf{Roman-CVaR} & \cellcolor[rgb]{ .627,  .816,  .498}37\% & \cellcolor[rgb]{ .976,  .914,  .514}24\% & \cellcolor[rgb]{ .388,  .745,  .482}-13\% & \cellcolor[rgb]{ .961,  .914,  .518}28\% & \cellcolor[rgb]{ .769,  .855,  .506}40\% & \cellcolor[rgb]{ .855,  .882,  .51}48\% & \cellcolor[rgb]{ .388,  .745,  .482}77\% \\
    \hline
    \textbf{k-OWA-CVaR-5} & \cellcolor[rgb]{ .388,  .745,  .482}40\% & \cellcolor[rgb]{ 1,  .867,  .51}24\% & \cellcolor[rgb]{ .592,  .804,  .494}-15\% & \cellcolor[rgb]{ .459,  .769,  .49}33\% & \cellcolor[rgb]{ .443,  .761,  .486}43\% & \cellcolor[rgb]{ .412,  .753,  .486}53\% & \cellcolor[rgb]{ .455,  .765,  .486}76\% \\
    \hline
    \textbf{k-OWA-CVaR-10} & \cellcolor[rgb]{ .553,  .792,  .494}38\% & \cellcolor[rgb]{ 1,  .875,  .51}24\% & \cellcolor[rgb]{ .961,  .914,  .518}-19\% & \cellcolor[rgb]{ .439,  .761,  .486}33\% & \cellcolor[rgb]{ .545,  .792,  .494}42\% & \cellcolor[rgb]{ .671,  .827,  .502}50\% & \cellcolor[rgb]{ .584,  .804,  .494}72\% \\
    \hline
    \textbf{k-OWA-CVaR-25} & \cellcolor[rgb]{ .847,  .878,  .51}34\% & \cellcolor[rgb]{ .925,  .898,  .51}23\% & \cellcolor[rgb]{ .894,  .894,  .514}-18\% & \cellcolor[rgb]{ .824,  .871,  .51}29\% & \cellcolor[rgb]{ .851,  .878,  .51}39\% & \cellcolor[rgb]{ .996,  .902,  .514}46\% & \cellcolor[rgb]{ .843,  .878,  .51}65\% \\
    \hline
    \textbf{k-OWA-CVaR-50} & \cellcolor[rgb]{ .996,  .855,  .502}31\% & \cellcolor[rgb]{ 1,  .894,  .514}24\% & \cellcolor[rgb]{ .992,  .78,  .49}-24\% & \cellcolor[rgb]{ .992,  .816,  .494}25\% & \cellcolor[rgb]{ .996,  .863,  .506}36\% & \cellcolor[rgb]{ .957,  .91,  .518}47\% & \cellcolor[rgb]{ .988,  .729,  .478}57\% \\
    \hline
    \textbf{k-OWA-CVaR-75} & \cellcolor[rgb]{ .984,  .694,  .471}27\% & \cellcolor[rgb]{ .584,  .8,  .49}21\% & \cellcolor[rgb]{ .992,  .804,  .494}-24\% & \cellcolor[rgb]{ .988,  .769,  .486}24\% & \cellcolor[rgb]{ .984,  .702,  .475}32\% & \cellcolor[rgb]{ .98,  .565,  .447}42\% & \cellcolor[rgb]{ .973,  .412,  .42}50\% \\
    \hline
    \textbf{k-OWA-CVaR-100} & \cellcolor[rgb]{ .992,  .839,  .502}30\% & \cellcolor[rgb]{ .435,  .757,  .482}21\% & \cellcolor[rgb]{ .996,  .91,  .514}-20\% & \cellcolor[rgb]{ .996,  .898,  .514}27\% & \cellcolor[rgb]{ .996,  .863,  .506}36\% & \cellcolor[rgb]{ .984,  .663,  .467}43\% & \cellcolor[rgb]{ .973,  .475,  .431}51\% \\
    \hline
    \textbf{k-OWA-CCVaR-5} & \cellcolor[rgb]{ .478,  .773,  .49}39\% & \cellcolor[rgb]{ 1,  .914,  .518}24\% & \cellcolor[rgb]{ .518,  .784,  .49}-14\% & \cellcolor[rgb]{ .565,  .796,  .494}32\% & \cellcolor[rgb]{ .592,  .804,  .494}42\% & \cellcolor[rgb]{ .573,  .8,  .494}51\% & \cellcolor[rgb]{ .494,  .776,  .49}75\% \\
    \hline
    \textbf{k-OWA-CCVaR-10} & \cellcolor[rgb]{ .431,  .761,  .486}39\% & \cellcolor[rgb]{ 1,  .894,  .514}24\% & \cellcolor[rgb]{ .608,  .808,  .498}-15\% & \cellcolor[rgb]{ .388,  .745,  .482}34\% & \cellcolor[rgb]{ .475,  .773,  .49}43\% & \cellcolor[rgb]{ .506,  .78,  .49}52\% & \cellcolor[rgb]{ .486,  .776,  .49}75\% \\
    \hline
    \textbf{k-OWA-CCVaR-25} & \cellcolor[rgb]{ .702,  .835,  .502}36\% & \cellcolor[rgb]{ .82,  .867,  .506}23\% & \cellcolor[rgb]{ .745,  .851,  .506}-17\% & \cellcolor[rgb]{ .725,  .843,  .502}30\% & \cellcolor[rgb]{ .788,  .863,  .506}40\% & \cellcolor[rgb]{ .925,  .902,  .514}47\% & \cellcolor[rgb]{ .69,  .835,  .502}69\% \\
    \hline
    \textbf{k-OWA-CCVaR-50} & \cellcolor[rgb]{ .824,  .871,  .51}34\% & \cellcolor[rgb]{ .929,  .898,  .51}23\% & \cellcolor[rgb]{ .886,  .89,  .514}-18\% & \cellcolor[rgb]{ .773,  .859,  .506}30\% & \cellcolor[rgb]{ .925,  .902,  .514}38\% & \cellcolor[rgb]{ .996,  .906,  .514}46\% & \cellcolor[rgb]{ .808,  .867,  .51}66\% \\
    \hline
    \textbf{k-OWA-CCVaR-75} & \cellcolor[rgb]{ .996,  .878,  .51}31\% & \cellcolor[rgb]{ 1,  .91,  .518}24\% & \cellcolor[rgb]{ .992,  .8,  .494}-24\% & \cellcolor[rgb]{ .996,  .859,  .502}26\% & \cellcolor[rgb]{ .996,  .886,  .51}37\% & \cellcolor[rgb]{ .988,  .918,  .518}46\% & \cellcolor[rgb]{ .988,  .753,  .482}57\% \\
    \hline
    \textbf{k-OWA-CCVaR-100} & \cellcolor[rgb]{ .988,  .729,  .478}28\% & \cellcolor[rgb]{ .773,  .855,  .502}22\% & \cellcolor[rgb]{ .988,  .737,  .482}-26\% & \cellcolor[rgb]{ .992,  .82,  .498}25\% & \cellcolor[rgb]{ .988,  .765,  .486}34\% & \cellcolor[rgb]{ .984,  .671,  .467}43\% & \cellcolor[rgb]{ .973,  .463,  .427}51\% \\
    \hline
    \textbf{MinV} & \cellcolor[rgb]{ .984,  .639,  .463}26\% & \cellcolor[rgb]{ .388,  .745,  .482}20\% & \cellcolor[rgb]{ .729,  .843,  .502}-16\% & \cellcolor[rgb]{ .976,  .514,  .439}18\% & \cellcolor[rgb]{ .973,  .447,  .424}27\% & \cellcolor[rgb]{ .984,  .663,  .467}43\% & \cellcolor[rgb]{ .976,  .549,  .443}53\% \\
    \hline
    \textbf{MinVsh} & \cellcolor[rgb]{ .941,  .906,  .518}33\% & \cellcolor[rgb]{ .973,  .412,  .42}31\% & \cellcolor[rgb]{ .973,  .412,  .42}-38\% & \cellcolor[rgb]{ .929,  .902,  .514}28\% & \cellcolor[rgb]{ .388,  .745,  .482}44\% & \cellcolor[rgb]{ .388,  .745,  .482}53\% & \cellcolor[rgb]{ .973,  .914,  .518}62\% \\
    \hline
    \textbf{MinVN1} & \cellcolor[rgb]{ .973,  .412,  .42}22\% & \cellcolor[rgb]{ .996,  .808,  .498}25\% & \cellcolor[rgb]{ .976,  .533,  .439}-34\% & \cellcolor[rgb]{ .973,  .412,  .42}16\% & \cellcolor[rgb]{ .973,  .412,  .42}26\% & \cellcolor[rgb]{ .973,  .412,  .42}40\% & \cellcolor[rgb]{ .976,  .51,  .435}52\% \\
    \hline
    \textbf{MinVN2} & \cellcolor[rgb]{ .988,  .722,  .478}28\% & \cellcolor[rgb]{ 1,  .918,  .518}24\% & \cellcolor[rgb]{ .98,  .624,  .459}-30\% & \cellcolor[rgb]{ .992,  .8,  .494}24\% & \cellcolor[rgb]{ .992,  .808,  .494}35\% & \cellcolor[rgb]{ .984,  .682,  .471}43\% & \cellcolor[rgb]{ .976,  .49,  .435}51\% \\
    \hline
    \textbf{Index} & \cellcolor[rgb]{ .992,  .824,  .498}30\% & \cellcolor[rgb]{ .608,  .808,  .494}22\% & \cellcolor[rgb]{ .992,  .847,  .502}-22\% & \cellcolor[rgb]{ .988,  .769,  .486}24\% & \cellcolor[rgb]{ .984,  .69,  .471}32\% & \cellcolor[rgb]{ .98,  .561,  .447}42\% & \cellcolor[rgb]{ .996,  .882,  .51}60\% \\
    \hline
    \end{tabular}%
    }
    }
  \label{tab:ROI750_out_SP500_DL125_DH20_new_C}%
\end{table}%

\newpage

\section{Conclusions \label{sec:Conclusions}}

This paper has analyzed portfolio selection strategies for EI that are based on SD relations.
We have introduced a new type of approximate stochastic dominance rule, named $C S \mathcal{E} S D$, and we have proved that it implies the almost SSD proposed by \cite{lizyayev2012tractable}.
We have proposed a new methodology that select portfolios using the OWA  operator which generalizes previous approaches based on minimax selection rules, and still leads to solving linear programming models.
Next, we have proved that our enhanced indexation model based on OWA selects portfolios that dominate the benchmarks in terms of this new form of stochastic dominance rule, and it highlights encouraging out-of-sample performances with respect to a benchmark index.
Finally, we have tested and validated  the performance of our new proposal reporting an extensive empirical analysis based on three real-world datasets belonging to major stock markets across the world (FTSE100, NASDAQ100 and SP500) and on the Fama and French 49 industry portfolios.

This paper has also opened new avenues of research into the EI field. In a follow-up paper, we propose to extend the analysis to portfolio selection strategies based on approximate stochastic dominance conditions by considering an appropriate reshaping of the original benchmark index, as suggested by \cite{valle2017novel} and \cite{cesarone2022comparing}.

%\newpage

{\footnotesize
\bibliographystyle{apa}
\bibliography{BibbaseCPM_SSD_20220523}
}

\appendix

\section{SD-based strategies using Tails}\label{sec:OWATails}

For the sake of completeness, we report here additional tables containing the
computational results obtained by the OWA maximization of centered Tails approach along with the RomanTail and KP2011Min strategies, described in Section \ref{sec:EmpiricalAnalysis}. These portfolios are also compared with the MinV portfolio and the Market Index.
More precisely, in Tables \ref{tab:Perf_out_FTSE100_DL125_DH20}, \ref{tab:Perf_out_NASDAQ100_DL125_DH20}, and \ref{tab:Perf_out_SP500_DL125_DH20},  we show the out-of-sample performance results obtained on the FTSE100, NASDAQ100, and SP500 datasets, respectively. %from
From these empirical findings we can observe three main aspects. On the one hand, as also highlighted by \cite{roman2013enhanced}, the portfolios obtained with the SD models using centered CVaRs, except for FTSE100, lead to better performance.
On the other hand, we can notice that the $k$-OWA-Tail-$\beta$ and $k$-OWA-CTail-$\beta$ strategies with low $\beta$ typically continue to show the best performance in terms of profitability.
Finally, the risk mitigation feature of our approach, obtained by increasing the $\beta$ values, is no longer found in the ``Tails'' case.
%
% Table generated by Excel2LaTeX from sheet 'test (2)'
\begin{table}[htbp]
  \centering
  \caption{Out-of-sample performance results of the SD-based strategies using Tails compared with the MinV portfolio and the Market Index on the FTSE100 daily dataset}
  \scalebox{0.65}{
  {\renewcommand{\arraystretch}{1.1}
        \begin{tabular}{|l|c|c|c|c|c|c|c|c|c|c|c|c|}
    \hline
    \textbf{Approach} & \textbf{ExpRet} & \textbf{Vol} & \textbf{Sharpe} & \textbf{MDD} & \textbf{Ulcer} & \textbf{Sortino} & \textbf{Rachev} & \textbf{AlphaJ} & \textbf{VaR1} & \textbf{Omega} & \textbf{ave \#} & \textbf{time} \\
    \hline
    \textbf{Index} & \cellcolor[rgb]{ .973,  .412,  .42}0.0083\% & \cellcolor[rgb]{ .973,  .412,  .42}1.2117\% & \cellcolor[rgb]{ .973,  .412,  .42}0.0069 & \cellcolor[rgb]{ .98,  .58,  .451}-0.4783 & \cellcolor[rgb]{ .973,  .412,  .42}0.1549 & \cellcolor[rgb]{ .973,  .412,  .42}0.0095 & \cellcolor[rgb]{ .973,  .412,  .42}0.9257 &       & \cellcolor[rgb]{ .973,  .412,  .42}0.0345 & \cellcolor[rgb]{ .973,  .412,  .42}1.0210 &       &  \\
    \hline
    \textbf{Roman-Tail} & \cellcolor[rgb]{ .451,  .765,  .486}0.0541\% & \cellcolor[rgb]{ .992,  .718,  .478}1.0800\% & \cellcolor[rgb]{ .58,  .804,  .494}0.0501 & \cellcolor[rgb]{ .973,  .412,  .42}-0.5032 & \cellcolor[rgb]{ 1,  .894,  .514}0.1140 & \cellcolor[rgb]{ .494,  .776,  .49}0.0742 & \cellcolor[rgb]{ .388,  .745,  .482}1.0280 & \cellcolor[rgb]{ .455,  .765,  .486}0.00049 & \cellcolor[rgb]{ .992,  .773,  .49}0.0293 & \cellcolor[rgb]{ .502,  .78,  .49}1.1630 & 10.5  & 30.0 \\
    \hline
    \textbf{k-OWA-Tail-5} & \cellcolor[rgb]{ .537,  .788,  .494}0.0520\% & \cellcolor[rgb]{ 1,  .882,  .51}1.0080\% & \cellcolor[rgb]{ .506,  .78,  .49}0.0516 & \cellcolor[rgb]{ .973,  .463,  .427}-0.4956 & \cellcolor[rgb]{ .996,  .784,  .494}0.1232 & \cellcolor[rgb]{ .514,  .784,  .49}0.0736 & \cellcolor[rgb]{ .839,  .878,  .51}0.9736 & \cellcolor[rgb]{ .541,  .788,  .494}0.00047 & \cellcolor[rgb]{ 1,  .855,  .506}0.0282 & \cellcolor[rgb]{ .529,  .788,  .494}1.1611 & 11.5  & 24.4 \\
    \hline
    \textbf{k-OWA-Tail-10} & \cellcolor[rgb]{ .741,  .847,  .506}0.0472\% & \cellcolor[rgb]{ .839,  .875,  .506}0.9732\% & \cellcolor[rgb]{ .659,  .824,  .498}0.0485 & \cellcolor[rgb]{ .992,  .816,  .494}-0.4436 & \cellcolor[rgb]{ .914,  .894,  .51}0.1085 & \cellcolor[rgb]{ .675,  .831,  .502}0.0686 & \cellcolor[rgb]{ .929,  .902,  .514}0.9628 & \cellcolor[rgb]{ .745,  .851,  .506}0.00042 & \cellcolor[rgb]{ .859,  .878,  .506}0.0270 & \cellcolor[rgb]{ .694,  .835,  .502}1.1497 & 12.6  & 32.4 \\
    \hline
    \textbf{k-OWA-Tail-25} & \cellcolor[rgb]{ .996,  .918,  .514}0.0409\% & \cellcolor[rgb]{ .878,  .886,  .51}0.9774\% & \cellcolor[rgb]{ .988,  .922,  .518}0.0418 & \cellcolor[rgb]{ .953,  .91,  .518}-0.4246 & \cellcolor[rgb]{ .996,  .918,  .514}0.1113 & \cellcolor[rgb]{ .988,  .922,  .518}0.0589 & \cellcolor[rgb]{ .992,  .792,  .49}0.9467 & \cellcolor[rgb]{ .996,  .875,  .506}0.00036 & \cellcolor[rgb]{ .706,  .835,  .498}0.0268 & \cellcolor[rgb]{ .992,  .922,  .518}1.1291 & 13.6  & 36.9 \\
    \hline
    \textbf{k-OWA-Tail-50} & \cellcolor[rgb]{ .996,  .894,  .51}0.0395\% & \cellcolor[rgb]{ .98,  .914,  .514}0.9882\% & \cellcolor[rgb]{ .996,  .898,  .51}0.0400 & \cellcolor[rgb]{ .706,  .839,  .502}-0.4049 & \cellcolor[rgb]{ .941,  .902,  .514}0.1095 & \cellcolor[rgb]{ .996,  .898,  .51}0.0564 & \cellcolor[rgb]{ .992,  .922,  .518}0.9552 & \cellcolor[rgb]{ .988,  .733,  .478}0.00034 & \cellcolor[rgb]{ .58,  .8,  .49}0.0267 & \cellcolor[rgb]{ .996,  .894,  .51}1.1232 & 14.7  & 47.7 \\
    \hline
    \textbf{k-OWA-Tail-75} & \cellcolor[rgb]{ .996,  .914,  .514}0.0407\% & \cellcolor[rgb]{ 1,  .91,  .518}0.9959\% & \cellcolor[rgb]{ .996,  .91,  .514}0.0408 & \cellcolor[rgb]{ .796,  .863,  .506}-0.4120 & \cellcolor[rgb]{ 1,  .922,  .518}0.1114 & \cellcolor[rgb]{ .996,  .91,  .514}0.0576 & \cellcolor[rgb]{ .996,  .882,  .51}0.9519 & \cellcolor[rgb]{ .992,  .843,  .502}0.00035 & \cellcolor[rgb]{ 1,  .886,  .514}0.0277 & \cellcolor[rgb]{ .996,  .91,  .514}1.1265 & 15.0  & 61.7 \\
    \hline
    \textbf{k-OWA-Tail-100} & \cellcolor[rgb]{ .969,  .914,  .518}0.0419\% & \cellcolor[rgb]{ 1,  .871,  .51}1.0138\% & \cellcolor[rgb]{ .996,  .918,  .514}0.0413 & \cellcolor[rgb]{ .996,  .894,  .51}-0.4320 & \cellcolor[rgb]{ .996,  .835,  .502}0.1190 & \cellcolor[rgb]{ .996,  .914,  .514}0.0580 & \cellcolor[rgb]{ .992,  .804,  .494}0.9475 & \cellcolor[rgb]{ .988,  .918,  .518}0.00036 & \cellcolor[rgb]{ 1,  .855,  .506}0.0282 & \cellcolor[rgb]{ .996,  .918,  .514}1.1279 & 14.3  & 69.7 \\
    \hline
    \textbf{k-OWA-CTail-5} & \cellcolor[rgb]{ .388,  .745,  .482}0.0555\% & \cellcolor[rgb]{ .996,  .835,  .502}1.0285\% & \cellcolor[rgb]{ .388,  .745,  .482}0.0540 & \cellcolor[rgb]{ .976,  .514,  .439}-0.4881 & \cellcolor[rgb]{ .996,  .847,  .506}0.1178 & \cellcolor[rgb]{ .388,  .745,  .482}0.0774 & \cellcolor[rgb]{ .729,  .847,  .506}0.9867 & \cellcolor[rgb]{ .388,  .745,  .482}0.00050 & \cellcolor[rgb]{ .996,  .831,  .502}0.0285 & \cellcolor[rgb]{ .388,  .745,  .482}1.1706 & 11.2  & 17.9 \\
    \hline
    \textbf{k-OWA-CTail-10} & \cellcolor[rgb]{ .608,  .808,  .498}0.0504\% & \cellcolor[rgb]{ .988,  .918,  .514}0.9893\% & \cellcolor[rgb]{ .541,  .792,  .494}0.0509 & \cellcolor[rgb]{ .988,  .725,  .478}-0.4570 & \cellcolor[rgb]{ .996,  .847,  .506}0.1180 & \cellcolor[rgb]{ .549,  .792,  .494}0.0725 & \cellcolor[rgb]{ .914,  .898,  .514}0.9646 & \cellcolor[rgb]{ .608,  .812,  .498}0.00045 & \cellcolor[rgb]{ .843,  .875,  .506}0.0270 & \cellcolor[rgb]{ .573,  .8,  .494}1.1582 & 12.1  & 17.4 \\
    \hline
    \textbf{k-OWA-CTail-25} & \cellcolor[rgb]{ .996,  .922,  .518}0.0412\% & \cellcolor[rgb]{ .757,  .851,  .502}0.9649\% & \cellcolor[rgb]{ .945,  .906,  .518}0.0427 & \cellcolor[rgb]{ .996,  .871,  .506}-0.4354 & \cellcolor[rgb]{ .992,  .918,  .514}0.1111 & \cellcolor[rgb]{ .949,  .91,  .518}0.0600 & \cellcolor[rgb]{ .988,  .749,  .482}0.9443 & \cellcolor[rgb]{ 1,  .922,  .518}0.00036 & \cellcolor[rgb]{ .4,  .749,  .482}0.0264 & \cellcolor[rgb]{ .961,  .91,  .518}1.1314 & 13.4  & 17.1 \\
    \hline
    \textbf{k-OWA-CTail-50} & \cellcolor[rgb]{ .996,  .898,  .514}0.0398\% & \cellcolor[rgb]{ .875,  .882,  .51}0.9770\% & \cellcolor[rgb]{ .996,  .906,  .514}0.0407 & \cellcolor[rgb]{ .635,  .82,  .498}-0.3995 & \cellcolor[rgb]{ .835,  .871,  .506}0.1059 & \cellcolor[rgb]{ .996,  .91,  .514}0.0573 & \cellcolor[rgb]{ .992,  .82,  .498}0.9483 & \cellcolor[rgb]{ .988,  .769,  .486}0.00035 & \cellcolor[rgb]{ .388,  .745,  .482}0.0264 & \cellcolor[rgb]{ .996,  .906,  .514}1.1253 & 14.2  & 34.2 \\
    \hline
    \textbf{k-OWA-CTail-75} & \cellcolor[rgb]{ .996,  .898,  .51}0.0397\% & \cellcolor[rgb]{ .937,  .902,  .514}0.9837\% & \cellcolor[rgb]{ .996,  .902,  .514}0.0403 & \cellcolor[rgb]{ .686,  .831,  .502}-0.4034 & \cellcolor[rgb]{ .925,  .898,  .51}0.1089 & \cellcolor[rgb]{ .996,  .902,  .514}0.0569 & \cellcolor[rgb]{ .996,  .898,  .51}0.9526 & \cellcolor[rgb]{ .988,  .749,  .482}0.00034 & \cellcolor[rgb]{ .455,  .761,  .482}0.0265 & \cellcolor[rgb]{ .996,  .898,  .514}1.1243 & 14.6  & 78.9 \\
    \hline
    \textbf{k-OWA-CTail-100} & \cellcolor[rgb]{ .996,  .91,  .514}0.0403\% & \cellcolor[rgb]{ 1,  .922,  .518}0.9911\% & \cellcolor[rgb]{ .996,  .906,  .514}0.0407 & \cellcolor[rgb]{ .757,  .855,  .506}-0.4091 & \cellcolor[rgb]{ 1,  .922,  .518}0.1115 & \cellcolor[rgb]{ .996,  .91,  .514}0.0574 & \cellcolor[rgb]{ .996,  .894,  .51}0.9525 & \cellcolor[rgb]{ .992,  .808,  .494}0.00035 & \cellcolor[rgb]{ 1,  .914,  .518}0.0273 & \cellcolor[rgb]{ .996,  .906,  .514}1.1257 & 14.8  & 132.7 \\
    \hline
    \textbf{MinV} & \cellcolor[rgb]{ .992,  .843,  .502}0.0362\% & \cellcolor[rgb]{ .388,  .745,  .482}0.9254\% & \cellcolor[rgb]{ .996,  .882,  .51}0.0391 & \cellcolor[rgb]{ .388,  .745,  .482}-0.3802 & \cellcolor[rgb]{ .388,  .745,  .482}0.0909 & \cellcolor[rgb]{ .996,  .878,  .506}0.0545 & \cellcolor[rgb]{ .969,  .914,  .518}0.9577 & \cellcolor[rgb]{ .973,  .412,  .42}0.00031 & \cellcolor[rgb]{ .902,  .894,  .51}0.0271 & \cellcolor[rgb]{ .996,  .886,  .51}1.1215 & 17.9  & 0.0 \\
    \hline
    \textbf{KP2011Min} & \cellcolor[rgb]{ .651,  .824,  .498}0.0494\% & \cellcolor[rgb]{ 1,  .867,  .51}  0.0102  & \cellcolor[rgb]{ .655,  .824,  .498} 0.0486  & \cellcolor[rgb]{ .894,  .894,  .514}-0.4199  & \cellcolor[rgb]{ .863,  .882,  .51} 0.1069  & \cellcolor[rgb]{ .643,  .82,  .498} 0.0695  & \cellcolor[rgb]{ .992,  .922,  .518}  0.9551  & \cellcolor[rgb]{ .655,  .824,  .498} 0.00044  & \cellcolor[rgb]{ .996,  .808,  .498} 0.0289  & \cellcolor[rgb]{ .631,  .816,  .498} 1.1541  & 13.1  & 51.4 \\
    \hline
    \end{tabular}%
    }
    }
  \label{tab:Perf_out_FTSE100_DL125_DH20}%
\end{table}%
%
%
% Table generated by Excel2LaTeX from sheet 'test (4)'
\begin{table}[htbp]
  \centering
  \caption{Out-of-sample performance results of the SD-based strategies using Tails compared with the MinV portfolio and the Market Index on the NASDAQ100 daily dataset}
  \scalebox{0.65}{
  {\renewcommand{\arraystretch}{1.1}
        \begin{tabular}{|l|r|r|r|r|r|r|r|r|r|r|r|r|}
    \hline
    \textbf{Approach} & \multicolumn{1}{l|}{\textbf{ExpRet}} & \multicolumn{1}{l|}{\textbf{Vol}} & \multicolumn{1}{l|}{\textbf{Sharpe}} & \multicolumn{1}{l|}{\textbf{MDD}} & \multicolumn{1}{l|}{\textbf{Ulcer}} & \multicolumn{1}{l|}{\textbf{Sortino}} & \multicolumn{1}{l|}{\textbf{Rachev5}} & \multicolumn{1}{l|}{\textbf{AlphaJ}} & \multicolumn{1}{l|}{\textbf{VaR1}} & \multicolumn{1}{l|}{\textbf{Omega}} & \multicolumn{1}{l|}{\textbf{ave \#}} & \multicolumn{1}{l|}{\textbf{time}} \\
    \hline
    \textbf{Index} & \cellcolor[rgb]{ .388,  .745,  .482}0.0651\% & \cellcolor[rgb]{ .973,  .412,  .42}1.4093\% & \cellcolor[rgb]{ .996,  .871,  .506}0.0462 & \cellcolor[rgb]{ .973,  .412,  .42}-0.5371 & \cellcolor[rgb]{ .973,  .412,  .42}0.1339 & \cellcolor[rgb]{ .992,  .839,  .498}0.0654 & \cellcolor[rgb]{ .996,  .871,  .506}0.9481 &       & \cellcolor[rgb]{ .973,  .412,  .42}0.0407 & \cellcolor[rgb]{ .996,  .906,  .514}1.1523 &       &  \\
    \hline
    \textbf{Roman-Tail} & \cellcolor[rgb]{ .996,  .855,  .502}0.0522\% & \cellcolor[rgb]{ 1,  .886,  .514}1.1618\% & \cellcolor[rgb]{ .988,  .769,  .486}0.0450 & \cellcolor[rgb]{ .988,  .733,  .478}-0.4293 & \cellcolor[rgb]{ .996,  .792,  .494}0.1081 & \cellcolor[rgb]{ .992,  .812,  .494}0.0649 & \cellcolor[rgb]{ .388,  .745,  .482}0.9742 & \cellcolor[rgb]{ .992,  .808,  .494} 0.00009  & \cellcolor[rgb]{ 1,  .875,  .51}0.0329 & \cellcolor[rgb]{ .988,  .722,  .478}1.1443 & 7.9   & 15.9 \\
    \hline
    \textbf{k-OWA-Tail-5} & \cellcolor[rgb]{ .996,  .894,  .51}0.0529\% & \cellcolor[rgb]{ 1,  .922,  .518}1.1439\% & \cellcolor[rgb]{ .996,  .882,  .51}0.0463 & \cellcolor[rgb]{ .988,  .753,  .482}-0.4231 & \cellcolor[rgb]{ .992,  .773,  .49}0.1094 & \cellcolor[rgb]{ .996,  .91,  .514}0.0666 & \cellcolor[rgb]{ .443,  .761,  .486}0.9722 & \cellcolor[rgb]{ .996,  .878,  .51} 0.00010  & \cellcolor[rgb]{ .882,  .886,  .51}0.0314 & \cellcolor[rgb]{ .992,  .792,  .49}1.1474 & 8.3   & 16.8 \\
    \hline
    \textbf{k-OWA-Tail-10} & \cellcolor[rgb]{ .627,  .816,  .498}0.0606\% & \cellcolor[rgb]{ 1,  .918,  .518}1.1446\% & \cellcolor[rgb]{ .388,  .745,  .482}0.0529 & \cellcolor[rgb]{ .925,  .902,  .514}-0.3657 & \cellcolor[rgb]{ .388,  .745,  .482}0.0824 & \cellcolor[rgb]{ .388,  .745,  .482}0.0760 & \cellcolor[rgb]{ .475,  .773,  .49}0.9710 & \cellcolor[rgb]{ .388,  .745,  .482} 0.00018  & \cellcolor[rgb]{ .847,  .875,  .506}0.0312 & \cellcolor[rgb]{ .388,  .745,  .482}1.1706 & 9.0   & 22.2 \\
    \hline
    \textbf{k-OWA-Tail-25} & \cellcolor[rgb]{ .898,  .894,  .514}0.0554\% & \cellcolor[rgb]{ 1,  .898,  .514}1.1547\% & \cellcolor[rgb]{ .878,  .886,  .514}0.0480 & \cellcolor[rgb]{ .957,  .91,  .518}-0.3667 & \cellcolor[rgb]{ .871,  .882,  .51}0.0957 & \cellcolor[rgb]{ .886,  .89,  .514}0.0686 & \cellcolor[rgb]{ .8,  .867,  .51}0.9584 & \cellcolor[rgb]{ .847,  .878,  .51} 0.00013  & \cellcolor[rgb]{ 1,  .894,  .514}0.0326 & \cellcolor[rgb]{ .922,  .898,  .514}1.1553 & 9.8   & 27.6 \\
    \hline
    \textbf{k-OWA-Tail-50} & \cellcolor[rgb]{ .969,  .914,  .518}0.0540\% & \cellcolor[rgb]{ .976,  .914,  .514}1.1368\% & \cellcolor[rgb]{ .925,  .902,  .514}0.0475 & \cellcolor[rgb]{ .878,  .89,  .514}-0.3643 & \cellcolor[rgb]{ 1,  .91,  .518}0.1002 & \cellcolor[rgb]{ .945,  .906,  .518}0.0677 & \cellcolor[rgb]{ .996,  .875,  .506}0.9484 & \cellcolor[rgb]{ .961,  .91,  .518} 0.00012  & \cellcolor[rgb]{ .992,  .918,  .514}0.0320 & \cellcolor[rgb]{ .957,  .91,  .518}1.1542 & 11.1  & 36.8 \\
    \hline
    \textbf{k-OWA-Tail-75} & \cellcolor[rgb]{ .996,  .875,  .506}0.0525\% & \cellcolor[rgb]{ .91,  .894,  .51}1.1219\% & \cellcolor[rgb]{ .988,  .922,  .518}0.0468 & \cellcolor[rgb]{ .882,  .89,  .514}-0.3643 & \cellcolor[rgb]{ 1,  .922,  .518}0.0994 & \cellcolor[rgb]{ .996,  .906,  .514}0.0666 & \cellcolor[rgb]{ .984,  .678,  .471}0.9383 & \cellcolor[rgb]{ .996,  .875,  .506} 0.00010  & \cellcolor[rgb]{ 1,  .914,  .518}0.0322 & \cellcolor[rgb]{ .996,  .894,  .51}1.1519 & 11.7  & 49.3 \\
    \hline
    \textbf{k-OWA-Tail-100} & \cellcolor[rgb]{ .992,  .8,  .494}0.0512\% & \cellcolor[rgb]{ .925,  .898,  .51}1.1251\% & \cellcolor[rgb]{ .992,  .82,  .498}0.0455 & \cellcolor[rgb]{ .702,  .835,  .502}-0.3585 & \cellcolor[rgb]{ .867,  .882,  .51}0.0956 & \cellcolor[rgb]{ .992,  .804,  .494}0.0647 & \cellcolor[rgb]{ .98,  .612,  .455}0.9349 & \cellcolor[rgb]{ .992,  .78,  .49} 0.00009  & \cellcolor[rgb]{ 1,  .922,  .518}0.0321 & \cellcolor[rgb]{ .992,  .804,  .494}1.1479 & 11.6  & 43.4 \\
    \hline
    \textbf{k-OWA-CTail-5} & \cellcolor[rgb]{ .996,  .914,  .514}0.0533\% & \cellcolor[rgb]{ 1,  .922,  .518}1.1433\% & \cellcolor[rgb]{ .996,  .91,  .514}0.0466 & \cellcolor[rgb]{ .988,  .769,  .486}-0.4185 & \cellcolor[rgb]{ .996,  .796,  .494}0.1079 & \cellcolor[rgb]{ .988,  .922,  .518}0.0670 & \cellcolor[rgb]{ .518,  .784,  .49}0.9694 & \cellcolor[rgb]{ 1,  .922,  .518} 0.00011  & \cellcolor[rgb]{ 1,  .918,  .518}0.0322 & \cellcolor[rgb]{ .992,  .804,  .494}1.1479 & 8.2   & 12.7 \\
    \hline
    \textbf{k-OWA-CTail-10} & \cellcolor[rgb]{ .808,  .867,  .51}0.0571\% & \cellcolor[rgb]{ .992,  .918,  .514}1.1406\% & \cellcolor[rgb]{ .671,  .827,  .502}0.0500 & \cellcolor[rgb]{ .992,  .843,  .502}-0.3935 & \cellcolor[rgb]{ .882,  .886,  .51}0.0961 & \cellcolor[rgb]{ .655,  .824,  .498}0.0720 & \cellcolor[rgb]{ .455,  .765,  .486}0.9718 & \cellcolor[rgb]{ .671,  .827,  .502} 0.00015  & \cellcolor[rgb]{ .843,  .875,  .506}0.0312 & \cellcolor[rgb]{ .757,  .855,  .506}1.1600 & 8.7   & 13.2 \\
    \hline
    \textbf{k-OWA-CTail-25} & \cellcolor[rgb]{ .941,  .906,  .518}0.0545\% & \cellcolor[rgb]{ 1,  .918,  .518}1.1453\% & \cellcolor[rgb]{ .914,  .898,  .514}0.0476 & \cellcolor[rgb]{ .996,  .89,  .51}-0.3784 & \cellcolor[rgb]{ .937,  .902,  .514}0.0975 & \cellcolor[rgb]{ .914,  .898,  .514}0.0681 & \cellcolor[rgb]{ .694,  .835,  .502}0.9625 & \cellcolor[rgb]{ .886,  .89,  .514} 0.00012  & \cellcolor[rgb]{ 1,  .89,  .514}0.0326 & \cellcolor[rgb]{ .98,  .918,  .518}1.1536 & 9.4   & 15.2 \\
    \hline
    \textbf{k-OWA-CTail-50} & \cellcolor[rgb]{ .996,  .918,  .514}0.0533\% & \cellcolor[rgb]{ 1,  .918,  .518}1.1446\% & \cellcolor[rgb]{ .996,  .91,  .514}0.0466 & \cellcolor[rgb]{ .996,  .914,  .514}-0.3695 & \cellcolor[rgb]{ 1,  .894,  .514}0.1012 & \cellcolor[rgb]{ .996,  .902,  .514}0.0665 & \cellcolor[rgb]{ .949,  .91,  .518}0.9528 & \cellcolor[rgb]{ .996,  .906,  .514} 0.00011  & \cellcolor[rgb]{ .929,  .898,  .51}0.0317 & \cellcolor[rgb]{ .996,  .875,  .506}1.1509 & 10.4  & 31.6 \\
    \hline
    \textbf{k-OWA-CTail-75} & \cellcolor[rgb]{ .961,  .91,  .518}0.0542\% & \cellcolor[rgb]{ .965,  .91,  .514}1.1343\% & \cellcolor[rgb]{ .898,  .894,  .514}0.0477 & \cellcolor[rgb]{ .945,  .906,  .518}-0.3664 & \cellcolor[rgb]{ .992,  .918,  .514}0.0990 & \cellcolor[rgb]{ .918,  .898,  .514}0.0681 & \cellcolor[rgb]{ .996,  .878,  .51}0.9487 & \cellcolor[rgb]{ .937,  .906,  .518} 0.00012  & \cellcolor[rgb]{ .945,  .906,  .514}0.0318 & \cellcolor[rgb]{ .922,  .902,  .514}1.1552 & 11.1  & 57.8 \\
    \hline
    \textbf{k-OWA-CTail-100} & \cellcolor[rgb]{ 1,  .922,  .518}0.0534\% & \cellcolor[rgb]{ .933,  .902,  .514}1.1269\% & \cellcolor[rgb]{ .933,  .906,  .518}0.0474 & \cellcolor[rgb]{ .898,  .894,  .514}-0.3648 & \cellcolor[rgb]{ .988,  .918,  .514}0.0989 & \cellcolor[rgb]{ .961,  .91,  .518}0.0674 & \cellcolor[rgb]{ .992,  .78,  .49}0.9435 & \cellcolor[rgb]{ .988,  .922,  .518} 0.00011  & \cellcolor[rgb]{ 1,  .91,  .518}0.0323 & \cellcolor[rgb]{ .969,  .914,  .518}1.1539 & 11.5  & 97.3 \\
    \hline
    \textbf{MinV} & \cellcolor[rgb]{ .976,  .537,  .443}0.0465\% & \cellcolor[rgb]{ .388,  .745,  .482}0.9999\% & \cellcolor[rgb]{ .996,  .902,  .514}0.0465 & \cellcolor[rgb]{ .388,  .745,  .482}-0.3486 & \cellcolor[rgb]{ .537,  .784,  .49}0.0865 & \cellcolor[rgb]{ .996,  .878,  .506}0.0661 & \cellcolor[rgb]{ .98,  .627,  .459}0.9357 & \cellcolor[rgb]{ .992,  .839,  .502} 0.00010  & \cellcolor[rgb]{ .388,  .745,  .482}0.0286 & \cellcolor[rgb]{ .976,  .918,  .518}1.1537 & 13.0  & 0.0 \\
    \hline
    \textbf{KP2011Min} & \cellcolor[rgb]{ .973,  .412,  .42}0.0442\% & \cellcolor[rgb]{ .745,  .847,  .502}1.0830\% & \cellcolor[rgb]{ .973,  .412,  .42}0.0408 & \cellcolor[rgb]{ .996,  .882,  .51}-0.3807 & \cellcolor[rgb]{ 1,  .914,  .518}0.1000 & \cellcolor[rgb]{ .973,  .412,  .42}0.0577 & \cellcolor[rgb]{ .973,  .412,  .42}0.9245 & \cellcolor[rgb]{ .973,  .412,  .42} 0.00004  & \cellcolor[rgb]{ .988,  .918,  .514}0.0320 & \cellcolor[rgb]{ .973,  .412,  .42}1.1310 & 10.6  & 47.5 \\
    \hline
    \end{tabular}%
    }
    }
  \label{tab:Perf_out_NASDAQ100_DL125_DH20}%
\end{table}%
%
% Table generated by Excel2LaTeX from sheet 'test (3)'
\begin{table}[htbp]
  \centering
  \caption{Out-of-sample performance results of the SD-based strategies using Tails compared with the MinV portfolio and the Market Index on the SP500 daily dataset}
  \scalebox{0.65}{
  {\renewcommand{\arraystretch}{1.1}
        \begin{tabular}{|l|r|r|r|r|r|r|r|r|r|r|r|r|}
    \hline
    \textbf{Approach} & \multicolumn{1}{l|}{\textbf{ExpRet}} & \multicolumn{1}{l|}{\textbf{Vol}} & \multicolumn{1}{l|}{\textbf{Sharpe}} & \multicolumn{1}{l|}{\textbf{MDD}} & \multicolumn{1}{l|}{\textbf{Ulcer}} & \multicolumn{1}{l|}{\textbf{Sortino}} & \multicolumn{1}{l|}{\textbf{Rachev5}} & \multicolumn{1}{l|}{\textbf{AlphaJ}} & \multicolumn{1}{l|}{\textbf{VaR1}} & \multicolumn{1}{l|}{\textbf{Omega}} & \multicolumn{1}{l|}{\textbf{ave \#}} & \multicolumn{1}{l|}{\textbf{time}} \\
    \hline
    \textbf{Index} & \cellcolor[rgb]{ .518,  .784,  .49}0.0352\% & \cellcolor[rgb]{ .973,  .412,  .42}1.2960\% & \cellcolor[rgb]{ .973,  .412,  .42}0.0272 & \cellcolor[rgb]{ .973,  .412,  .42}-0.5678 & \cellcolor[rgb]{ .984,  .918,  .514}0.1653 & \cellcolor[rgb]{ .973,  .412,  .42}0.0378 & \cellcolor[rgb]{ .388,  .745,  .482}0.9099 &       & \cellcolor[rgb]{ .973,  .412,  .42}0.0393 & \cellcolor[rgb]{ .973,  .412,  .42}1.0926 &       &  \\
    \hline
    \textbf{Roman-Tail} & \cellcolor[rgb]{ .996,  .871,  .506}0.0312\% & \cellcolor[rgb]{ 1,  .871,  .51}0.9921\% & \cellcolor[rgb]{ .992,  .82,  .498}0.0314 & \cellcolor[rgb]{ .992,  .804,  .494}-0.5120 & \cellcolor[rgb]{ .973,  .412,  .42}0.1950 & \cellcolor[rgb]{ .992,  .812,  .494}0.0435 & \cellcolor[rgb]{ .678,  .831,  .502}0.9018 & \cellcolor[rgb]{ .984,  .69,  .471} 0.00010  & \cellcolor[rgb]{ .898,  .89,  .51}0.0288 & \cellcolor[rgb]{ .992,  .788,  .49}1.1015 & 13.8  & 70.3 \\
    \hline
    \textbf{k-OWA-Tail-5} & \cellcolor[rgb]{ .992,  .784,  .49}0.0306\% & \cellcolor[rgb]{ 1,  .894,  .514}0.9754\% & \cellcolor[rgb]{ .992,  .812,  .494}0.0313 & \cellcolor[rgb]{ .996,  .878,  .506}-0.5016 & \cellcolor[rgb]{ .78,  .855,  .502}0.1520 & \cellcolor[rgb]{ .992,  .796,  .49}0.0432 & \cellcolor[rgb]{ .886,  .89,  .514}0.8960 & \cellcolor[rgb]{ .988,  .749,  .482} 0.00010  & \cellcolor[rgb]{ .988,  .918,  .514}0.0294 & \cellcolor[rgb]{ .988,  .714,  .475}1.0997 & 14.3  & 37.4 \\
    \hline
    \textbf{k-OWA-Tail-10} & \cellcolor[rgb]{ .91,  .898,  .514}0.0322\% & \cellcolor[rgb]{ .973,  .91,  .514}0.9535\% & \cellcolor[rgb]{ .859,  .882,  .51}0.0338 & \cellcolor[rgb]{ .835,  .875,  .51}-0.4853 & \cellcolor[rgb]{ .522,  .78,  .486}0.1353 & \cellcolor[rgb]{ .867,  .882,  .51}0.0468 & \cellcolor[rgb]{ .98,  .616,  .459}0.8819 & \cellcolor[rgb]{ .839,  .875,  .51} 0.00012  & \cellcolor[rgb]{ .624,  .812,  .494}0.0269 & \cellcolor[rgb]{ .929,  .902,  .514}1.1065 & 15.0  & 49.4 \\
    \hline
    \textbf{k-OWA-Tail-25} & \cellcolor[rgb]{ .388,  .745,  .482}0.0362\% & \cellcolor[rgb]{ .941,  .902,  .514}0.9500\% & \cellcolor[rgb]{ .388,  .745,  .482}0.0381 & \cellcolor[rgb]{ .663,  .827,  .502}-0.4748 & \cellcolor[rgb]{ .388,  .745,  .482}0.1266 & \cellcolor[rgb]{ .388,  .745,  .482}0.0532 & \cellcolor[rgb]{ .804,  .867,  .51}0.8983 & \cellcolor[rgb]{ .388,  .745,  .482} 0.00016  & \cellcolor[rgb]{ 1,  .902,  .514}0.0299 & \cellcolor[rgb]{ .388,  .745,  .482}1.1208 & 16.8  & 56.3 \\
    \hline
    \textbf{k-OWA-Tail-50} & \cellcolor[rgb]{ .961,  .91,  .518}0.0318\% & \cellcolor[rgb]{ .933,  .902,  .514}0.9488\% & \cellcolor[rgb]{ .882,  .89,  .514}0.0336 & \cellcolor[rgb]{ .996,  .914,  .514}-0.4964 & \cellcolor[rgb]{ .996,  .847,  .506}0.1706 & \cellcolor[rgb]{ .867,  .882,  .51}0.0468 & \cellcolor[rgb]{ .992,  .922,  .518}0.8931 & \cellcolor[rgb]{ .925,  .902,  .514} 0.00012  & \cellcolor[rgb]{ 1,  .918,  .518}0.0295 & \cellcolor[rgb]{ .953,  .91,  .518}1.1058 & 18.4  & 64.4 \\
    \hline
    \textbf{k-OWA-Tail-75} & \cellcolor[rgb]{ .988,  .757,  .482}0.0304\% & \cellcolor[rgb]{ 1,  .918,  .518}0.9602\% & \cellcolor[rgb]{ .992,  .839,  .502}0.0316 & \cellcolor[rgb]{ .937,  .906,  .518}-0.4916 & \cellcolor[rgb]{ .992,  .745,  .486}0.1763 & \cellcolor[rgb]{ .996,  .847,  .502}0.0440 & \cellcolor[rgb]{ .984,  .667,  .467}0.8836 & \cellcolor[rgb]{ .984,  .671,  .467} 0.00010  & \cellcolor[rgb]{ 1,  .918,  .518}0.0296 & \cellcolor[rgb]{ .988,  .718,  .478}1.0998 & 19.1  & 77.3 \\
    \hline
    \textbf{k-OWA-Tail-100} & \cellcolor[rgb]{ .988,  .765,  .486}0.0304\% & \cellcolor[rgb]{ 1,  .878,  .51}0.9866\% & \cellcolor[rgb]{ .988,  .765,  .486}0.0308 & \cellcolor[rgb]{ .965,  .914,  .518}-0.4935 & \cellcolor[rgb]{ .984,  .62,  .463}0.1833 & \cellcolor[rgb]{ .988,  .765,  .486}0.0428 & \cellcolor[rgb]{ .973,  .412,  .42}0.8745 & \cellcolor[rgb]{ .98,  .588,  .451} 0.00009  & \cellcolor[rgb]{ 1,  .922,  .518}0.0295 & \cellcolor[rgb]{ .98,  .6,  .455}1.0970 & 18.5  & 90.7 \\
    \hline
    \textbf{k-OWA-CTail-5} & \cellcolor[rgb]{ .988,  .765,  .486}0.0304\% & \cellcolor[rgb]{ 1,  .875,  .51}0.9902\% & \cellcolor[rgb]{ .988,  .753,  .482}0.0307 & \cellcolor[rgb]{ .992,  .816,  .494}-0.5104 & \cellcolor[rgb]{ 1,  .91,  .518}0.1669 & \cellcolor[rgb]{ .988,  .737,  .478}0.0424 & \cellcolor[rgb]{ .769,  .855,  .506}0.8992 & \cellcolor[rgb]{ .984,  .631,  .459} 0.00009  & \cellcolor[rgb]{ 1,  .89,  .514}0.0300 & \cellcolor[rgb]{ .984,  .651,  .463}1.0982 & 14.2  & 27.5 \\
    \hline
    \textbf{k-OWA-CTail-10} & \cellcolor[rgb]{ .992,  .82,  .498}0.0308\% & \cellcolor[rgb]{ 1,  .914,  .518}0.9623\% & \cellcolor[rgb]{ .996,  .878,  .506}0.0320 & \cellcolor[rgb]{ .988,  .757,  .486}-0.5185 & \cellcolor[rgb]{ .961,  .91,  .514}0.1638 & \cellcolor[rgb]{ .996,  .859,  .506}0.0442 & \cellcolor[rgb]{ .98,  .6,  .455}0.8813 & \cellcolor[rgb]{ .996,  .886,  .51} 0.00011  & \cellcolor[rgb]{ .91,  .894,  .51}0.0288 & \cellcolor[rgb]{ .988,  .769,  .486}1.1010 & 14.6  & 24.8 \\
    \hline
    \textbf{k-OWA-CTail-25} & \cellcolor[rgb]{ .98,  .918,  .518}0.0317\% & \cellcolor[rgb]{ .929,  .898,  .51}0.9484\% & \cellcolor[rgb]{ .902,  .894,  .514}0.0334 & \cellcolor[rgb]{ .851,  .878,  .51}-0.4864 & \cellcolor[rgb]{ .42,  .753,  .482}0.1287 & \cellcolor[rgb]{ .89,  .89,  .514}0.0465 & \cellcolor[rgb]{ .882,  .89,  .514}0.8961 & \cellcolor[rgb]{ .878,  .89,  .514} 0.00012  & \cellcolor[rgb]{ .851,  .878,  .506}0.0284 & \cellcolor[rgb]{ .98,  .918,  .518}1.1051 & 16.0  & 23.9 \\
    \hline
    \textbf{k-OWA-CTail-50} & \cellcolor[rgb]{ .769,  .855,  .506}0.0333\% & \cellcolor[rgb]{ .875,  .882,  .51}0.9415\% & \cellcolor[rgb]{ .686,  .831,  .502}0.0354 & \cellcolor[rgb]{ .984,  .918,  .518}-0.4946 & \cellcolor[rgb]{ .804,  .863,  .506}0.1537 & \cellcolor[rgb]{ .675,  .827,  .502}0.0494 & \cellcolor[rgb]{ .886,  .89,  .514}0.8960 & \cellcolor[rgb]{ .718,  .843,  .502} 0.00013  & \cellcolor[rgb]{ 1,  .918,  .518}0.0295 & \cellcolor[rgb]{ .725,  .843,  .502}1.1119 & 17.8  & 46.7 \\
    \hline
    \textbf{k-OWA-CTail-75} & \cellcolor[rgb]{ .98,  .918,  .518}0.0317\% & \cellcolor[rgb]{ .898,  .89,  .51}0.9445\% & \cellcolor[rgb]{ .886,  .89,  .514}0.0335 & \cellcolor[rgb]{ .996,  .914,  .514}-0.4968 & \cellcolor[rgb]{ 1,  .89,  .514}0.1681 & \cellcolor[rgb]{ .867,  .886,  .514}0.0468 & \cellcolor[rgb]{ .996,  .91,  .514}0.8925 & \cellcolor[rgb]{ .933,  .902,  .514} 0.00011  & \cellcolor[rgb]{ .918,  .898,  .51}0.0289 & \cellcolor[rgb]{ .953,  .91,  .518}1.1058 & 18.3  & 86.9 \\
    \hline
    \textbf{k-OWA-CTail-100} & \cellcolor[rgb]{ .996,  .898,  .51}0.0313\% & \cellcolor[rgb]{ .969,  .91,  .514}0.9531\% & \cellcolor[rgb]{ .957,  .91,  .518}0.0329 & \cellcolor[rgb]{ .992,  .922,  .518}-0.4952 & \cellcolor[rgb]{ .996,  .784,  .494}0.1740 & \cellcolor[rgb]{ .937,  .906,  .518}0.0458 & \cellcolor[rgb]{ .992,  .8,  .494}0.8884 & \cellcolor[rgb]{ 1,  .922,  .518} 0.00011  & \cellcolor[rgb]{ .945,  .906,  .514}0.0291 & \cellcolor[rgb]{ .996,  .894,  .51}1.1039 & 18.9  & 170.2 \\
    \hline
    \textbf{MinV} & \cellcolor[rgb]{ .973,  .412,  .42}0.0279\% & \cellcolor[rgb]{ .388,  .745,  .482}0.8816\% & \cellcolor[rgb]{ .992,  .847,  .502}0.0317 & \cellcolor[rgb]{ .388,  .745,  .482}-0.4578 & \cellcolor[rgb]{ .612,  .808,  .494}0.1412 & \cellcolor[rgb]{ .996,  .851,  .502}0.0440 & \cellcolor[rgb]{ .992,  .831,  .498}0.8897 & \cellcolor[rgb]{ .973,  .412,  .42} 0.00008  & \cellcolor[rgb]{ .388,  .745,  .482}0.0253 & \cellcolor[rgb]{ .949,  .91,  .518}1.1060 & 23.1  & 0.1 \\
    \hline
    \textbf{KP2011Min} & \cellcolor[rgb]{ .863,  .882,  .51}0.0326\% & \cellcolor[rgb]{ 1,  .902,  .514}0.9715\% & \cellcolor[rgb]{ .882,  .89,  .514}0.0336 & \cellcolor[rgb]{ .988,  .729,  .478}-0.5227 & \cellcolor[rgb]{ .992,  .773,  .49}0.1748 & \cellcolor[rgb]{ .875,  .886,  .514}0.0467 & \cellcolor[rgb]{ .988,  .765,  .486}0.8873 & \cellcolor[rgb]{ .906,  .894,  .514} 0.00012  & \cellcolor[rgb]{ .996,  .851,  .506}0.0309 & \cellcolor[rgb]{ .835,  .875,  .51}1.1090 & 17.0  & 69.1 \\
    \hline
    \end{tabular}%
    }
    }
  \label{tab:Perf_out_SP500_DL125_DH20}%
\end{table}%
%
%%%%

\noindent
Finally, Tables \ref{tab:ROI750_out_FTSE100_DL125_DH20_new}, \ref{tab:ROI750_out_NASDAQ100_DL125_DH20_new}, \ref{tab:ROI750_out_SP500_DL125_DH20_new} report some statistics of the 3-years time horizon ROI obtained by all the analyzed portfolio strategies on the FTSE100, NASDAQ100, and SP500 daily datasets, respectively.
Again, we can note that the most profitable portfolios are those obtained by the $k$-OWA-Tail-$\beta$ and $k$-OWA-CTail-$\beta$ models with low $\beta$, although in general, the SD models with Tails are less performing w.r.t. those formulated with CVaRs.
%
% Table generated by Excel2LaTeX from sheet 'ROI750 (2)'
\begin{table}[htbp]
  \centering
  \caption{Summary statistics of ROI based on a 3-years time horizon for FTSE100}
  \scalebox{0.74}{
  {\renewcommand{\arraystretch}{1.1}
    \begin{tabular}{|l|c|c|c|c|c|c|c|}
    \hline
    \textbf{Approach} & \multicolumn{1}{l|}{\textbf{ExpRet}} & \multicolumn{1}{l|}{\textbf{Vol}} & \multicolumn{1}{l|}{\textbf{5\%-perc}} & \multicolumn{1}{l|}{\textbf{25\%-perc}} & \multicolumn{1}{l|}{\textbf{50\%-perc}} & \multicolumn{1}{l|}{\textbf{75\%-perc}} & \multicolumn{1}{l|}{\textbf{95\%-perc}} \\
    \hline
    \textbf{Index} & \cellcolor[rgb]{ .973,  .412,  .42}8\% & \cellcolor[rgb]{ .388,  .745,  .482}15\% & \cellcolor[rgb]{ .973,  .412,  .42}-19\% & \cellcolor[rgb]{ .973,  .412,  .42}1\% & \cellcolor[rgb]{ .973,  .412,  .42}10\% & \cellcolor[rgb]{ .973,  .412,  .42}17\% & \cellcolor[rgb]{ .973,  .412,  .42}31\% \\
    \hline
    \textbf{Roman-Tail} & \cellcolor[rgb]{ .439,  .761,  .486}60\% & \cellcolor[rgb]{ .976,  .482,  .435}36\% & \cellcolor[rgb]{ .412,  .753,  .486}2\% & \cellcolor[rgb]{ .475,  .773,  .49}33\% & \cellcolor[rgb]{ .816,  .871,  .51}54\% & \cellcolor[rgb]{ .424,  .757,  .486}88\% & \cellcolor[rgb]{ .447,  .765,  .486}124\% \\
    \hline
    \textbf{k-OWA-Tail-5} & \cellcolor[rgb]{ .49,  .776,  .49}59\% & \cellcolor[rgb]{ .98,  .537,  .447}35\% & \cellcolor[rgb]{ .706,  .839,  .502}0\% & \cellcolor[rgb]{ .388,  .745,  .482}35\% & \cellcolor[rgb]{ .996,  .91,  .514}51\% & \cellcolor[rgb]{ .471,  .769,  .49}86\% & \cellcolor[rgb]{ .502,  .78,  .49}120\% \\
    \hline
    \textbf{k-OWA-Tail-10} & \cellcolor[rgb]{ .835,  .875,  .51}49\% & \cellcolor[rgb]{ .984,  .914,  .514}28\% & \cellcolor[rgb]{ .42,  .757,  .486}2\% & \cellcolor[rgb]{ .702,  .835,  .502}26\% & \cellcolor[rgb]{ .996,  .871,  .506}48\% & \cellcolor[rgb]{ .886,  .89,  .514}69\% & \cellcolor[rgb]{ .827,  .875,  .51}97\% \\
    \hline
    \textbf{k-OWA-Tail-25} & \cellcolor[rgb]{ .996,  .918,  .514}44\% & \cellcolor[rgb]{ 1,  .906,  .518}29\% & \cellcolor[rgb]{ .996,  .918,  .514}-3\% & \cellcolor[rgb]{ .996,  .902,  .514}16\% & \cellcolor[rgb]{ .992,  .922,  .518}52\% & \cellcolor[rgb]{ .996,  .914,  .514}64\% & \cellcolor[rgb]{ .996,  .914,  .514}83\% \\
    \hline
    \textbf{k-OWA-Tail-50} & \cellcolor[rgb]{ .996,  .89,  .51}42\% & \cellcolor[rgb]{ .933,  .902,  .514}27\% & \cellcolor[rgb]{ .992,  .843,  .502}-5\% & \cellcolor[rgb]{ .996,  .867,  .506}15\% & \cellcolor[rgb]{ 1,  .922,  .518}52\% & \cellcolor[rgb]{ .996,  .89,  .51}62\% & \cellcolor[rgb]{ .992,  .824,  .498}74\% \\
    \hline
    \textbf{k-OWA-Tail-75} & \cellcolor[rgb]{ .996,  .922,  .518}44\% & \cellcolor[rgb]{ .98,  .914,  .514}28\% & \cellcolor[rgb]{ .996,  .859,  .502}-5\% & \cellcolor[rgb]{ .996,  .898,  .51}16\% & \cellcolor[rgb]{ .757,  .855,  .506}54\% & \cellcolor[rgb]{ .996,  .918,  .514}65\% & \cellcolor[rgb]{ .996,  .867,  .506}79\% \\
    \hline
    \textbf{k-OWA-Tail-100} & \cellcolor[rgb]{ .894,  .894,  .514}47\% & \cellcolor[rgb]{ .996,  .839,  .502}30\% & \cellcolor[rgb]{ .992,  .831,  .498}-6\% & \cellcolor[rgb]{ .996,  .886,  .51}16\% & \cellcolor[rgb]{ .388,  .745,  .482}58\% & \cellcolor[rgb]{ .882,  .89,  .514}70\% & \cellcolor[rgb]{ .992,  .922,  .518}85\% \\
    \hline
    \textbf{k-OWA-CTail-5} & \cellcolor[rgb]{ .388,  .745,  .482}61\% & \cellcolor[rgb]{ .973,  .412,  .42}38\% & \cellcolor[rgb]{ .6,  .808,  .498}1\% & \cellcolor[rgb]{ .4,  .749,  .486}35\% & \cellcolor[rgb]{ .847,  .878,  .51}53\% & \cellcolor[rgb]{ .388,  .745,  .482}89\% & \cellcolor[rgb]{ .388,  .745,  .482}128\% \\
    \hline
    \textbf{k-OWA-CTail-10} & \cellcolor[rgb]{ .6,  .808,  .498}56\% & \cellcolor[rgb]{ .988,  .655,  .467}33\% & \cellcolor[rgb]{ .541,  .792,  .494}1\% & \cellcolor[rgb]{ .467,  .769,  .49}33\% & \cellcolor[rgb]{ .996,  .875,  .506}48\% & \cellcolor[rgb]{ .584,  .804,  .494}81\% & \cellcolor[rgb]{ .58,  .8,  .494}115\% \\
    \hline
    \textbf{k-OWA-CTail-25} & \cellcolor[rgb]{ .996,  .918,  .514}44\% & \cellcolor[rgb]{ 1,  .902,  .514}29\% & \cellcolor[rgb]{ .996,  .922,  .518}-3\% & \cellcolor[rgb]{ .984,  .918,  .518}17\% & \cellcolor[rgb]{ .996,  .863,  .506}47\% & \cellcolor[rgb]{ 1,  .922,  .518}65\% & \cellcolor[rgb]{ .945,  .906,  .518}88\% \\
    \hline
    \textbf{k-OWA-CTail-50} & \cellcolor[rgb]{ .996,  .894,  .51}42\% & \cellcolor[rgb]{ .957,  .91,  .514}27\% & \cellcolor[rgb]{ .996,  .89,  .51}-4\% & \cellcolor[rgb]{ .996,  .886,  .51}16\% & \cellcolor[rgb]{ .996,  .91,  .514}51\% & \cellcolor[rgb]{ .996,  .89,  .51}62\% & \cellcolor[rgb]{ .996,  .863,  .506}78\% \\
    \hline
    \textbf{k-OWA-CTail-75} & \cellcolor[rgb]{ .996,  .894,  .51}42\% & \cellcolor[rgb]{ .941,  .902,  .514}27\% & \cellcolor[rgb]{ .996,  .855,  .502}-5\% & \cellcolor[rgb]{ .996,  .863,  .506}15\% & \cellcolor[rgb]{ .996,  .918,  .514}52\% & \cellcolor[rgb]{ .996,  .89,  .51}62\% & \cellcolor[rgb]{ .992,  .835,  .498}75\% \\
    \hline
    \textbf{k-OWA-CTail-100} & \cellcolor[rgb]{ .996,  .91,  .514}44\% & \cellcolor[rgb]{ .965,  .91,  .514}28\% & \cellcolor[rgb]{ .996,  .855,  .502}-5\% & \cellcolor[rgb]{ .996,  .89,  .51}16\% & \cellcolor[rgb]{ .859,  .882,  .51}53\% & \cellcolor[rgb]{ .996,  .91,  .514}64\% & \cellcolor[rgb]{ .996,  .851,  .502}77\% \\
    \hline
    \textbf{MinV} & \cellcolor[rgb]{ .992,  .82,  .498}37\% & \cellcolor[rgb]{ .557,  .792,  .49}19\% & \cellcolor[rgb]{ .388,  .745,  .482}3\% & \cellcolor[rgb]{ .722,  .843,  .502}25\% & \cellcolor[rgb]{ .992,  .78,  .49}41\% & \cellcolor[rgb]{ .988,  .761,  .486}50\% & \cellcolor[rgb]{ .988,  .722,  .478}63\% \\
    \hline
    \textbf{KP2011Min} & \cellcolor[rgb]{ .541,  .792,  .494}57\% & \cellcolor[rgb]{ .98,  .557,  .447}35\% & \cellcolor[rgb]{ .612,  .812,  .498}1\% & \cellcolor[rgb]{ .631,  .816,  .498}28\% & \cellcolor[rgb]{ .675,  .831,  .502}55\% & \cellcolor[rgb]{ .478,  .773,  .49}85\% & \cellcolor[rgb]{ .588,  .804,  .494}114\% \\
    \hline
    \end{tabular}%
    }
    }
  \label{tab:ROI750_out_FTSE100_DL125_DH20_new}%
\end{table}%
%
% Table generated by Excel2LaTeX from sheet 'ROI750 (2)'
% Table generated by Excel2LaTeX from sheet 'ROI750 (3)'
\begin{table}[htbp]
  \centering
  \caption{Summary statistics of ROI based on a 3-years time horizon for NASDAQ100}
  \scalebox{0.74}{
  {\renewcommand{\arraystretch}{1.1}
    \begin{tabular}{|l|c|c|c|c|c|c|c|}
    \hline
    \textbf{Approach} & \textbf{ExpRet} & \textbf{Vol} & \textbf{5\%-perc} & \textbf{25\%-perc} & \textbf{50\%-perc} & \textbf{75\%-perc} & \textbf{95\%-perc} \\
    \hline
    \textbf{Index} & \cellcolor[rgb]{ .792,  .863,  .506}54\% & \cellcolor[rgb]{ .98,  .533,  .443}23\% & \cellcolor[rgb]{ .992,  .843,  .502}2\% & \cellcolor[rgb]{ .663,  .824,  .498}45\% & \cellcolor[rgb]{ .925,  .902,  .514}55\% & \cellcolor[rgb]{ .792,  .863,  .506}68\% & \cellcolor[rgb]{ .624,  .816,  .498}88\% \\
    \hline
    \textbf{Roman-Tail} & \cellcolor[rgb]{ .949,  .91,  .518}52\% & \cellcolor[rgb]{ .973,  .412,  .42}24\% & \cellcolor[rgb]{ .988,  .741,  .482}1\% & \cellcolor[rgb]{ .961,  .91,  .518}40\% & \cellcolor[rgb]{ .98,  .608,  .455}48\% & \cellcolor[rgb]{ .996,  .89,  .51}65\% & \cellcolor[rgb]{ .388,  .745,  .482}95\% \\
    \hline
    \textbf{k-OWA-Tail-5} & \cellcolor[rgb]{ .965,  .914,  .518}52\% & \cellcolor[rgb]{ 1,  .886,  .514}21\% & \cellcolor[rgb]{ .984,  .69,  .471}0\% & \cellcolor[rgb]{ .78,  .859,  .506}43\% & \cellcolor[rgb]{ .996,  .878,  .51}54\% & \cellcolor[rgb]{ .996,  .875,  .506}65\% & \cellcolor[rgb]{ .8,  .867,  .51}84\% \\
    \hline
    \textbf{k-OWA-Tail-10} & \cellcolor[rgb]{ .388,  .745,  .482}59\% & \cellcolor[rgb]{ .937,  .902,  .514}21\% & \cellcolor[rgb]{ .388,  .745,  .482}10\% & \cellcolor[rgb]{ .388,  .745,  .482}49\% & \cellcolor[rgb]{ .388,  .745,  .482}61\% & \cellcolor[rgb]{ .388,  .745,  .482}72\% & \cellcolor[rgb]{ .549,  .792,  .494}91\% \\
    \hline
    \textbf{k-OWA-Tail-25} & \cellcolor[rgb]{ .992,  .922,  .518}51\% & \cellcolor[rgb]{ .471,  .769,  .486}19\% & \cellcolor[rgb]{ .831,  .875,  .51}5\% & \cellcolor[rgb]{ .843,  .878,  .51}42\% & \cellcolor[rgb]{ .996,  .894,  .51}54\% & \cellcolor[rgb]{ .992,  .8,  .494}63\% & \cellcolor[rgb]{ .949,  .91,  .518}79\% \\
    \hline
    \textbf{k-OWA-Tail-50} & \cellcolor[rgb]{ .996,  .91,  .514}51\% & \cellcolor[rgb]{ .988,  .918,  .514}21\% & \cellcolor[rgb]{ .953,  .91,  .518}3\% & \cellcolor[rgb]{ .992,  .804,  .494}37\% & \cellcolor[rgb]{ .69,  .835,  .502}58\% & \cellcolor[rgb]{ .965,  .914,  .518}66\% & \cellcolor[rgb]{ .992,  .78,  .49}76\% \\
    \hline
    \textbf{k-OWA-Tail-75} & \cellcolor[rgb]{ .996,  .875,  .506}50\% & \cellcolor[rgb]{ 1,  .922,  .518}21\% & \cellcolor[rgb]{ .98,  .918,  .518}3\% & \cellcolor[rgb]{ .992,  .792,  .49}37\% & \cellcolor[rgb]{ .863,  .882,  .51}56\% & \cellcolor[rgb]{ .965,  .914,  .518}66\% & \cellcolor[rgb]{ .988,  .725,  .478}75\% \\
    \hline
    \textbf{k-OWA-Tail-100} & \cellcolor[rgb]{ .996,  .851,  .502}50\% & \cellcolor[rgb]{ 1,  .855,  .506}21\% & \cellcolor[rgb]{ .992,  .922,  .518}3\% & \cellcolor[rgb]{ .992,  .808,  .494}37\% & \cellcolor[rgb]{ .996,  .894,  .51}54\% & \cellcolor[rgb]{ .91,  .898,  .514}67\% & \cellcolor[rgb]{ .992,  .831,  .498}77\% \\
    \hline
    \textbf{k-OWA-CTail-5} & \cellcolor[rgb]{ .702,  .839,  .502}55\% & \cellcolor[rgb]{ .98,  .529,  .443}24\% & \cellcolor[rgb]{ .996,  .902,  .514}2\% & \cellcolor[rgb]{ .718,  .843,  .502}44\% & \cellcolor[rgb]{ .996,  .871,  .506}54\% & \cellcolor[rgb]{ .714,  .839,  .502}69\% & \cellcolor[rgb]{ .392,  .749,  .486}95\% \\
    \hline
    \textbf{k-OWA-CTail-10} & \cellcolor[rgb]{ .663,  .824,  .498}55\% & \cellcolor[rgb]{ .992,  .761,  .49}22\% & \cellcolor[rgb]{ .894,  .894,  .514}4\% & \cellcolor[rgb]{ .608,  .812,  .498}46\% & \cellcolor[rgb]{ .796,  .863,  .506}57\% & \cellcolor[rgb]{ .816,  .871,  .51}68\% & \cellcolor[rgb]{ .561,  .796,  .494}90\% \\
    \hline
    \textbf{k-OWA-CTail-25} & \cellcolor[rgb]{ .984,  .918,  .518}51\% & \cellcolor[rgb]{ .655,  .82,  .494}20\% & \cellcolor[rgb]{ .996,  .906,  .514}2\% & \cellcolor[rgb]{ .792,  .863,  .506}43\% & \cellcolor[rgb]{ .996,  .894,  .51}54\% & \cellcolor[rgb]{ .992,  .796,  .49}63\% & \cellcolor[rgb]{ .945,  .906,  .518}80\% \\
    \hline
    \textbf{k-OWA-CTail-50} & \cellcolor[rgb]{ .992,  .843,  .502}50\% & \cellcolor[rgb]{ .784,  .859,  .502}20\% & \cellcolor[rgb]{ .992,  .812,  .494}1\% & \cellcolor[rgb]{ .996,  .882,  .51}39\% & \cellcolor[rgb]{ .949,  .91,  .518}55\% & \cellcolor[rgb]{ .992,  .812,  .494}64\% & \cellcolor[rgb]{ .988,  .753,  .482}75\% \\
    \hline
    \textbf{k-OWA-CTail-75} & \cellcolor[rgb]{ .996,  .914,  .514}51\% & \cellcolor[rgb]{ .765,  .851,  .502}20\% & \cellcolor[rgb]{ .902,  .894,  .514}4\% & \cellcolor[rgb]{ .996,  .882,  .51}39\% & \cellcolor[rgb]{ .753,  .851,  .506}57\% & \cellcolor[rgb]{ .996,  .906,  .514}66\% & \cellcolor[rgb]{ .988,  .745,  .482}75\% \\
    \hline
    \textbf{k-OWA-CTail-100} & \cellcolor[rgb]{ .996,  .902,  .514}51\% & \cellcolor[rgb]{ .902,  .89,  .51}21\% & \cellcolor[rgb]{ .949,  .91,  .518}3\% & \cellcolor[rgb]{ .996,  .847,  .502}38\% & \cellcolor[rgb]{ .776,  .859,  .506}57\% & \cellcolor[rgb]{ .969,  .914,  .518}66\% & \cellcolor[rgb]{ .988,  .729,  .478}75\% \\
    \hline
    \textbf{MinV} & \cellcolor[rgb]{ .973,  .412,  .42}41\% & \cellcolor[rgb]{ .388,  .745,  .482}19\% & \cellcolor[rgb]{ .984,  .698,  .475}0\% & \cellcolor[rgb]{ .973,  .447,  .424}30\% & \cellcolor[rgb]{ .973,  .412,  .42}44\% & \cellcolor[rgb]{ .973,  .412,  .42}55\% & \cellcolor[rgb]{ .973,  .412,  .42}70\% \\
    \hline
    \textbf{KP2011Min} & \cellcolor[rgb]{ .976,  .486,  .431}43\% & \cellcolor[rgb]{ .996,  .831,  .502}21\% & \cellcolor[rgb]{ .973,  .412,  .42}-3\% & \cellcolor[rgb]{ .973,  .412,  .42}29\% & \cellcolor[rgb]{ .976,  .533,  .443}46\% & \cellcolor[rgb]{ .984,  .643,  .463}60\% & \cellcolor[rgb]{ .976,  .514,  .439}71\% \\
    \hline
    \end{tabular}%
    }
    }
  \label{tab:ROI750_out_NASDAQ100_DL125_DH20_new}%
\end{table}%
%
%
% Table generated by Excel2LaTeX from sheet 'ROI750 (3)'
\begin{table}[htbp]
  \centering
  \caption{Summary statistics of ROI based on a 3-years time horizon for SP500}
  \scalebox{0.74}{
  {\renewcommand{\arraystretch}{1.1}
    \begin{tabular}{|l|c|c|c|c|c|c|c|}
    \hline
    \textbf{Approach} & \textbf{ExpRet} & \textbf{Vol} & \textbf{5\%-perc} & \textbf{25\%-perc} & \textbf{50\%-perc} & \textbf{75\%-perc} & \textbf{95\%-perc} \\
    \hline
    \textbf{Index} & \cellcolor[rgb]{ .761,  .855,  .506}30\% & \cellcolor[rgb]{ .996,  .827,  .502}22\% & \cellcolor[rgb]{ .992,  .804,  .494}-22\% & \cellcolor[rgb]{ .996,  .863,  .506}24\% & \cellcolor[rgb]{ .992,  .824,  .498}32\% & \cellcolor[rgb]{ .996,  .878,  .506}42\% & \cellcolor[rgb]{ .388,  .745,  .482}60\% \\
    \hline
    \textbf{Roman-Tail} & \cellcolor[rgb]{ .992,  .8,  .494}28\% & \cellcolor[rgb]{ .973,  .412,  .42}24\% & \cellcolor[rgb]{ .973,  .455,  .427}-26\% & \cellcolor[rgb]{ .988,  .741,  .482}22\% & \cellcolor[rgb]{ .894,  .894,  .514}34\% & \cellcolor[rgb]{ .388,  .745,  .482}44\% & \cellcolor[rgb]{ .706,  .839,  .502}56\% \\
    \hline
    \textbf{k-OWA-Tail-5} & \cellcolor[rgb]{ .996,  .922,  .518}29\% & \cellcolor[rgb]{ .929,  .898,  .51}21\% & \cellcolor[rgb]{ .831,  .875,  .51}-17\% & \cellcolor[rgb]{ .992,  .792,  .49}23\% & \cellcolor[rgb]{ 1,  .922,  .518}34\% & \cellcolor[rgb]{ .996,  .886,  .51}42\% & \cellcolor[rgb]{ .941,  .906,  .518}53\% \\
    \hline
    \textbf{k-OWA-Tail-10} & \cellcolor[rgb]{ .847,  .878,  .51}29\% & \cellcolor[rgb]{ .631,  .812,  .494}19\% & \cellcolor[rgb]{ .663,  .827,  .502}-13\% & \cellcolor[rgb]{ .859,  .882,  .51}25\% & \cellcolor[rgb]{ .996,  .851,  .502}33\% & \cellcolor[rgb]{ .984,  .639,  .463}40\% & \cellcolor[rgb]{ .761,  .855,  .506}55\% \\
    \hline
    \textbf{k-OWA-Tail-25} & \cellcolor[rgb]{ .388,  .745,  .482}32\% & \cellcolor[rgb]{ .416,  .753,  .482}18\% & \cellcolor[rgb]{ .388,  .745,  .482}-7\% & \cellcolor[rgb]{ .439,  .761,  .486}28\% & \cellcolor[rgb]{ .627,  .816,  .498}36\% & \cellcolor[rgb]{ .996,  .894,  .51}42\% & \cellcolor[rgb]{ .714,  .839,  .502}56\% \\
    \hline
    \textbf{k-OWA-Tail-50} & \cellcolor[rgb]{ .859,  .882,  .51}29\% & \cellcolor[rgb]{ 1,  .863,  .506}21\% & \cellcolor[rgb]{ .988,  .773,  .486}-22\% & \cellcolor[rgb]{ .675,  .827,  .502}27\% & \cellcolor[rgb]{ .863,  .882,  .51}34\% & \cellcolor[rgb]{ .502,  .78,  .49}44\% & \cellcolor[rgb]{ .988,  .773,  .486}51\% \\
    \hline
    \textbf{k-OWA-Tail-75} & \cellcolor[rgb]{ .984,  .655,  .467}27\% & \cellcolor[rgb]{ .996,  .824,  .502}22\% & \cellcolor[rgb]{ .98,  .6,  .455}-24\% & \cellcolor[rgb]{ .992,  .827,  .498}23\% & \cellcolor[rgb]{ .992,  .827,  .498}32\% & \cellcolor[rgb]{ .996,  .91,  .514}42\% & \cellcolor[rgb]{ .973,  .412,  .42}50\% \\
    \hline
    \textbf{k-OWA-Tail-100} & \cellcolor[rgb]{ .988,  .761,  .486}28\% & \cellcolor[rgb]{ .988,  .671,  .471}23\% & \cellcolor[rgb]{ .973,  .412,  .42}-26\% & \cellcolor[rgb]{ .996,  .918,  .514}25\% & \cellcolor[rgb]{ .996,  .914,  .514}33\% & \cellcolor[rgb]{ .635,  .82,  .498}43\% & \cellcolor[rgb]{ .984,  .69,  .471}51\% \\
    \hline
    \textbf{k-OWA-CTail-5} & \cellcolor[rgb]{ .98,  .612,  .455}27\% & \cellcolor[rgb]{ 1,  .914,  .518}21\% & \cellcolor[rgb]{ .957,  .91,  .518}-20\% & \cellcolor[rgb]{ .984,  .651,  .463}21\% & \cellcolor[rgb]{ .996,  .918,  .514}34\% & \cellcolor[rgb]{ .988,  .722,  .478}41\% & \cellcolor[rgb]{ .988,  .745,  .482}51\% \\
    \hline
    \textbf{k-OWA-CTail-10} & \cellcolor[rgb]{ .992,  .808,  .494}28\% & \cellcolor[rgb]{ .902,  .89,  .51}20\% & \cellcolor[rgb]{ .941,  .906,  .518}-19\% & \cellcolor[rgb]{ 1,  .922,  .518}25\% & \cellcolor[rgb]{ .996,  .882,  .51}33\% & \cellcolor[rgb]{ .988,  .702,  .475}41\% & \cellcolor[rgb]{ 1,  .922,  .518}52\% \\
    \hline
    \textbf{k-OWA-CTail-25} & \cellcolor[rgb]{ .996,  .898,  .51}28\% & \cellcolor[rgb]{ .388,  .745,  .482}17\% & \cellcolor[rgb]{ .518,  .784,  .49}-10\% & \cellcolor[rgb]{ .996,  .882,  .51}24\% & \cellcolor[rgb]{ .988,  .757,  .482}31\% & \cellcolor[rgb]{ .973,  .412,  .42}39\% & \cellcolor[rgb]{ .871,  .886,  .514}54\% \\
    \hline
    \textbf{k-OWA-CTail-50} & \cellcolor[rgb]{ .635,  .816,  .498}31\% & \cellcolor[rgb]{ .8,  .863,  .506}20\% & \cellcolor[rgb]{ .851,  .878,  .51}-17\% & \cellcolor[rgb]{ .486,  .773,  .49}28\% & \cellcolor[rgb]{ .545,  .792,  .494}37\% & \cellcolor[rgb]{ .976,  .918,  .518}42\% & \cellcolor[rgb]{ .992,  .839,  .502}52\% \\
    \hline
    \textbf{k-OWA-CTail-75} & \cellcolor[rgb]{ .929,  .902,  .514}29\% & \cellcolor[rgb]{ .988,  .918,  .514}21\% & \cellcolor[rgb]{ .992,  .835,  .498}-22\% & \cellcolor[rgb]{ .773,  .859,  .506}26\% & \cellcolor[rgb]{ .867,  .886,  .514}34\% & \cellcolor[rgb]{ .757,  .851,  .506}43\% & \cellcolor[rgb]{ .973,  .471,  .431}50\% \\
    \hline
    \textbf{k-OWA-CTail-100} & \cellcolor[rgb]{ .996,  .914,  .514}29\% & \cellcolor[rgb]{ .996,  .835,  .502}21\% & \cellcolor[rgb]{ .984,  .667,  .467}-24\% & \cellcolor[rgb]{ .918,  .898,  .514}25\% & \cellcolor[rgb]{ .925,  .902,  .514}34\% & \cellcolor[rgb]{ .655,  .824,  .498}43\% & \cellcolor[rgb]{ .973,  .451,  .427}50\% \\
    \hline
    \textbf{MinV} & \cellcolor[rgb]{ .973,  .412,  .42}26\% & \cellcolor[rgb]{ .898,  .89,  .51}20\% & \cellcolor[rgb]{ .808,  .867,  .51}-16\% & \cellcolor[rgb]{ .973,  .412,  .42}18\% & \cellcolor[rgb]{ .973,  .412,  .42}27\% & \cellcolor[rgb]{ .635,  .816,  .498}43\% & \cellcolor[rgb]{ .961,  .914,  .518}53\% \\
    \hline
    \textbf{KP2011Min} & \cellcolor[rgb]{ .631,  .816,  .498}31\% & \cellcolor[rgb]{ .992,  .776,  .49}22\% & \cellcolor[rgb]{ .988,  .769,  .486}-22\% & \cellcolor[rgb]{ .388,  .745,  .482}28\% & \cellcolor[rgb]{ .388,  .745,  .482}38\% & \cellcolor[rgb]{ .443,  .761,  .486}44\% & \cellcolor[rgb]{ .996,  .914,  .514}52\% \\
    \hline
    \end{tabular}%
    }
    }
  \label{tab:ROI750_out_SP500_DL125_DH20_new}%
\end{table}%

\end{document}